\documentclass[journal,twocolumn,10pt]{IEEEtran}
\usepackage{amsmath,epsfig}
\usepackage{amssymb}
\usepackage{booktabs}
\usepackage{amsthm}
\usepackage{graphicx}
\usepackage{caption}
\usepackage{subcaption}
\usepackage{rotating}
\usepackage{floatflt} 
\usepackage{paralist} 
\usepackage{algorithm} 
\usepackage{color}
\usepackage{xcolor}
\usepackage{epstopdf}
\usepackage[normalem]{ulem}
\usepackage{float}
\usepackage{todonotes}
\usepackage{mathtools}

\newcommand{\mypar}[1]{{\bf #1.}}
\theoremstyle{definition}

\newtheorem{myLem}{Lemma}

\newtheorem{myThm}{Theorem}
\newtheorem{myCorollary}{Corollary}
\newcommand{\R}{\ensuremath{\mathbb{R}}}

\DeclareMathOperator{\Id}{I}

\DeclareMathOperator{\KL}{KL}

\def\x{\mathbf{x}}

\def\y{\mathbf{y}}

\def\e{\mathbf{e}}
\def\w{\mathbf{w}}
\def\s{\mathbf{s}}

\def\N{\mathcal{N}}

\def\V{\mathcal{V}}

\def\E{\mathcal{E}}
\def\Ep{\mathbb{E}}
\def\one{{\bf 1}}

\DeclareMathOperator{\TV}{TV}
\DeclareMathOperator{\Adj}{A}

\DeclareMathOperator{\W}{W}

\usepackage{xr}
\begin{document}
\title{ \huge{Detecting Localized Categorical Attributes on Graphs} }
\author{Siheng~Chen,~\IEEEmembership{Student~Member,~IEEE},
  Yaoqing~Yang,~\IEEEmembership{Student~Member,~IEEE}, Shi Zong, Aarti
  Singh,
  Jelena~Kova\v{c}evi\'c,~\IEEEmembership{Fellow,~IEEE}
  \thanks{S. Chen (sihengc@andrew.cmu.edu), Y. Yang
    (yaoqingy@andrew.cmu.edu) and S. Zong (szong@andrew.cmu.edu) are
    with the Dept. of Electrical and Computer Engineering, A. Singh
    (aarti@cs.cmu.edu) is with the Dept. of Machine Learning and
    J. Kova\v{c}evi\'c (jelenak@cmu.edu) is with the Depts. of
    Electrical and Computer Engineering and Biomedical Engineering,
    Carnegie Mellon University, Pittsburgh, PA.
  }
  \thanks{The authors gratefully acknowledge support from the NSF
    through awards 1130616 and 1421919 and the University
    Transportation Center grant DTRT12-GUTC11 from the US Dept. of
    Transportation.}
}
 \maketitle


\begin{abstract}
  Do users from Carnegie Mellon University form social communities on
  Facebook? Do signal processing researchers from tightly collaborate
  with each other? Do Chinese restaurants in Manhattan cluster
  together? These seemingly different problems share a common
  structure: an attribute that may be localized on a graph. In other
  words, nodes activated by an attribute form a subgraph that can be
  easily separated from other nodes. In this paper, we thus focus on
  the task of detecting localized attributes on a graph. We are
  particularly interested in categorical attributes such as attributes
  in online social networks, ratings in recommender systems and
  viruses in cyber-physical systems because they are widely used in
  numerous data mining applications. To solve the task, we formulate
   a statistical hypothesis testing problem to decide whether a
  given attribute is localized or not. We propose two statistics:
  graph wavelet statistic and graph scan statistic, both of which are
  provably effective in detecting localized attributes.  We validate
  the robustness of the proposed statistics on both simulated data and
  two real-world applications: high air-pollution detection and
  keyword ranking in a co-authorship network collected from IEEE
  Xplore. Experimental results show that the proposed graph wavelet
  statistic and graph scan statistic are effective and efficient.
\end{abstract}

\section{Introduction}
\label{sec:intro}
Massive amounts of data being generated from various sources including
social networks, citation, biological, and physical infrastructure
have spurred the emerging area of analyzing data supported on
graphs~\cite{Jackson:08, Newman:10} giving rise to a variety of
scientific and engineering studies; for example, selecting
representative training data to improve semi-supervised learning with
graphs~\cite{ChenVSK:15c}; detecting communities in communication or
social networks~\cite{GirvanN:02}; ranking the most important websites
on the Internet~\cite{BrinP:1998}; and detecting anomalies in sensor
networks~\cite{XieHTP:2011}.

\begin{figure}[htb]
  \begin{center}
    \begin{tabular}{cc}
      \includegraphics[width=0.4\columnwidth]{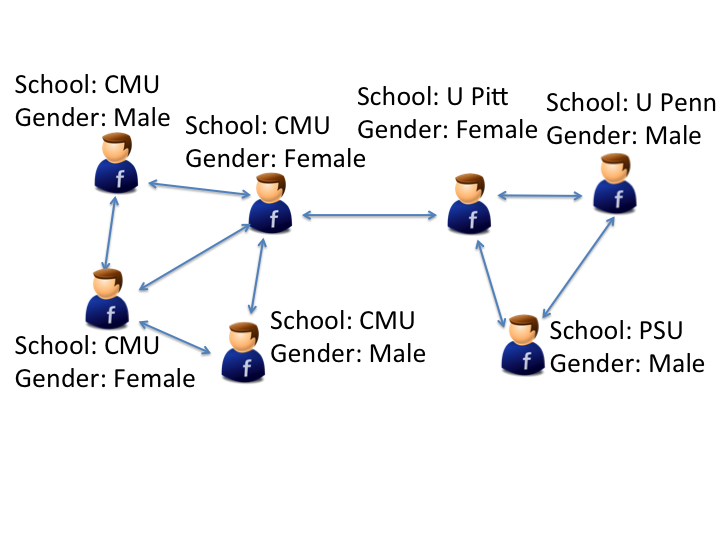} 
      &
      \includegraphics[width=0.4\columnwidth]{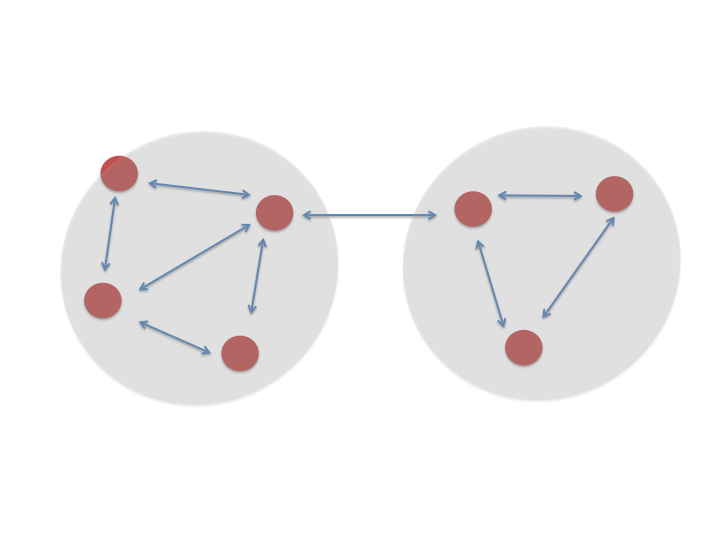} 
      \\
      {\small (a) Graph with two attributes.} &  {\small (b) Graph with two
    communities.} 
      \\
      \includegraphics[width=0.4\columnwidth]{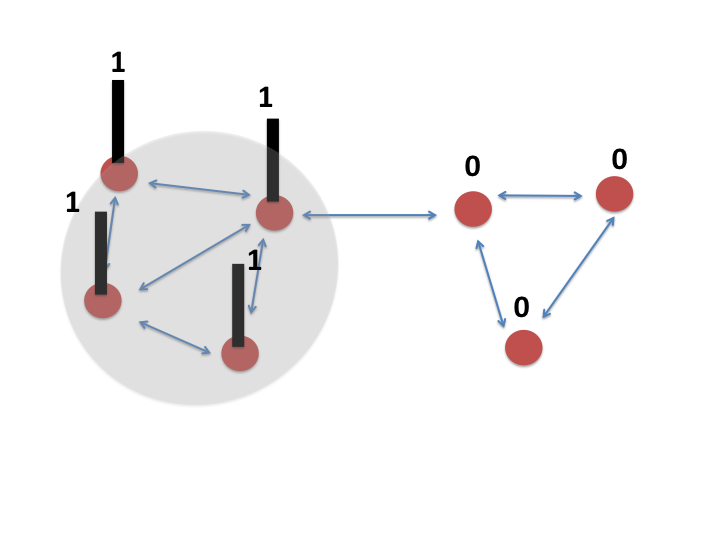}  &
      \includegraphics[width=0.4\columnwidth]{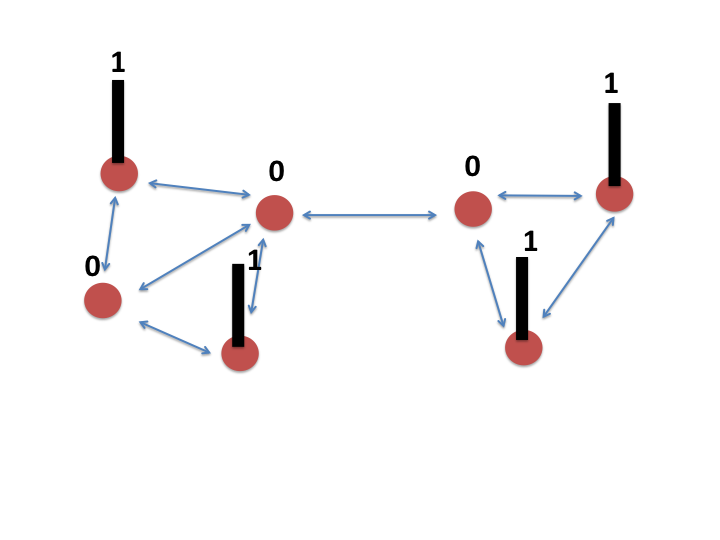} 
      \\
      {\small (c)  Attribute \emph{from CMU?}}  &  {\small (c)   Attribute \emph{male?}} 
    \end{tabular}
  \end{center}
  \caption{\label{fig:intuition} Detecting localized categorical
    attributes. (a) Graph with two attributes. (b) Graph with two
    communities. (c) Attribute \emph{Is this user from CMU?} is
    localized. (d) Attribute \emph{Is this user male?} is not
    localized. The goal of community detection is to identify
    subgraphs as in (b), while the goal of the localized attribute
    detection is to identify whether an attribute is localized
    (yes in (c), no in (d)). }
\end{figure}

Graph signal processing is a theoretical framework for the analysis of
high-dimensional data with complex, irregular
structure~\cite{ShumanNFOV:13,SandryhailaM:13}. It extends classical
signal processing concepts such as signals, filters, Fourier
transform, frequency response, low- and highpass filtering, from
signals residing on regular lattices to data residing on general
graphs; for example, a graph signal models the data value assigned to
each node in a graph. Recent work involves sampling for graph
signals~\cite{ChenVSK:15,AnisGO:15,WangLG:14,MarquesSGR:15}, recovery
for graph signals~\cite{NarangGO:13, ChenSMK:14,RomeroMG:16,
  KotzagiannidisD:16}, representations for graph
signals~\cite{ZhuM:12, ChenVSK:16}, uncertainty principles on
graphs~\cite{AgaskarL:13, TsitsveroBL:15}, stationary graph signal
processing~\cite{PerraudinV:16, MarquesSGR:16}, graph dictionary
construction~\cite{ThanouSF:14}, graph-based filter
banks~\cite{NarangO:12,NarangO:13, EkambaramFAR:15, ZengCO:16},
denoising on graphs~\cite{NarangO:12, ChenSMK:14a}, community
detection and clustering on graphs~\cite{Tremblay:14, DongFVN:14,
  ChenO:14}, distributed computing~\cite{KarM:09, DuMWP:14} and graph-based
transforms~\cite{HammondVG:11,NarangSO:10,ShumanFV:16}.

We here consider detecting~\emph{localized categorical attributes} on
graphs.  A categorical attribute is defined as a variable that can be
put into a countable number of categories. It can be represented by
several binary attributes and is widely used in data and graph mining
applications~\cite{AkogluTK:2015}.  We model categorical attributes by
binary graph signals\footnote{Attributes and binary graph
    signals are the same in this context.}: when a signal coefficient
is one, the corresponding node is activated by the attribute, and vice
versa.  A localized categorical attribute, or a localized pattern,
  is defined as an attribute whose activated nodes form a subgraph
  that can be easily separated from the rest of nodes; in other words,
  the cut cost is small. In practice, detecting a localized attribute
is nontrivial because an observed attribute is often corrupted by
noise. The goal of localized attribute detection is to identify
  localized attributes hidden in noisy attributes using graph
topology.  This task is relevant to many real-world applications such
as identifying localized attributes in online social networks,
activity in the brain connectivity networks and viruses in
cyber-physical systems.

This localized attribute detection task is related to, yet different
from conventional community detection in network
science~\cite{AhnBL:10, GirvanN:02, AbbeBH:16}. The goal of community
detection is to identify modules and hierarchical organization by
using the information encoded in the \emph{graph topology
  only}~\cite{Fortunato:10}. A module is typically considered to be a
node set with dense internal and sparse external
  connections. The difference between the two is that community
detection considers detecting patterns in graph topology while
localized attribute detection considers detecting patterns in an
attribute (binary graph signal).  For example, suppose that we want to
identify whether users from Carnegie Mellon University (CMU) form a
localized attribute on Facebook. The binary answer to \emph{Is this
  user from CMU?} is an attribute on a graph, which activates a
subgraph with few external connections. These activated users thus
form a localized attribute; see Figure~\ref{fig:intuition}.

To describe the localization level of an attribute, we consider
external connections of the subgraph activated by the attribute,
leading to scalable detection algorithms.  We define the localization
level of a localized attribute as the difficulty of separating the
corresponding subgraph from the rest of the nodes, which is quantified
by the total variation on graphs.  We then formulate a hypothesis
testing problem to decide whether a categorical attribute is localized
and propose two statistics: graph wavelet statistic and graph scan
statistic. Similarly to detecting transient changes in time-series
signals by using wavelet techniques, we design a graph wavelet
statistic based on a Haar-like graph wavelet basis. Since the graph
wavelet basis is preconstructed, the computational cost is linear with
the number of nodes. We also formulate a generalized likelihood test
and propose a graph scan statistic, which can be efficiently solved by
a standard graph-cut algorithm. The intuition behind the proposed
statistics is to find the underlying localized attribute in a graph,
which is equivalent to denoising the given attribute based on the
graph structure, and then calculating the statistic values based on
the denoised attribute. We demonstrate the effectiveness and
robustness of the proposed methods through validation on simulated and
real data.

We here consider only one attribute per test; when given multiple
attributes, we can rank their localization levels according to the
proposed statistics. Note that the proposed statistics does not
address correlation among attributes.

\mypar{Contributions} The main contributions of the paper are:
\begin{itemize}
\item A novel hypothesis testing framework for \emph{detecting localized
  attributes} corrupted by Bernoulli noise;
\item A novel, effective, and scalable \emph{graph wavelet statistic} for
  detecting localized attributes with analysis of detection error;
\item A novel, effective, and scalable \emph{graph scan statistic} for
  detecting localized attributes with analysis of detection error; and
\item Validation on both \emph{simulated and real datasets} with
  applications to detection of high air pollution and ranking keywords
  in a co-authorship network.
\end{itemize}

\mypar{Paper outline} Section~\ref{sec:formulation} formulates
the problem; Section~\ref{sec:relatedwork} reviews related work;
Section~\ref{sec:methodology} proposes graph wavelet and scan
statistics; Section~\ref{sec:experiments} validates the proposed
methods on simulated and real data; and Section~\ref{sec:conclusions}
concludes the paper.

\section{Problem Formulation}
\label{sec:formulation}

We consider a weighted, undirected graph $G = (\V, \E, \Adj)$, with
$\V = \{v_1,\ldots, v_{N}\}$ the set of nodes, $\E = \{e_1,\ldots,
e_{M}\}$ the set of edges and $\Adj \in \R^{N \times N}$ a weighted
adjacency matrix.  A~\emph{graph signal} is defined as the map that
assigns the signal coefficient $s_n \in \R$ to the graph node $v_n$;
it can be written as a vector $ \s \ = \ \begin{bmatrix} s_1 & \ldots
  & s_{N}\end{bmatrix}^T \in \R^N$.  The edge weight $\Adj_{i,j}$
between nodes $v_i$ and $v_j$ quantifies the underlying relation
between the $i$th and the $j$th signal coefficients, such as a
similarity, a dependency, or a communication pattern. In this paper,
all graph signals are binary ($\s \in \{0, 1\}^N$) and represent attributes; we thus use the
word \emph{attribute} instead of \emph{binary graph signal} in what
follows.

Let $\Delta \in \R^{|\E| \times |\V|}$ be the~\emph{graph incidence
  matrix}, whose rows correspond to
edges~\cite{SharpnackRS:12,WangSST:15}; for example, if $e_i$ is the
edge that connects the $j$th node to the $k$th node ($j < k$), the
elements of the $i$th row of $\Delta$ are
\begin{equation*}
  \label{eq:Delta}
  \Delta_{i, \ell} = 
  \left\{ 
    \begin{array}{rl}
      \Adj_{j,k}, & \ell = j;\\
      -\Adj_{j,k}, & \ell = k;\\
      0, & \mbox{otherwise}.
    \end{array} \right.
\end{equation*}

An~\emph{activated} node set $C \subseteq \V$ is denoted by its
indicator function (attribute) ${\bf 1}_{C} \in \{0, 1\}^N$,
\begin{equation*}
\left(  {\bf 1}_{C}  \right)_i = 
  \left\{ 
    \begin{array}{rl}
      1, & v_i \in C;\\
      0, & \mbox{otherwise}.
  \end{array} \right.
\end{equation*}
When $C$ forms a connected subgraph, we call $C$ a~\emph{local set}.

We consider the localization level of the attribute $\one_C$ as the
difficulty of separating $C$ from $\overline{C} = \V \backslash C$,
and use the $\ell_p$-norm-based total variation to quantify it,
\begin{equation}
  \label{eq:TV1}
  \TV_p ( \one_C ) = \left\| \Delta \one_C \right\|_p.
\end{equation} 
While $\TV_0 ( {\bf 1}_{C} )$ counts the number of edges connecting
$C$ and $\overline{C}$, $\TV_1 ( {\bf 1}_{C} )$ takes edge weights
into account; when edges are unweighted, the two are the same. Total
variation builds a connection between an attribute and graph structure
and measures the localization level of an attribute on a specific
graph; that is, an attribute with smaller total variation is more
localized on a graph because it is easier to separate its activated
part $C$ from its nonactivated part $\overline{C}$.

The task of localized attribute detection is made harder when noise is
present.  Given a noisy attribute $\y \in \{0, 1\}^N$, the general
statistical testing formulation is:
\begin{eqnarray}
  \label{eq:testing}
  H_0^N & : &    \y \sim f( 0 , \epsilon ),
  \\
  \nonumber
  H_1^N & : &    \y \sim f( \s, \epsilon) \text{ with }  \s \in  \mathcal{S}_N,
\end{eqnarray}
where $N$ indicates that the observation is $N$-dimensional, with $N$
the number of nodes in the graph, and is independent in each
dimension, $\mathcal{S}_N$ is a predefined class of localized
attributes, $\epsilon > 0$ is the noise level and the link function
$f(\cdot, \cdot)$ specifies the noise model. For example, if a signal
is corrupted by Gaussian noise, $f( \s, \epsilon) = \s + \e$, where
$\e \sim \N(0, \epsilon \Id )$.

The null (default) hypothesis thus represents no particular
localization for the attribute and the alternative hypothesis
represents a localized attribute. The two key factors in
\eqref{eq:testing} are the noise model and the localized attribute, and
we can make independent assumptions on these two. For example, noise
can follow Gaussian or Bernoulli distribution and the localization
level can be described by small cut costs or cliques~\cite{Newman:10}.

Let the test be a mapping $T(\y) = \{0, 1\}$, where 1 indicates
rejecting the null hypothesis. It is imperative that we control both
the probability of false positives (incorrectly rejecting a true null
hypothesis, \emph{type-1 error}) and false negatives (incorrectly
retaining a false null hypothesis, \emph{type-2 error}).  We thus
  define the risk to be
\begin{equation*}
  R_N(T) = \underbrace{\mathbb{P} (T = 1 | H_0^N \text{ is true})}_{\text{type-1 error}} 
  +  \underbrace{\sup_{\s \in \mathcal{S}_N  } \mathbb{P} (T = 0 | H_1^N \text{ is true})}_{\text{type-2 error}},
\end{equation*}
where the class of localized attributes $\mathcal{S}_N$ is related to
the number of nodes $N$.
Using the definition from~\cite{CastroCD:11,SharpnackRS:13,
  SharpnackKS:13a}, we say that $H_0^N$ and $H_1^N$ are
\emph{asymptotically distinguishable} by a test $T$, if $\lim_{N
  \rightarrow \infty} R_N(T)=0$. In other words, when the number of
nodes goes to infinity and the detection risk goes to zero, $H_0^N$
and $H_1^N$ are asymptotically distinguishable.

\mypar{Bernoulli noise model} In this paper, we are particularly
interested in~\eqref{eq:testing} with the Bernoulli noise model,
 $ f(\s, \epsilon) = \text{Bernoulli}( \s + \epsilon \one_\V) \in
  \R^N, $ where each element $f(\s, \epsilon)_i$ is an independent
  Bernoulli random variable with mean $(\s + \epsilon)_i$,
\begin{eqnarray}
  \label{eq:BernoullieM}
  H_0^N & : &    \y \sim \text{Bernoulli}( \epsilon \one_\V),
  \\
  \nonumber
  H_1^N & : &    \y \sim \text{Bernoulli}( \mu \one_C + \epsilon \one_{\overline{C}} )  \text{ for all } \TV_p( \one_C) \leq \rho,
\end{eqnarray}
$\mu$ is the activation probability within the localized
attribute, $\epsilon$ is the noise level and $0 \leq \epsilon < \mu
\leq 1$. The number of external edges of a localized attribute,
$\rho$, reflects the shapes of candidate localized attributes and
characterizes the alternative hypothesis $H_1^N$. Here, the average
value under $H_1^N$ is larger than the average value under $H_0^N$. A
naive approach is to use the average of the observation as the
statistic (see Appendix~\ref{sec:Appendices4}). We set the class of
localized attributes $\mathcal{S}_N$ to model correlation among
nodes as
\begin{displaymath}
  \mathcal{S}_N \ = \ \bigg\{ \s: \s = (\mu - \epsilon) \one_C, C \in
  \mathcal{C} \bigg\},
\end{displaymath}
with the localized attributes that the user is testing for specified
through the class $ \mathcal{C} = \{ C \subseteq \V: \TV_p ( {\bf
  1}_{C} ) \leq \rho \}$, while $\rho, p$ control the cut cost of the
activated node set.  The cut cost $\rho$ is a user-defined parameter:
when $\rho$ is large, all candidate localized attributes are allowed
to have any number of external edges and the test always succeeds,
while when $\rho$ is small, all candidate localized attributes have
few external edges. Note that the Bernoulli model here is similar to
the setting in community detection with categorical
attributes~\cite{YangML:13, AkogluTK:2015}. For example, suppose that
we want to identify whether users who graduated from CMU form a social
community on Facebook.  The binary value \emph{Is this user from CMU?}
is an attribute on Facebook. When this attribute leads to a community,
we should find a subgraph such that (1) most nodes are activated
within the subgraph and few nodes are activated outside the subgraph;
(2) the connection between this subgraph and its complement is
weak. \emph{We describe a binary attribute by the Bernoulli noise
  model and a localized attribute by an attribute with small
  total-variation.}

\section{Related Work}
\label{sec:relatedwork}
In classical signal processing, a localized signal is constant over
local connected regions separated by lower-dimensional boundaries.  It
is often related to concepts such as impulse function, step function,
square wave and Haar basis~\cite{VetterliKG:12}. Detecting localized
signals has been considered through signal/noise
discrimination~\cite{North:63}, edge detection~\cite{Gonzalez:02},
pattern matching~\cite{MahalanobisKC:87} and support recovery of
sparse signals~\cite{HauptCN:09,CevherIHB:09}, among others. We here
look at the counterpart problem on graphs. A localized attribute
(graph signal) is constant over a subgraph that is easily separated
from the rest of the nodes. Similarly to localized signals in
classical signal processing, a localized attribute emphasizes fast
transitions (corresponding to boundaries) and localization in the
graph vertex domain (corresponding to attributes that are nonzero in a
local neighborhood).

Our detection problem bears resemblance to many detection problems in
the current graph-related literature, such as detecting a smooth graph
signal or a localized graph signal under a specific noise model. For
example,~\cite{CastroCD:11,ZouLP:2016} detects a cluster in a lattice
graph that exhibits unusual behavior;~\cite{HuCSFLL:13} constructs a
generalized likelihood test to detect smooth graph
signals;~\cite{Krishnamurthy:15} considers a general graph-structured
normal means test;~\cite{HeydariTP:15} considers combining data gathering and decision-making to design the quickest detection in the markov random field;~\cite{SharpnackKS:13}, constructs the uniform
spanning tree wavelet statistic to approximate the epsilon scan
statistic; and ~\cite{SharpnackRS:13, SharpnackKS:13a}, considers the
Lovasz extended scan statistic and spectral relaxation as relaxations
of the combinatorial scan statistic.

The uniform spanning tree wavelet statistic and the Lovasz extended
scan statistic lay a foundation for this paper; we extend the Gaussian
noise model to the Bernoulli one, that is, we deal with binary instead
of real values to address categorical attributes.  Although one could
do this by thresholding a real value, the process leads to information
loss. Thus, handling binary-valued attributes is a nontrivial task.

Our detection problem is also related to community detection, which,
as one of the key topics in network science and graph mining, aims to
extract tightly connected subgraphs in a network, also known as graph
partitioning and graph clustering~\cite{GirvanN:02, BoykovVZ:01,
  Luxburg:07}. While the traditional community detection algorithms
focus on the graph structure only~\cite{AhnBL:10,YangL:13}, some
recent studies tried to combine the knowledge of both graph structure
and node attributes~\cite{YangML:13} as such attributes not only
improve the accuracy of community detection, but also provide the
interpretation of detected communities. However, as not all attributes
are relevant for all communities, community detection accuracy may
suffer. It is also computational inefficient to include a large number
of attributes in the training phase~\cite{ZhangLZ:15}. We here aim to
find useful attributes for improving community detection and aiding
interpretation; as an example, in Section~\ref{sec:experiments}, the
proposed statistics select useful keywords in a co-authorship
network.

\section{Methodology}
\label{sec:methodology}
We now propose two statistics for testing the
  hypothesis in~\eqref{eq:BernoullieM}: graph wavelet statistic and
graph scan statistic. The first is based on a graph wavelet basis;
when a given attribute has large graph wavelet coefficients, the
attribute agrees with the graph structure and is localized.  The
second is based on matching all possible node sets to a given
attribute via an optimization problem; when we find such a
feasible node set, the attribute is localized. The first statistic is
more efficient as the graph wavelet basis is pre-constructed, while
the second is more accurate as it adaptively searches for localized
attributes.

\subsection{Graph Wavelet Statistic} In classical signal
  processing, one way of detecting a transient change in a time-series
  signal is by projecting it on the wavelet
  basis~\cite{VetterliKG:12}; when a high-frequency coefficient is
  large, a transient change is present.  Similarly, we detect a
  boundary in a localized attribute by projecting it on a Haar-like
  graph wavelet basis; when a large graph wavelet coefficient exists
  in the high-frequency wavelet matrix, the boundary is present, and
  we reject the null hypothesis.

 Similarly to how we construct the Haar basis in classical signal
  processing~\cite{VetterliK:95}, we construct the \emph{graph wavelet
    basis} as in~\cite{SharpnackKS:13,ChenVSK:16}. The idea is to
  recursively partition each local parent set into two disjoint local
  child sets of roughly similar sizes, irrespective of the connections
  between them. We start from the entire node set $\V$, corresponding
  to the coarsest resolution in the graph vertex domain, and finish
  with each local set being either an individual node or an empty set,
  corresponding to the finest resolution in the graph vertex domain as
  illustrated in Figure~\ref{fig:decomposition}. For each partition, a
  new basis vector $\w$ is added as in Algorithm~\ref{alg:wavelet}.
  The decomposition level, or the depth of a decomposition tree, is
  the maximum number of partitions to reach an individual node.  As
  the proposed decomposition provides a series of redundant local sets
  with various sizes at various positions, we can either exactly
  localize attributes or approximate them by using local sets.

\begin{algorithm}[htb]
  \footnotesize
  \caption{\label{alg:wavelet} Local-set-based graph wavelet basis}
  \begin{tabular}{@{}lll@{}}
    \addlinespace[1mm]
    {\bf Input} 
    & $G(\V, \E, \Adj )$~~~~~graph \\
    {\bf Output}  
    & $\W \in \R^{N \times N}$~~~wavelet basis \\
    \addlinespace[2mm]
    {\bf Function} & &\\
    \multicolumn{3}{l}{~initialize a stack of node sets $\mathbb{S}$ and a set of basis vectors $\W$} \\
    \multicolumn{3}{l}{~set $\mathbb{S} = \{S = \V\}$} \\ 
    \multicolumn{3}{l}{~set $\w_1 = \frac{1}{ \sqrt{|S|}} {\bf 1}_S$ as the first column of $\W$} \\
    \multicolumn{3}{l}{~while the cardinality of the largest element of $\mathbb{S}$ is larger than $1$ } \\ 
    \multicolumn{3}{l}{~~~~~take one element from $\mathbb{S}$ at a time as $S$}\\
    \multicolumn{3}{l}{~~~~~partition $S$ into two disjoint local sets $S_1, S_2$  by 2-means clustering~\cite{ChenVSK:16}} \\
    \multicolumn{3}{l}{~~~~~if $|S_1|$ and/or $|S_2|$ is larger than 1,  put that local set(s) into $\mathbb{S}$} \\
    \addlinespace[0.5mm]
    \multicolumn{3}{l}{~~~~~add $\w = \sqrt{ \frac{|S_1|  |S_2|}{|S_1| +  |S_2|}} \left(  \frac{1}{|S_1|} {\bf 1}_{S_1} - \frac{1}{|S_2|} {\bf 1}_{S_2}\right)$ as a new column of $\W$}\\
    \addlinespace[0.5mm]
    \multicolumn{3}{l}{~~~{\bf return} $\W$} \\  
    \addlinespace[1mm]
  \end{tabular}
\end{algorithm}



To ensure the detection property of the graph wavelet basis, we impose
three requirements on each partition: (1) the two local child sets are
disjoint; (2) the union of the two local child sets is the local
parent set; and (3) the cardinalities of the two local child sets are
as close as possible. The first two requirements lead to the
orthogonality of the graph wavelet basis while the third promotes the
sparsity for all attributes with small $\ell_0$-norm-based total
variation. In general, any algorithm that satisfies these three
requirements can be used to generate a graph wavelet basis;~\cite{ChenVSK:16} introduces three such algorithms. Here we use
  2-means clustering, which approximately partitions a local set
  evenly.  Inspired by $K$-means clustering~\cite{Bishop:06}, for each
  local set, we select two nodes with the longest geodesic distance from each other as the cluster heads and assign every other
  node to its nearest cluster head based on the geodesic distance. We
  then recompute the cluster head for each cluster by minimizing the
  geodesic distance sum to all other nodes in the cluster and assign
  node to its nearest cluster head again, until convergence.

\begin{figure}[t]
  \begin{center}
     \includegraphics[width= 0.8\columnwidth]{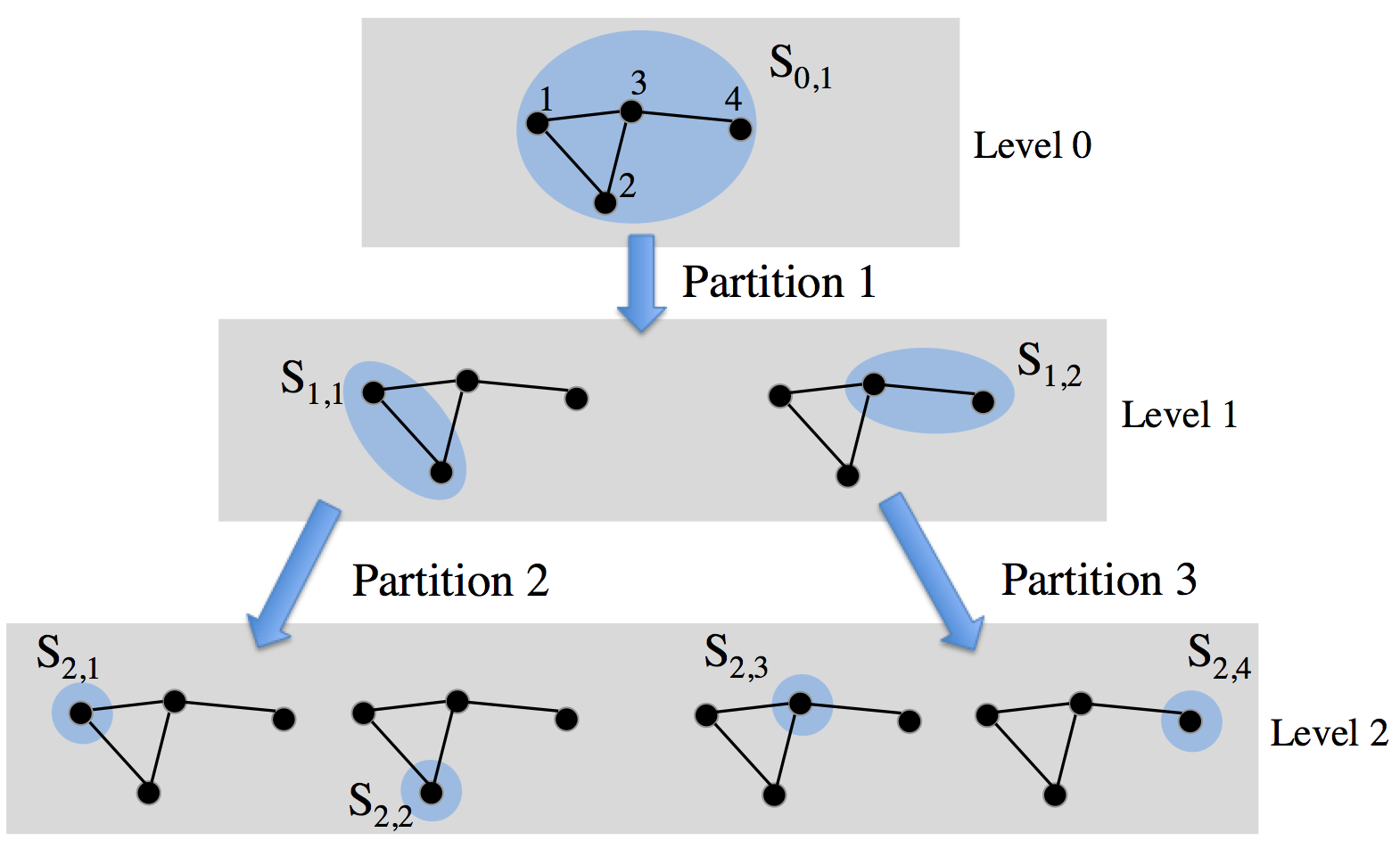}
  \end{center}
  \caption{\label{fig:decomposition} Wavelet decomposition tree. In
    each partition, we decompose a local set into two disjoint local
    sets and generate a wavelet basis vector. For example, in
    Partition 1, we partition $S_{0,1} = \{1, 2, 3, 4 \}$ into
    $S_{1,1} = \{1, 2 \}, S_{1,2} = \{3, 4 \}$ and generate a wavelet
    basis vector $[ 1~1~-1~-1 ]/2$. Note that the decomposition is not
    unique. }
\vspace{-3mm}
\end{figure}

 The output of Algorithm~\ref{alg:wavelet} is the graph wavelet
  basis $\W \in \R^{N \times N}$. It is orthonormal and
  preserves the energy of any input.
\begin{myLem}{\bf (Orthogonality~\cite{ChenVSK:16})}
  \label{lem:orthonormal}
  Let $\W$ be the output of Algorithm~\ref{alg:wavelet}.  $\W$ is an
  orthonormal basis; that is,
  \begin{equation*}
    \W^T \W  \ = \ \W \W^T \ =  \ \Id.
  \end{equation*}
\end{myLem}

 We use this graph wavelet basis with balanced splits to
  construct a sparse representation for localized attributes (which
  are not communities in general).  The upper bound in
  Lemma~\ref{lem:sparse} below establishes the worst-case scenario,
  because the graph wavelet basis once constructed is fixed and thus
  should work for any attribute. The decomposition level $L$ is
  the bottleneck in a sparse representation and is minimized when we
  partition each local set evenly, leading to balanced splits.
\begin{myLem}{\bf (Sparsity~\cite{ChenVSK:16})}
  \label{lem:sparse}
  Let $\W$ be the output of Algorithm~\ref{alg:wavelet} and $L$ be the
  total decomposition level.
  For all $\y \in \R^N$,
  \begin{eqnarray*}
    \left\| \W^T \y \right\|_0 \leq  1 + \left\|  \Delta \y \right\|_0  L.
  \end{eqnarray*}
\end{myLem}


Combining Lemmas~\ref{lem:orthonormal} and~\ref{lem:sparse}, we see
that the graph wavelet representation concentrates the energy of an attribute
into a few wavelet coefficients; it is thus indeed a sparse representation,
\begin{eqnarray*}
  \left\| \W_{(-1)}^T \y \right\|_{\infty}^2 
  \ \stackrel{(a)}{\geq} \
  \frac{ \left\| \W_{(-1)}^T \y \right\|_{2}^2  }{ \left\| \W_{(-1)}^T \y \right\|_{0}  } 
  \ \stackrel{(b)}{\geq} \
  \frac{ \left\| \y \right\|_{2}^2 - \left( \frac{1}{\sqrt{N}} \one_{\V}^T \y \right)^2  }{ 1 + \left\|  \Delta \y \right\|_0  L  },
\end{eqnarray*}
where $\W_{(-1)} \in \R^{ N \times (N-1)}$ is $\W$ without its first
(constant) column as that column only calculates the mean of $\y$,
which is not informative for detection; we call it~\emph{high-frequency wavelet matrix}. The inequality $(a)$ follows
from the basic norm inequality, and $(b)$ from
Lemmas~\ref{lem:orthonormal}
and~\ref{lem:sparse}. Theorem~\ref{thm:wavelet_statistic} will show
that the largest nontrivial wavelet coefficient is an important metric
in distinguishing whether $\y$ is a localized attribute or not.  We
see that the lower bound on that largest nontrivial wavelet
coefficient is related to the total decomposition level.  To lift the
largest nontrivial wavelet coefficient up, we need to minimize the
total decomposition level $L$, which is satisfied with even local set
partition, with $L = O(\log_2 N)$. For a localized attribute with a
small $\ell_0$-norm-based total variation, the corresponding graph
wavelet coefficients are sparse and the energy of the original
attribute concentrates in a few graph wavelet coefficients. However,
for a noisy attribute with a large $\ell_0$-norm-based total
variation, the energy of the original attribute spreads over all graph
wavelet coefficients.

The projection of an attribute on a graph wavelet basis vector
  calculates the absolute difference between its average values on two
  local node sets (see Algorithm~\ref{alg:wavelet}); for example, the
  projection on a basis vector stemming from the partition into $S_1$
  and $S_2$ is $\sqrt{ \frac{|S_1| |S_2|}{|S_1| + |S_2|}} \y^T \left(
    {\bf 1}_{S_1}/{|S_1|} - {\bf 1}_{S_2}/{|S_2|} \right)$. When one
local set captures significantly larger average value than the other
local set, that local set detects a localized attribute. Because of
the multiresolution construction, the graph wavelet basis searches for
localized attributes of different sizes. Thus, the maximum value of
the graph wavelet coefficient identifies whether the original
attribute contains a localized attribute.

We thus define the~\emph{graph wavelet statistic} as the maximum absolute value over the high-frequency wavelet coefficients,
\begin{equation}
\label{eq:gws}
\widehat{w} = \left\|  \W_{(-1)}^T \y \right\|_\infty,
\end{equation}
where  $\y$ is the noisy observation.
When $\widehat{w}$ is larger than a threshold, we reject the null hypothesis. 

 To analyze the graph wavelet statistic,
  Lemma~\ref{lem:wavelet_statistic} shows that given a threshold
  related to the graph size $N$ and a user-defined error tolerance
  $\delta$, the type-1 error is upper bounded,
  Theorem~\ref{thm:wavelet_statistic} shows that when the attribute strength
  is sufficiently large, both type-1 and type-2 errors can be upper
  bounded and Corollary~\ref{cor:wavelet_statistic} states the
  condition for the asymptotic distinguishability.
\begin{myLem}
    \label{lem:wavelet_statistic}
    Let the graph wavelet statistic be $\widehat{w}$
    in~\eqref{eq:gws}.  Under the statistical
    test~\eqref{eq:BernoullieM} with $p = 0$, we reject the null
    hypothesis for all $\widehat{w} > \tau$, with threshold $\tau = \sqrt{\log
      N} + \sqrt{ 2 \log ({2}/{\delta})}$. The corresponding type-1
    error is $\mathbb{P} \{ T = 1 | H_0^N \text{ is true} \} \leq
    \delta$.
\end{myLem}

\begin{myThm}
  \label{thm:wavelet_statistic}
  Let the attribute strength be sufficiently large, 
  \begin{eqnarray}
    \label{eq:wavelet_condition}
    && \sqrt{ |C| \left(1- \frac{|C|}{N} \right) }  \left( \mu - \epsilon \right) \  \geq  \
    \\
    && \sqrt{1+ \rho \log N } \Bigg( \sqrt{ \log N } +  \sqrt{ 2 \log(\frac{2}{\delta_1}) } + \sqrt{ 2 \log(\frac{2}{\delta_2})  }  \Bigg). \nonumber 
  \end{eqnarray}
  Then, by using the graph wavelet statistic $\widehat{w}$
  in Lemma~\ref{lem:wavelet_statistic}, the type-1 error is $\mathbb{P} ( T = 1 | H_0^N
  \text{ is true} ) \leq \delta_1$ and the type-2 error is $\mathbb{P}
  ( T = 0 | H_1^N \text{ is true} ) \leq 1-(1-\delta_2)^4$.
\end{myThm} 

We set $\delta_1 = \delta_2 = 1/N$ to obtain the following corollary.

\begin{myCorollary}
  \label{cor:wavelet_statistic}
  Using the graph wavelet statistic $\widehat{w}$ in~\eqref{eq:gws},
  $H_0^N$ and $H_1^N$ are asymptotically distinguishable, that is,
  $\lim_{N \rightarrow \infty} R_N(T)=0$, when
  \begin{displaymath} 
    \sqrt{ |C| \left(1- \frac{|C|}{N} \right) } \left( \mu - \epsilon
    \right) \ \geq \ O \left( \sqrt{\rho} \log N \right).
  \end{displaymath}
\end{myCorollary}
 The proofs of Lemma~\ref{lem:wavelet_statistic} and
Theorem~\ref{thm:wavelet_statistic} are merged in
Appendix~\ref{sec:Appendices1}. The main idea is to show that under
the null hypothesis, each graph wavelet coefficient is a
  sub-Gaussian random variable whose distribution is similar to a
  Gaussian distribution~\cite{Ledoux:01}, while under the alternative
hypothesis, the maximum value of the graph wavelet coefficients is
large because the energy of the original attribute concentrates in a
few graph wavelet coefficients.

While asymptotic distinguishability in
  Corollary~\ref{cor:wavelet_statistic} cannot be evaluated in practice
  because $\mu$, $\epsilon$ and $C$ are unknown, it quantifies the
  fundamental detection performance of an algorithm and depends
  on $\mu, \epsilon$ and $C$. When the attribute strength is too weak,
  for example, $\mu = \epsilon$ or $C = 0$, it is impossible for any
  algorithm to achieve asymptotic distinguishability. In the
  detection literature, it is common to show that when a predefined
  signal strength is sufficiently large, the null and alternative
  hypotheses are asymptotically distinguishable.
  Theorem~\ref{thm:wavelet_statistic} and Corollary~\ref{cor:wavelet_statistic} follow the same path; see other
  similar examples in~\cite{SinghNC:10,CastroCD:11,
    SharpnackKS:13,SharpnackKS:13a, ZouLP:2016}.

  Theorem~\ref{thm:wavelet_statistic} relates the size of a localized
  attribute $|C|$, the activation probability difference $\mu -
  \epsilon$ and asymptotic distinguishability.  With a constant $|C|$,
  it is easier to detect a localized attribute when the activation
  probability difference $\mu - \epsilon$ is large; with constant
  $\mu, \epsilon$, it is easier to detect a larger localized attribute
  with a small cut cost $\rho$. When $\rho$ is large, all candidate
  localized attributes are allowed to have any number of external
  edges and a larger $|C|$ is required to increase the attribute
  strength, while when $\rho$ is small, all candidate localized
  attributes have few external edges.  When $|C|$ is fairly large;
  that is, $O(N) \gg |C| \gg O(1)$, $(1-|C|/N)\sqrt{ |C|} (\mu -
  \epsilon)$ asymptotically approximates $\sqrt{ |C|} (\mu -
  \epsilon)$; When $|C|$ is too close to $N$, the
  condition~\eqref{eq:wavelet_condition} fails because the observation
  $\y$ is close to an all-one vector and most of the energy is
  captured by the first column vector of the graph wavelet basis,
  causing a small $\widehat{w}$; however, in practice we typically
  consider $|C| \ll O(N)$, because localized attributes are relatively
  small compared to the entire graph.

Since the distribution of graph wavelet statistic does not have an
analytical form, it is hard to calculate the exact $p$-value. Instead,
we use sub-Gaussianity to provide an upper bound on the $p$-value.
Given the graph wavelet statistic $\widehat{w}$ in~\eqref{eq:gws}, the
upper bound on the $p$-value is $\exp{ \left(-(1/2) \left( \widehat{w}
      - \sqrt{\log N} \right)^2 \right)}$. Let our test be level
$\alpha$. When $\alpha$ is larger than the upper bound on the
$p$-value, $\alpha$ is definitely larger than the exact $p$-value. We
then reject the null hypothesis for all $\alpha \geq \exp{\left(-(1/2)
    \left( \widehat{w} - \sqrt{\log N} \right)^2 \right)}$. The
threshold is only related to the size of the graph $N$.

The computational bottleneck in constructing the graph wavelet basis
is the graph partition algorithm from
Figure~\ref{fig:decomposition}. Let the computational cost of the
graph partition algorithm be of the order $O( h(N) )$, where
$h(\cdot)$ is a polynomial function. The total computational cost to
construct a graph wavelet basis behaves as $O \left( \sum_{i=0}^{\log
    N} 2^i h(N/2^i) \right)$; for example, when the cost of a graph
partition algorithm is of the order $O(N \log N)$, the total
computational cost to construct a graph wavelet basis behaves as $O(N
\log^2 N)$. Since the graph wavelet basis is constructed based on the
graph structure only, the construction is performed only once and
works for any attribute supported on this graph. The total
computational cost to obtain the graph wavelet statistic only involves
a matrix-vector multiplication and a search for the maximum value. The
graph wavelet statistic is thus scalable to large-scale graphs.

\subsection{Graph Scan Statistic}
\label{sec:gss}
In the previous subsection, we constructed a graph wavelet
statistic to test whether a given attribute is localized. While the graph wavelet statistic is efficient, the
construction of the graph wavelet basis does not depend on the
attribute. We now propose a data-adaptive approach, which scans all
feasible node sets based on a given attribute. The intuition behind
the proposed statistics is that given the noisy observation $\y$, we
search for an activated node set $C$. If we can find such a $C$, we
reject the null hypothesis, and vice versa.

If we knew the true activated node set $C \in \mathcal{C}$, we could
test the null hypothesis $H_0^N: \s = 0$ against the alternative
$H_1^N: \s = \mu \one_C$ by using the likelihood ratio test.  Given
the observation $\y$ and the Bernoulli noise model, the
likelihood is
\begin{eqnarray*}
  \mathbb{P} (\y |H_1^N \text{ is true})
  \ = \ 
  \prod_{i \in C}  \mu^{y_i}  (1- \mu)^{1-y_i}  
  \prod_{i \in \overline{C} }  \epsilon^{y_i}  (1-\epsilon)^{1-y_i},
\end{eqnarray*}
and we estimate the unknown parameters as $\widehat{\mu} = \one_C^T \y
/ |C|$ and $\widehat{\epsilon} = \one^T \y /N$. The likelihood ratio
is
\begin{eqnarray*}
  &&
  \frac{  \prod_{i \in \V}  \widehat{\epsilon}^{y_i}  (1- \widehat{\epsilon})^{1-y_i}  } { \prod_{i \in C}  \widehat{\mu}^{y_i}  (1- \widehat{\mu})^{1-y_i}  \prod_{i \in \overline{C} }  \widehat{\epsilon}^{y_i}  (1-\widehat{\epsilon})^{1-y_i}  }
  \\
  & = &   \prod_{i \in C}  \left( \frac{\widehat{\epsilon}}{ \widehat{\mu} } \right)^{y_i}  \left( \frac{ 1- \widehat{\epsilon} } {1- \widehat{\mu}} \right)^{1-y_i} .
\end{eqnarray*}
The log likelihood ratio is
\begin{eqnarray*}
&&  \sum_{i \in C}  y_i  \log \left( \frac{\widehat{\epsilon}}{ \widehat{\mu} } \right) +  \sum_{i \in C}  (1-y_i)  \log  \left( \frac{ 1- \widehat{\epsilon} } {1- \widehat{\mu}} \right)
\\
& = & |C| \left(  \widehat{\mu} \log \left( \frac{\widehat{\epsilon}}{ \widehat{\mu} } \right) +  (1- \widehat{\mu}) \log \left( \frac{1-\widehat{\epsilon}}{1- \widehat{\mu} } \right) \right)
\\
& = & - |C| \KL( \widehat{\mu} \| \widehat{\epsilon} ),
\end{eqnarray*}
where $\KL(\cdot \| \cdot)$ is the  Kullback-Leibler divergence~\cite{Bishop:06}.

In practice, however, the true activated node set $C$ is unknown; we then
consider the generalized likelihood ratio
\begin{eqnarray}
  \label{eq:GSS}
  \widehat{g}  & = &  
  \max_C  |C| \KL \left( \frac{ \one_C^T \y} {|C| } \| \frac{ \one^T \y} {N} \right) \\ \nonumber
  & & \text{subject to }  \TV_1 (\one_C) \leq \rho.
\end{eqnarray}
We call $\widehat{g}$~\emph{graph scan statistic}. To maximize the
objective in~\eqref{eq:GSS}, the localized attribute $C$ should
trade-off between its size $|C|$ and the average value inside $\one_C^T
\y /|C|$. When $|C|$ is large, $\one_C^T \y /|C|$ tends to be small;
on the other hand, when $C$ fits the activated nodes in $\y$,
$\one_C^T \y /|C|$ is large; however, due to the cut cost constraint,
$C$ can only fit a few scattered nodes and $|C|$ is small. The goal of
the graph scan statistic is to search for a node set with both large
cardinality and large average value.

When $\widehat{g}$ is larger than some threshold, we detect an
activated node set and reject the null hypothesis. To analyze
  the graph scan statistic, Lemma~\ref{lem:graph_scan_statistic} shows
  that given a threshold related to the size of the graph $N$, a
  user-defined error tolerance $\delta$ and cut cost $\rho$, the
  type-1 error is upper bounded,
  Theorem~\ref{thm:graph_scan_statistic} shows that when the attribute
  strength is sufficiently large, both type-1 and type-2 errors can be
  upper bounded and Corollary~\ref{cor:graph_scan_statistic} states
  the condition of the asymptotic distinguishability.

 \begin{myLem}
    \label{lem:graph_scan_statistic}
    Let the graph scan statistic be $\widehat{g}$ in~\eqref{eq:GSS}.
    Under the statistical test~\eqref{eq:BernoullieM} with $p = 1$, we
    reject the null hypothesis for all $\widehat{g} > \tau$, with
    \begin{eqnarray*}
      \tau & = & 8 \Bigg( \bigg( \sqrt{\rho} + \sqrt{\frac{1}{2} \log N} \bigg) \sqrt{ 2 \log (N-1)}
      \\ 
      && + \sqrt{2 \log 2} + \sqrt{ \frac{9}{2} \log (\frac{2}{\delta})  } \Bigg)^2.
    \end{eqnarray*} 
    The corresponding type-1 error is $\mathbb{P} \{ T = 1 | H_0^N
    \text{ is true} \} \leq 1- (1-\delta)^2$.
\end{myLem}

\begin{myThm}
  \label{thm:graph_scan_statistic}
  Let the attribute strength be sufficiently large,
  \begin{eqnarray}
    \label{eq:graph_scan_condition}
    &&  \left( 1-\frac{|C|}{N}  \right) \sqrt{|C|} (\mu - \epsilon)
    \ \geq  \  \nonumber
    \Bigg( 4 \sqrt{ \log 2} + 6 \sqrt{ \log (\frac{2}{\delta_1} ) } +
    \\  \nonumber
    && 4 \bigg( \sqrt{\rho} + \sqrt{\frac{1}{2} \log N} \bigg) \sqrt{ \log (N-1)}  +  
    \\ 
    &&
    \left(  \sqrt{ \frac{1}{2} } + \sqrt{\frac{|C|}{2N} } \right) \sqrt{ \log (\frac{2}{\delta_2}) }    \Bigg).
  \end{eqnarray}
  Then, by using the graph scan statistic $\widehat{g}$
  in Lemma~\ref{lem:graph_scan_statistic}, the type-1 error is $\mathbb{P} ( T = 1 | H_0^N
  \text{ is true} ) \leq 1-(1-\delta_1)^2$ and the type-2 error is
  $\mathbb{P} ( T = 0 | H_1^N \text{ is true} ) \leq 1-(1-\delta_2)^3$.
\end{myThm} 

We set $\delta_1 = \delta_2 = 1/N$ and obtain the following corollary.
\begin{myCorollary}
  \label{cor:graph_scan_statistic}
  Using the graph scan statistic $\widehat{g}$ in~\eqref{eq:GSS},
  $H_0^N$ and $H_1^N$ are asymptotically distinguishable, that is,
  $\lim_{N \rightarrow \infty} R_N(T)=0$, when
  \begin{displaymath} 
    \left( 1-\frac{|C|}{N}  \right) \sqrt{|C|} (\mu - \epsilon) \ \geq \ O \left( \max( \sqrt{\rho}, \sqrt{\log N} ) \sqrt{\log N}  \right).
  \end{displaymath}
\end{myCorollary}
 The proofs of Lemma~\ref{lem:graph_scan_statistic} and
Theorem~\ref{thm:graph_scan_statistic} are merged in
Appendix~\ref{sec:Appendices2}. The main idea is to show that under
the null hypothesis, $\one_{C}^T (\y-\epsilon \one) / \sqrt{|C|} $ is
a sub-Gaussian random variable, while under the alternative
hypothesis, the maximum likelihood estimator $\one_C^T \y/|C|$ is
close to $\mu$ with high probability. Similarly to the graph wavelet
statistic,~\eqref{eq:graph_scan_condition} shows that the key to
detecting the activation is related to the properties of the
ground-truth activated node set. When the size of the ground-truth
activated node set is larger and the ground-truth activated node set
has a small $\ell_1$-norm-based total variation, it is easier for
graph scan statistic to detect the activation. Similarly
  to~\eqref{eq:wavelet_condition}, \eqref{eq:graph_scan_condition}
  fails when $|C|$ is too close to $N$, because the two mean values,
  ${ \one_C^T \y}/{|C| }$ and ${ \one^T \y}/{N}$, in~\eqref{eq:GSS}
  are too close. In other words, ${\one^T \y}/{N}$ is a poor estimate
  for the background noise $\epsilon$. Again, in practice we typically
  consider $|C| \ll O(N)$, because localized attributes are relatively
  small compared to the entire graph.

Given the graph scan statistic $\widehat{g}$, the upper bound on the  $p$-value
is $ 2 e^{-\frac{\sqrt{2}}{3} \left( \sqrt{\frac{\widehat{g}}{8}} - 2
    \log 2 - \left( \sqrt{\rho} + \sqrt{\frac{1}{2} \log N} \right)
    \sqrt{ 2 \log(N-1) } \right)^2 }$.  Let our test be level
$\alpha$.  We reject the null hypothesis at all $\alpha \geq 2
e^{-\frac{\sqrt{2}}{3} \left( \sqrt{\frac{\widehat{g}}{8}} - 2 \log 2
    - \left( \sqrt{\rho} + \sqrt{\frac{1}{2} \log N} \right) \sqrt{ 2
      \log(N-1) } \right)^2 }$. The threshold is only related to the
size of graph $N$ and the cut cost $\rho$, which is a user-defined
parameter.

There are two advantages to the graph scan statistic over the graph
wavelet statistic: it is data adaptive and flexible by considering
edge weights. Instead of using a pre-constructed graph wavelet basis,
graph scan statistic actively searches for the activated node
set. Thus, it not only detects whether localized
activated node sets exist, but also localizes such regions. It also
takes into account edge weights by using the $\ell_1$-norm-based total
variation and is more general compared to $\ell_0$-norm-based total
variation used in graph wavelet statistic. Note that the
$\ell_1$-norm-based total variation and the $\ell_0$-norm-based total
variation are the same when we only consider binary edge
weights. Thus, all the results based on the $\ell_1$-norm can be
directly applied to the $\ell_0$-norm.

\mypar{Practical algorithms} In the previous analysis, we used the
global optimum of $\widehat{g}$ in~\eqref{eq:GSS}; this global
  optimum is hard to obtain, however, because the optimization problem
  is combinatorial. We instead consider two practical methods to
compute the graph scan statistic: the first obtains a local optimum of
the original optimization problem and the second a global optimum of a
relaxed optimization problem.

In the first method, we reformulate~\eqref{eq:GSS} and solve
\begin{eqnarray}
\label{eq:GSS_Re}
  \widehat{g}  & = &  \max_t  \max_{\x}  t \KL \left( \frac{ \x^T \y} {t} \| \frac{\one^T \y}{N} \right)
\\
\nonumber
&& \text{subject to } \x \in \{0, 1\}^N, \TV_1(\x) \leq \rho,  \one^T \x \leq t,
\end{eqnarray}
where $\x$ is an auxiliary attribute to represent $\one_C$ and $t$
denotes $|C|$. Since $\one^T \y/N$ is a small constant, for each $t$,
we optimize over $\x$ to move $\x^T \y/t$ as far away from $\one^T
\y/N$ as possible, which is equivalent to maximizing $\x^T \y$ within
the feasible region.\footnote{We implicitly assume that $\one_C^T
  \y/|C| > \one^T \y/N$.} Given a fixed $t$, we solve
\begin{eqnarray}
\label{eq:GSS_Re_t}
 \x_t^*  & = & \arg \min_{\x}  ( -\x^T \y),
\\  \nonumber
&& \text{subject to } \x \in \{0, 1\}^N, \TV_1(\x) \leq \rho,  \one^T \x \leq t.
\end{eqnarray}
The corresponding Lagrange function is  
 \begin{displaymath}
   L( \eta_1, \eta_2 , \x)  =  -\x^T \y + \eta_1  ( \one^T \x - t ) + \eta_2 \left( \left\| \Delta \x \right\|_1 - \rho \right).
 \end{displaymath} The Lagrange dual function is
\begin{eqnarray*}
&& Q(\eta_1, \eta_2) \ = \ \min_{\x \in \{0, 1\}^N} L( \eta_1, \eta_2 , \x)
\\
&= &    \min_{\x \in \{0, 1\}^N}  \left( -\x^T \y + \eta_1 \one^T \x + \eta_2 \left\| \Delta \x \right\|_1 \right) - \eta_1 t - \eta_2 \rho
\\
&= &   q( \eta_1, \eta_2 ) - \eta_1 t - \eta_2 \rho.
\end{eqnarray*}
For given  $\eta_1, \eta_2$, the function $q( \eta_1, \eta_2 )$ can be efficiently solved by $s$-$t$ graph cuts~\cite{BoykovVZ:01,KolmogorovZ:04}. We then maximize $Q(\eta_1, \eta_2)$ by using the simulated annealing and obtain $\x_t^*$ as the optimum of~\eqref{eq:GSS_Re_t}. Finally, we  optimize over $t$ by evaluating each pair of $t$ and $\x^*_t$ in the objective function~\eqref{eq:GSS_Re}. Since $\x$ takes only binary values, the optimization problem~\eqref{eq:GSS_Re_t} is not convex. However, previous works show that even the local minimum provides decent results~\cite{KolmogorovZ:04} and the computation is remarkably efficient.  We call the solution~\emph{local graph scan statistic} (LGSS) because it is a local optimum of the original optimization problem~\eqref{eq:GSS_Re} by using  graph cuts.

In the second method, we compute the graph scan statistic in a convex
fashion by relaxing the original combinatorial optimization
problem~\eqref{eq:GSS_Re},
\begin{eqnarray}
\label{eq:GSS_Re_cvx}
  \widehat{r}  & = &  \max_t  \max_{\x}  t \KL \left( \frac{ \x^T \y} {t} \| \frac{\one^T \y}{N}  \right)
\\
\nonumber
&& \text{subject to } \x \in [0, 1]^N, \TV_1(\x) \leq \rho,  \one^T \x \leq t.
\end{eqnarray}
The only difference between~\eqref{eq:GSS_Re}
and~\eqref{eq:GSS_Re_cvx} is that we relax the feasible set of
$\{0,1\}^N$ to be a convex set $[0,1]^N$. Given a fixed $t$, we
  obtain the optimum $\x_t^*$ by convex programming. We then optimize
  over $t$ by evaluating each pair of $t$ and $\x_t^*$ in the
  objective function~\eqref{eq:GSS_Re_cvx}. We call
$\widehat{r}$ the \emph{convex graph scan statistic} (CGSS). When
$\widehat{r}$ is larger than some threshold, we detect the activated
node set and reject the null hypothesis.  

 To analyze the convex graph scan statistic,
  Lemma~\ref{lem:relax_scan_statistic} shows that given a threshold
  related to the size of the graph $N$, a user-defined error tolerance
  $\delta$ and cut cost $\rho$, the type-1 error is upper bounded,
  Theorem~\ref{thm:relax_scan_statistic} shows that when the attribute
  strength is sufficiently large, both type-1 and type-2 errors can be
  upper bounded and Corollary~\ref{cor:relax_scan_statistic} states
  the condition of the asymptotic distinguishability.
 
\begin{myLem}
    \label{lem:relax_scan_statistic}
    Let the convex graph scan statistic be $\widehat{r}$
    in~\eqref{eq:GSS_Re_cvx}.  Under the statistical
    test~\eqref{eq:BernoullieM} with $p = 0$, we reject the null
    hypothesis for all $\widehat{r} > \tau$, with
    \begin{eqnarray*}
      \tau & = & 8 \Bigg( \frac{ \log 2N +1 }{ \sqrt{ \left( \sqrt{\rho} + \sqrt{\frac{1}{2} \log N} \right)^2 \log N  } }  + \sqrt{2 \log 2} +
      \\
      && 2 \sqrt{ \left( \sqrt{\rho} + \sqrt{\frac{1}{2} \log N}\right)^2 \log N }+ \sqrt{ \frac{9}{2} \log (\frac{2}{\delta})} \Bigg)^2.
    \end{eqnarray*}
    The corresponding type-1 error is $\mathbb{P} \{ T = 1 | H_0^N
    \text{ is true} \} \leq 1- (1-\delta)^2$.
\end{myLem}

\begin{myThm}
  \label{thm:relax_scan_statistic}
  Let the attribute strength be sufficiently large,
  \begin{eqnarray}
    \label{eq:relax_scan_condition}
    &&  \left( 1-\frac{|C|}{N}  \right) \sqrt{|C|} (\mu - \epsilon)
    \ \geq  \  \nonumber
    \Bigg( 4 \sqrt{ \log 2} + 6 \sqrt{ \log (\frac{2}{\delta_1} ) } +
    \\  \nonumber
    &&   \frac{ 2 \sqrt{2} ( \log 2N +1 ) }{ \sqrt{ \left( \sqrt{\rho} + \sqrt{\frac{1}{2} \log N} \right)^2 \log N  } }  + 
    \\ \nonumber
    &&
    4 \sqrt{ 2 \left( \sqrt{\rho} + \sqrt{\frac{1}{2} \log N}\right)^2 \log N } +
    \\ 
    &&
    \left(  \sqrt{ \frac{1}{2} } + \sqrt{\frac{|C|}{2N} } \right) \sqrt{ \log (\frac{2}{\delta_2}) }    \Bigg).
  \end{eqnarray}
  Then, by using the convex graph scan statistic $\widehat{r}$
  in Lemma~\ref{lem:relax_scan_statistic}, the type-1 error is $\mathbb{P} ( T = 1 |
  H_0^N \text{ is true} ) \leq 1-(1-\delta_1)^2$ and the type-2 error
  is $\mathbb{P} ( T = 0 | H_1^N \text{ is true} ) \leq
  1-(1-\delta_2)^3$.
\end{myThm} 

We set $\delta_1 = \delta_2 = 1/N$ and obtain the following corollary.
\begin{myCorollary}
  \label{cor:relax_scan_statistic}
  Using the convex graph scan statistic $\widehat{r}$
  in~\eqref{eq:GSS_Re_cvx}, $H_0^N$ and $H_1^N$ are asymptotically
  distinguishable, that is, $\lim_{N \rightarrow \infty} R_N(T)=0$,
  when
  \begin{displaymath} 
    \left( 1-\frac{|C|}{N}  \right) \sqrt{|C|} (\mu - \epsilon) \ \geq \ O \left( \max( \sqrt{\rho}, \sqrt{\log N} ) \sqrt{\log N}  \right).
  \end{displaymath}
\end{myCorollary}
 The proofs of Lemma~\ref{lem:relax_scan_statistic} and
Theorem~\ref{thm:relax_scan_statistic} are merged in
Appendix~\ref{sec:Appendices3}. The main idea is to show that under
the null hypothesis, $\x^T (\y - \epsilon \one) / \sqrt{ \one^T \x }$
is a sub-Gaussian random variable with mean zero, while under the
alternative hypothesis, the maximum likelihood estimator $\one_C^T
\y/|C|$ is a sub-Gaussian random variable with mean $\mu$. Similarly
to the graph wavelet statistic and graph scan
statistic,~\eqref{eq:relax_scan_condition} shows that the key to
detecting the activation is related to the properties of the
ground-truth activated node set.

To compute the convex graph scan statistic, we solve
\begin{eqnarray}
\label{eq:GSS_Re_cvx_t}
 \x_t^*  & = & \arg \min_{\x}   -\x^T \y,
\\  \nonumber
&& \text{subject to } \x \in [ 0, 1]^N, \TV_1(\x) \leq \rho,  \one^T \x \leq t,
\end{eqnarray}
for a given $t$.  The objective function is linear and all the
constraints are convex, so~\eqref{eq:GSS_Re_cvx_t} can be easily
solved by a convex optimization solver.
Finally, we optimize over $t$ by evaluating each pair of $t$ and
$\x^*_t$ in the objective function~\eqref{eq:GSS_Re_cvx}. Because of
the convex relaxation, the final solution of $\x$ is not binary and a
higher value of $x_i$ indicates a higher confidence that the $i$th
node is activated.

We summarize thse two methods for graph scan statistic computation in
Algorithm~\ref{alg:GSS}. While in practice the convex graph scan
statistic outperforms the local graph scan statistic, the local graph
scan statistic is more appealing when dealing with large-scale graphs.

\begin{algorithm}[htb]
  \footnotesize
  \caption{\label{alg:GSS}  Graph Scan Statistic}
  \begin{tabular}{@{}lll@{}}
    \addlinespace[1mm]
   {\bf Input} 
      ~~~~~~ $\y$~~~~~input attribute \\
     {\bf Output}  
      ~~~~ $\x^*$~~~~activated local set \\
    \addlinespace[2mm]
    {\bf Function} & &\\
     For a given $t$ \\
	~~~~~~ \emph{local graph scan statistic}: solve~\eqref{eq:GSS_Re} using graph cuts, or \\
	~~~~~~ \emph{convex graph scan statistic}: solve~\eqref{eq:GSS_Re_cvx} using convex optimization solver\\
	 search over $t$, return the largest $ t \KL \left( \frac{ {\x_t^*}^T \y} {t} \| \frac{\one^T \y}{N}  \right)$ and $\x_t^*$ as $\x^*$\\
     \addlinespace[1mm]
  \end{tabular}
\end{algorithm}

\subsection{Discussion}
We now compare the proposed statistics. 
\begin{itemize}

\item The graph wavelet statistic selects a feature by projecting
  given attributes onto a pre-constructed graph wavelet basis and is a
  data-independent and discriminative approach, which works only for
  detection.\footnote{It is possible to use the graph wavelet
    basis and nonlinear approximation to localize the localized
    attribute; this is beyond the scope of this paper.}  The graph
  scan statistic searches over graphs and localizes the localized
  attribute and is a data-dependent and generative approach, which
  works for both detection and localization.

\item From a statistical perspective, the graph wavelet statistic
  requires that the attribute strength $\sqrt{|C|} (\mu - \epsilon)$
  be larger than $O( \rho \log^2 N )$ as in
  Theorem~\ref{thm:wavelet_statistic} while the graph scan statistic
  requirea that the attribute strength $(1-|C|/N)\sqrt{|C|} (\mu -
  \epsilon)$ be larger than $O( \max ( \rho, \log N ) \log N )$ as in
  Theorems~\ref{thm:graph_scan_statistic}
  and~\ref{thm:relax_scan_statistic}. Appendix~\ref{sec:Appendices4}
  shows that the proposed graph statistics significantly outperforms a
  naive approach, which uses the mean value as statistic.

\item From a computational perspective, the graph wavelet statistic is
  the cheapest to compute. The graph scan statistic implemented
  through the local graph scan statistic is also efficient by using
  efficient graph cuts, while the convex graph scan statistic costs
  the most because it needs to solve a series of convex optimization
  problems.

\item From the perspective of empirical performance, the convex graph
  scan statistic typically provides the best performance followed by
  the graph wavelet statistic. Because graph cuts only provide a local
  solution, computing the local graph scan statistic is sensitive to
  the choice of parameters and initial conditions.
\end{itemize}

\section{Experimental Results}
\label{sec:experiments}
We now evaluate our proposed methods on three datasets. We study how
detection performance changes according to parameters on a simulated
dataset. We observe that the size of the ground-truth
activated node set is crucial for the detection, which is consistent
with
Theorems~\ref{thm:wavelet_statistic},~\ref{thm:graph_scan_statistic}
and~\ref{thm:relax_scan_statistic}.  We validate the effectiveness of
our proposed methods on two real-world problems: air-pollution
detection and attribute ranking for community detection.


\begin{figure}[h]
  \begin{center}
    \begin{tabular}{cccc}
 \includegraphics[width=0.4\columnwidth]{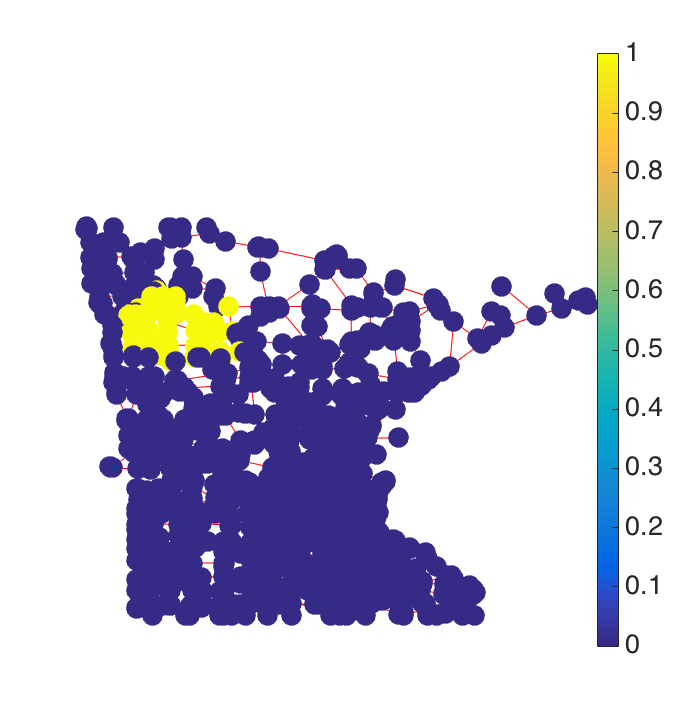} 
 &
  \includegraphics[width=0.4\columnwidth]{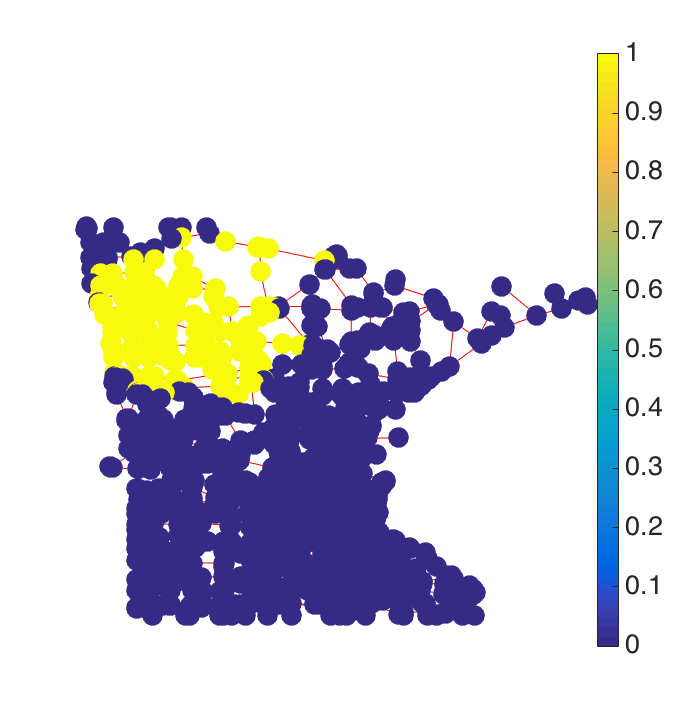} 
  \\
    {\small (a) Activated region (small).} &  {\small (b) Activated region (large).} 
    \\
    \includegraphics[width=0.4\columnwidth]{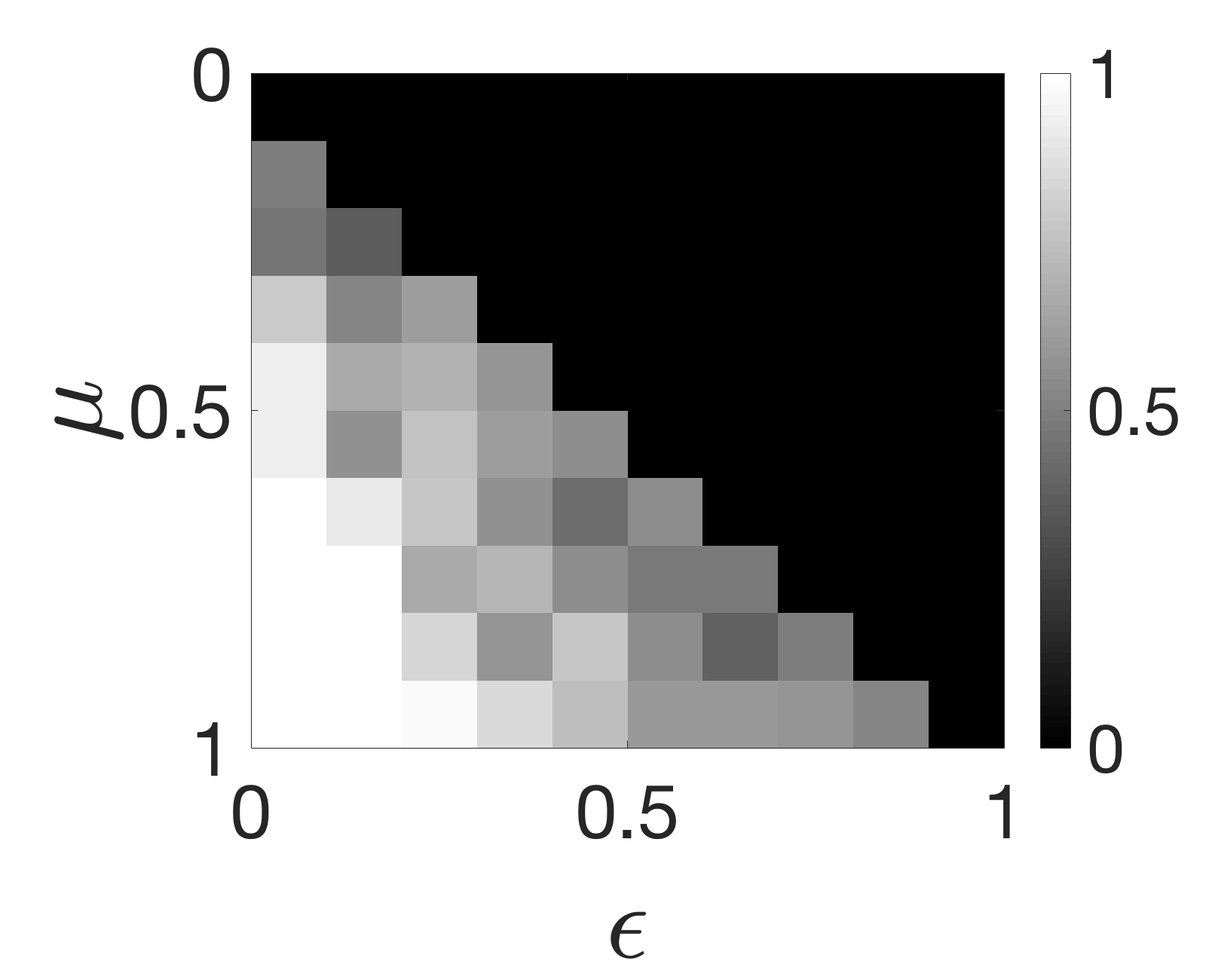}
    &
      \includegraphics[width=0.4\columnwidth]{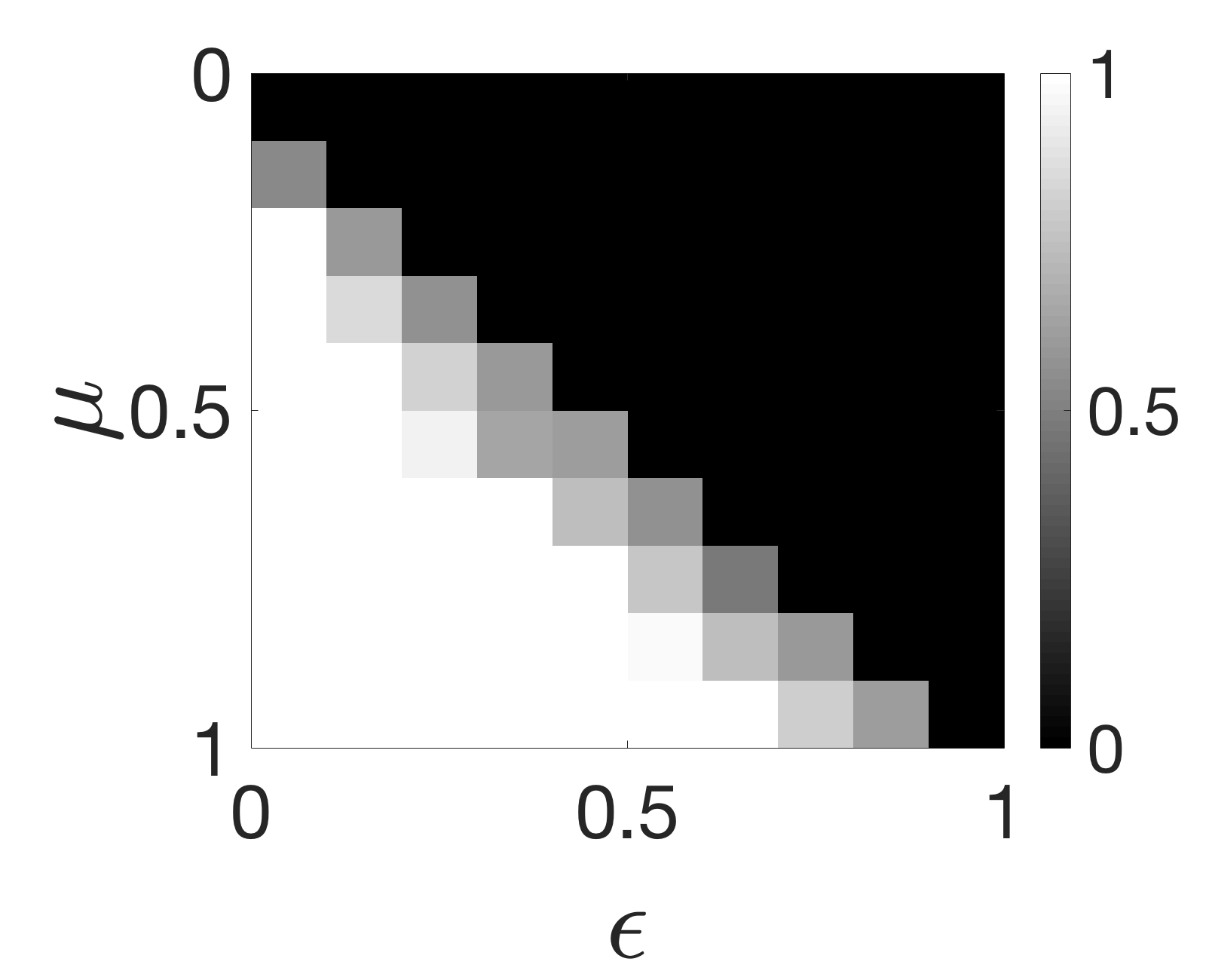}
 \\
   {\small (c) Graph wavelet statistic (small).}  &  {\small (d) Graph wavelet statistic (large).} 
   \\
   \includegraphics[width=0.4\columnwidth]{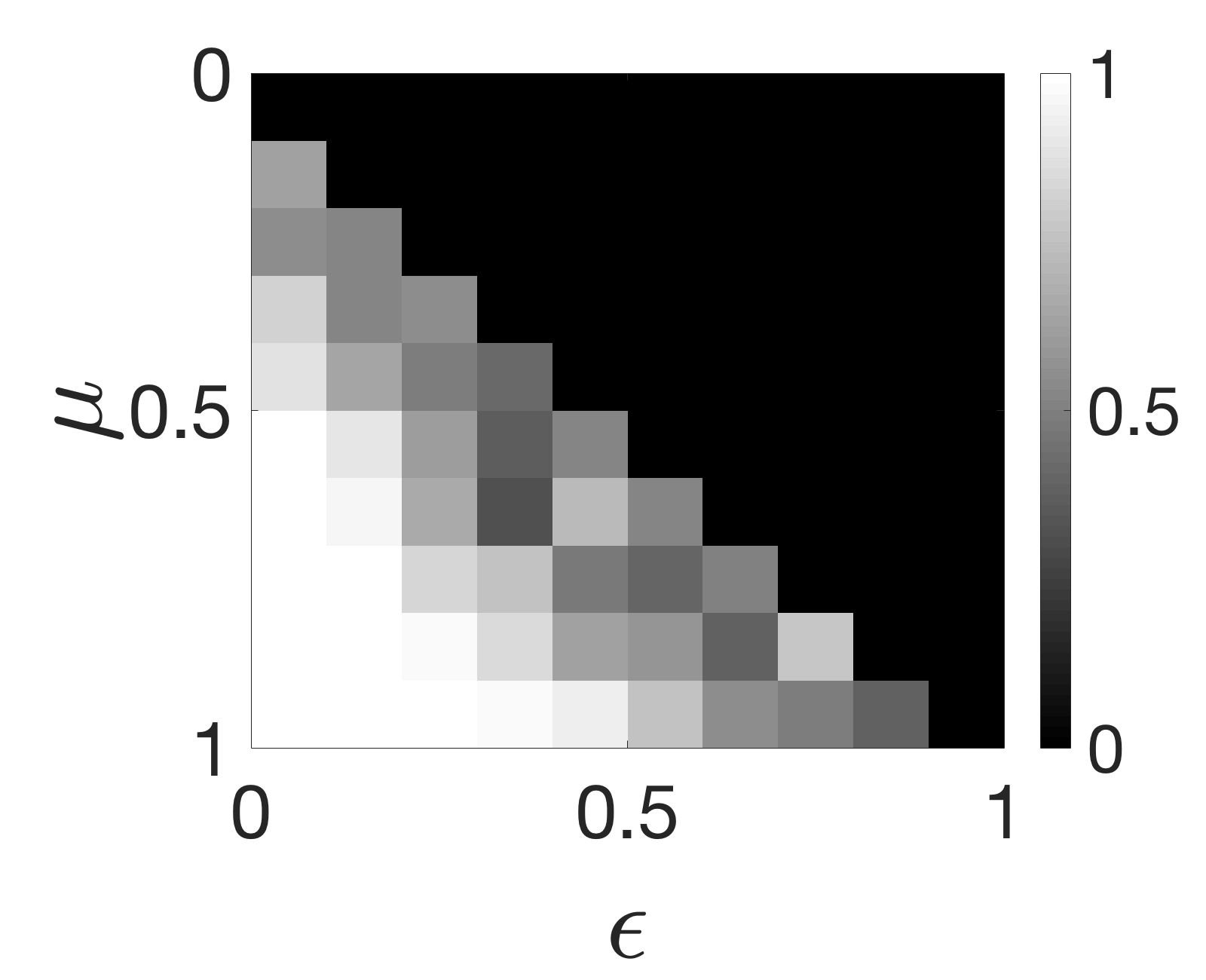}
      &
        \includegraphics[width=0.4\columnwidth]{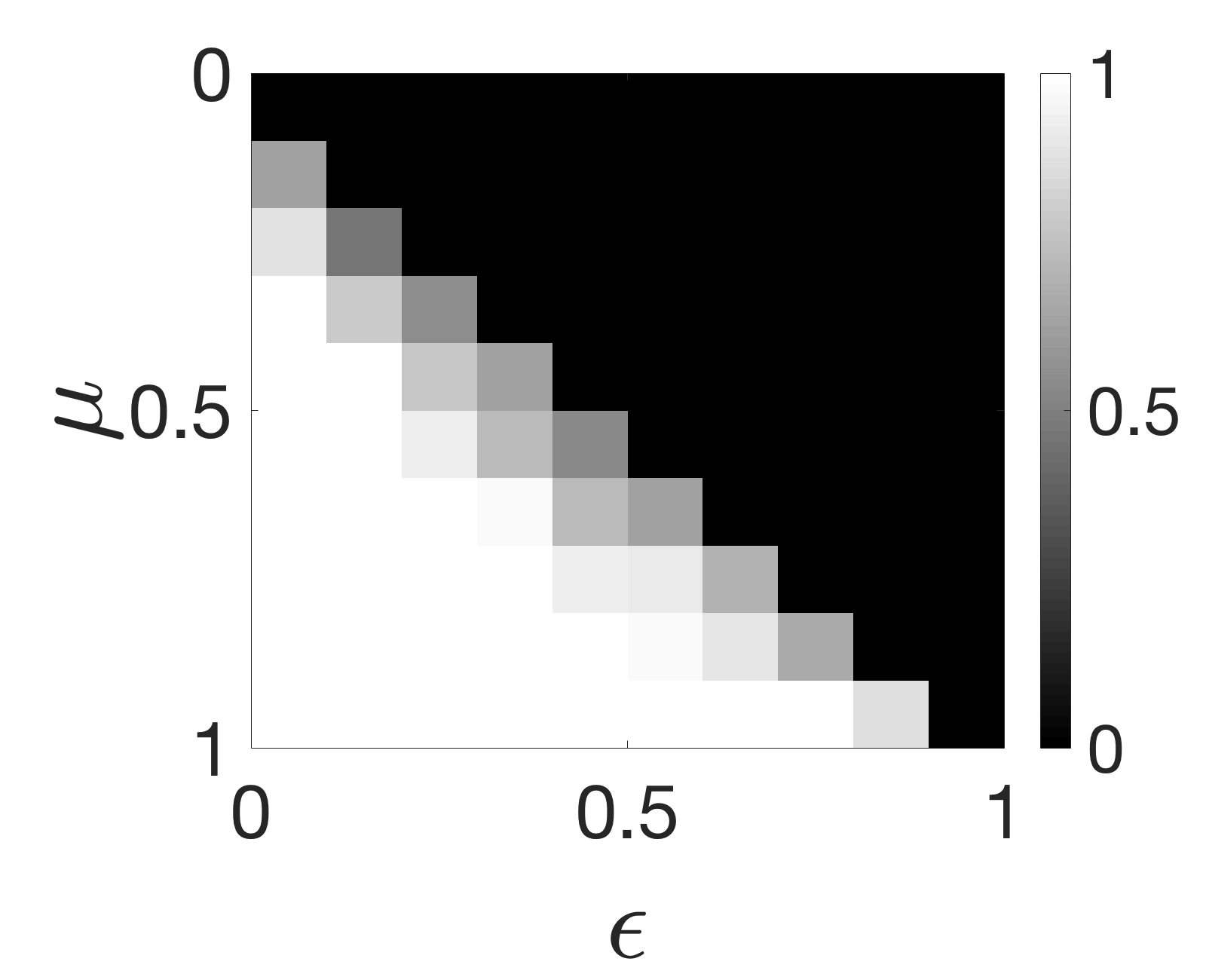}
     \\
      {\small (e) LGSS (small).}  &     {\small (f) LGSS (large).}
     \\
          \includegraphics[width=0.4\columnwidth]{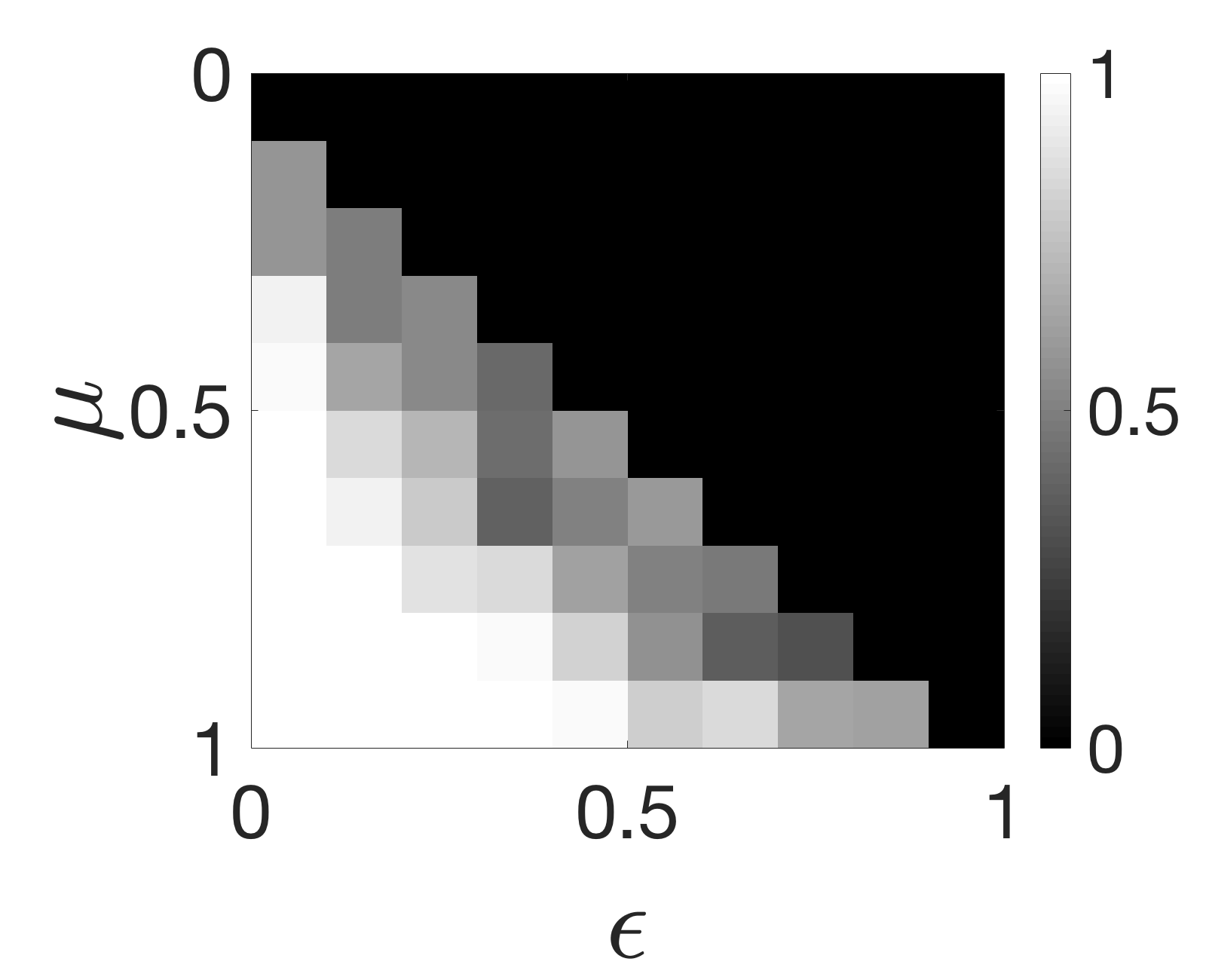}
   &
     \includegraphics[width=0.4\columnwidth]{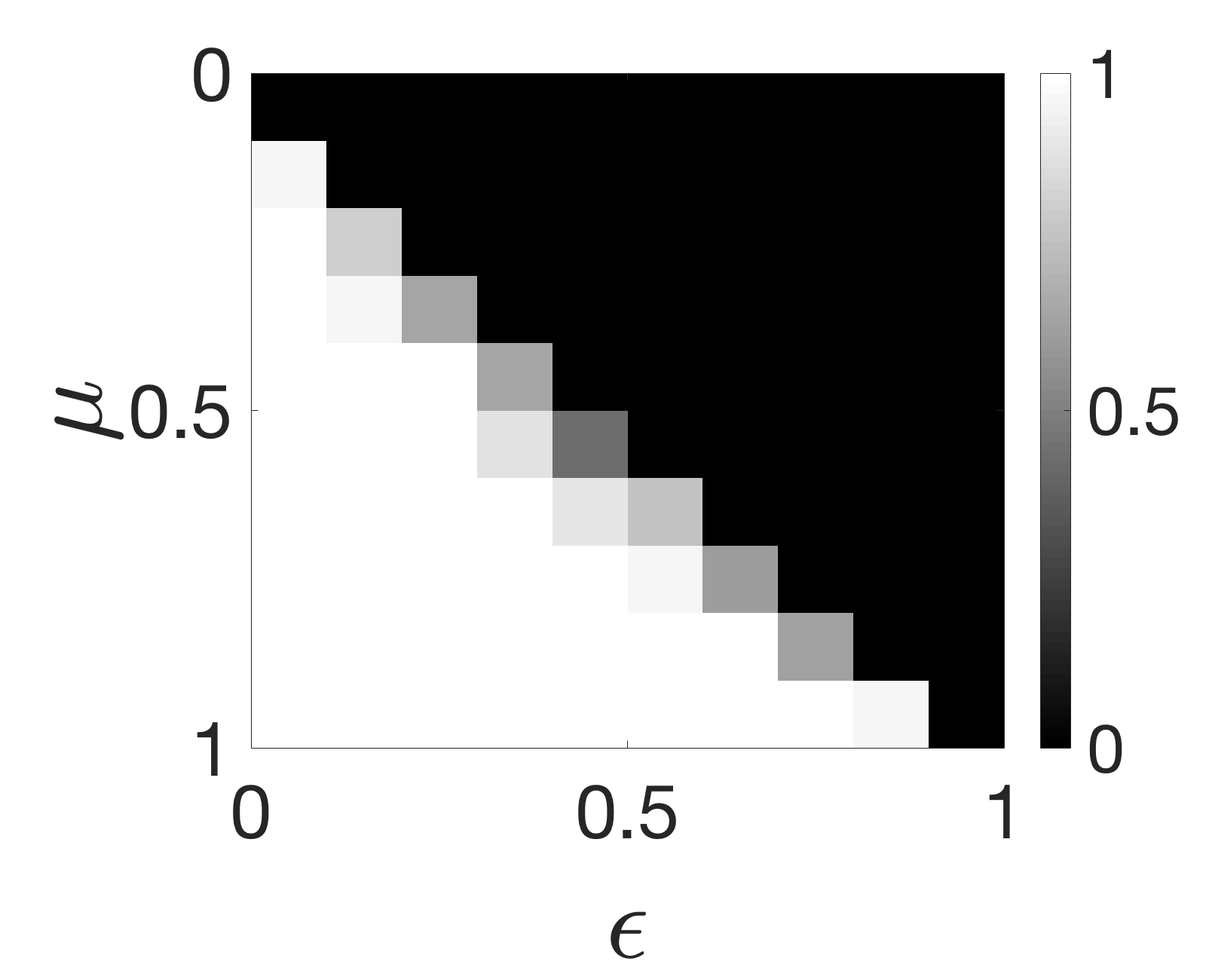}
    \\
   {\small (g) CGSS (small).}  &    {\small (h) CGSS (large).} \\
\end{tabular}
  \end{center}
   \caption{\label{fig:simulation_phase} Comparison of graph wavelet statistic and graph scan statistic on the simulated dataset. The left column shows the results for a small activated region and  the right column shows the results for a large activated region. All methods perform better when an activated region is larger. For a same activated region, all methods perform better when $\mu-\epsilon$ is larger.   }
\end{figure}

\begin{figure*}[htb]
  \begin{center}
    \begin{tabular}{ccccc}
   \includegraphics[width=0.35\columnwidth]{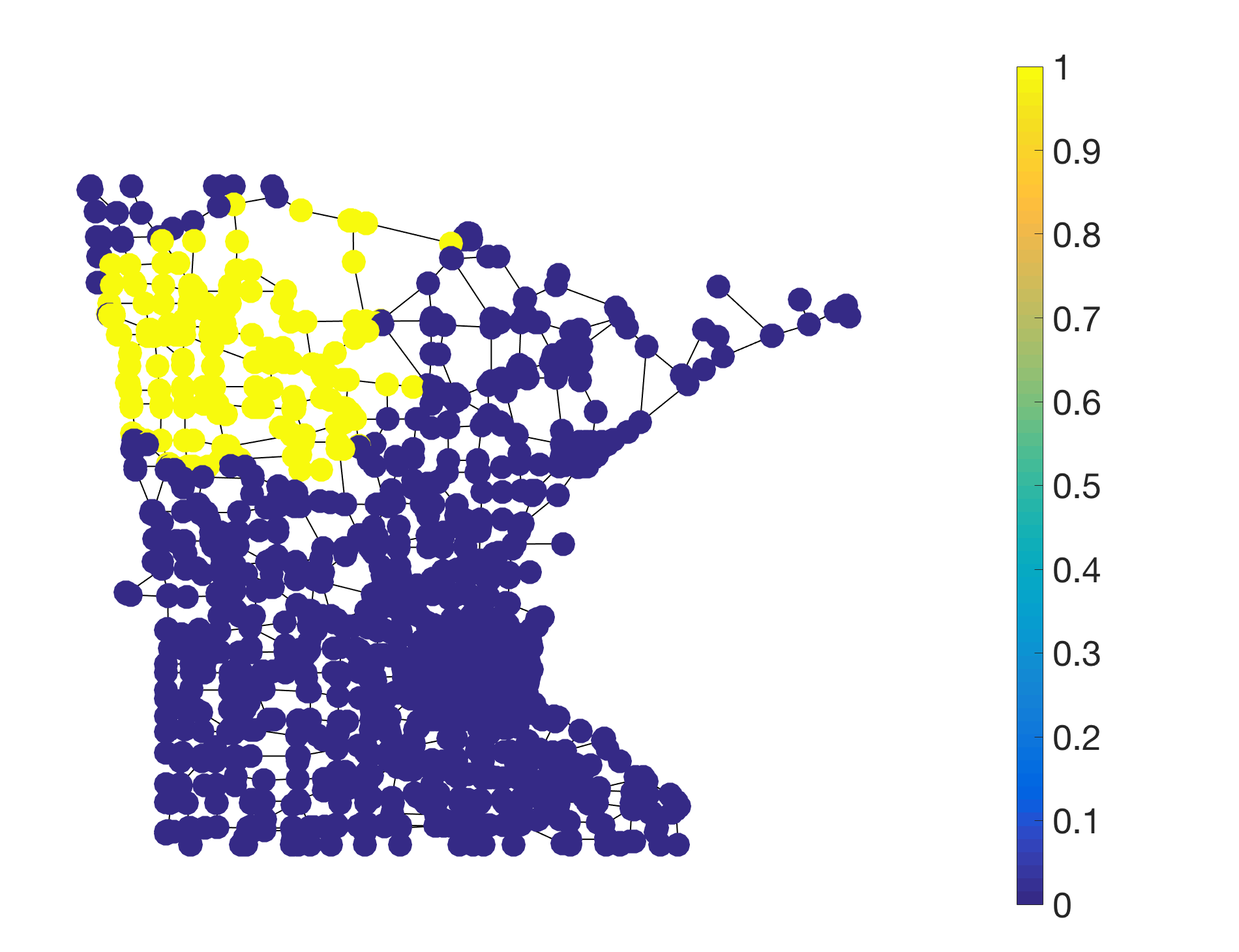} 
   &
     \includegraphics[width=0.35\columnwidth]{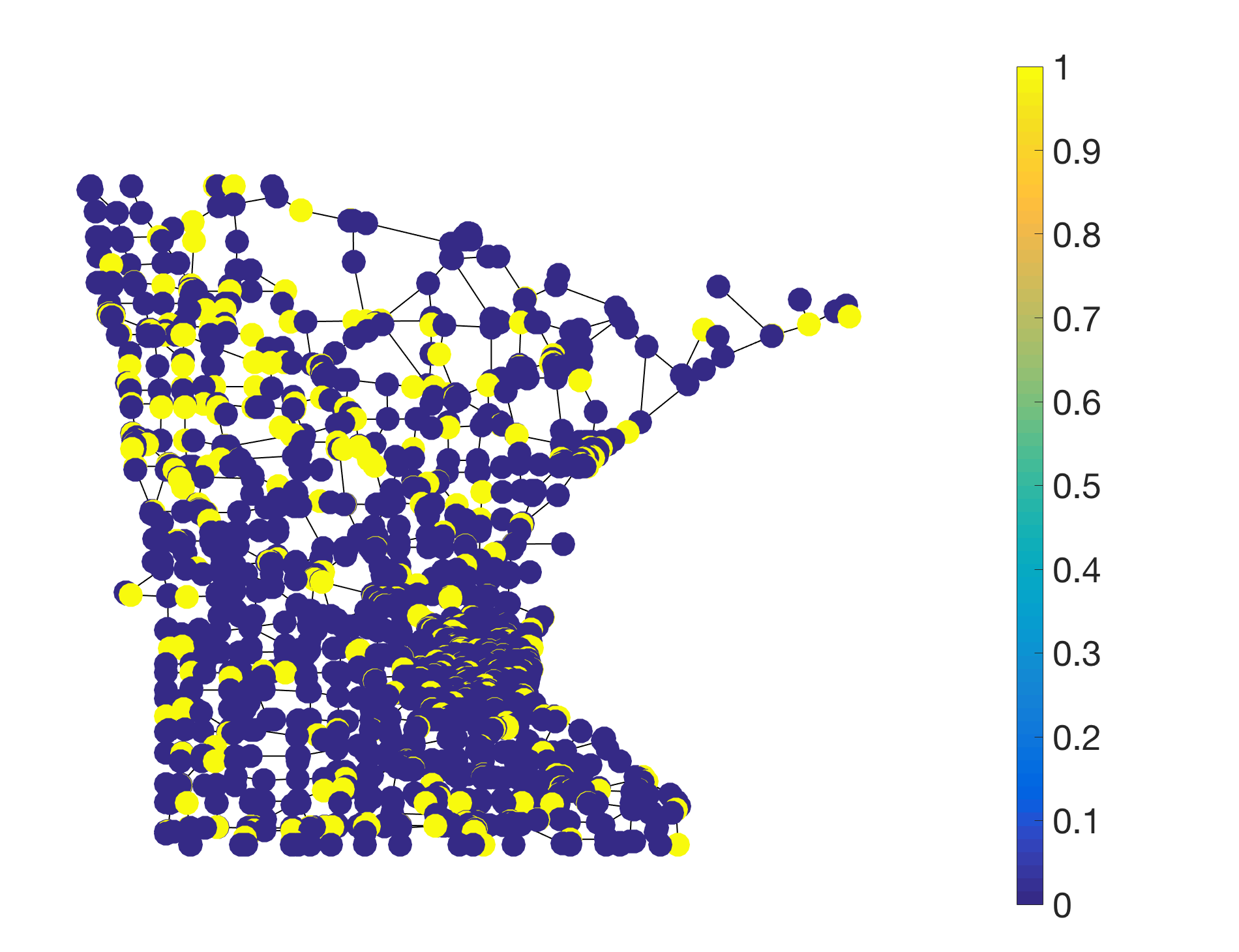} 
   &
 \includegraphics[width=0.35\columnwidth]{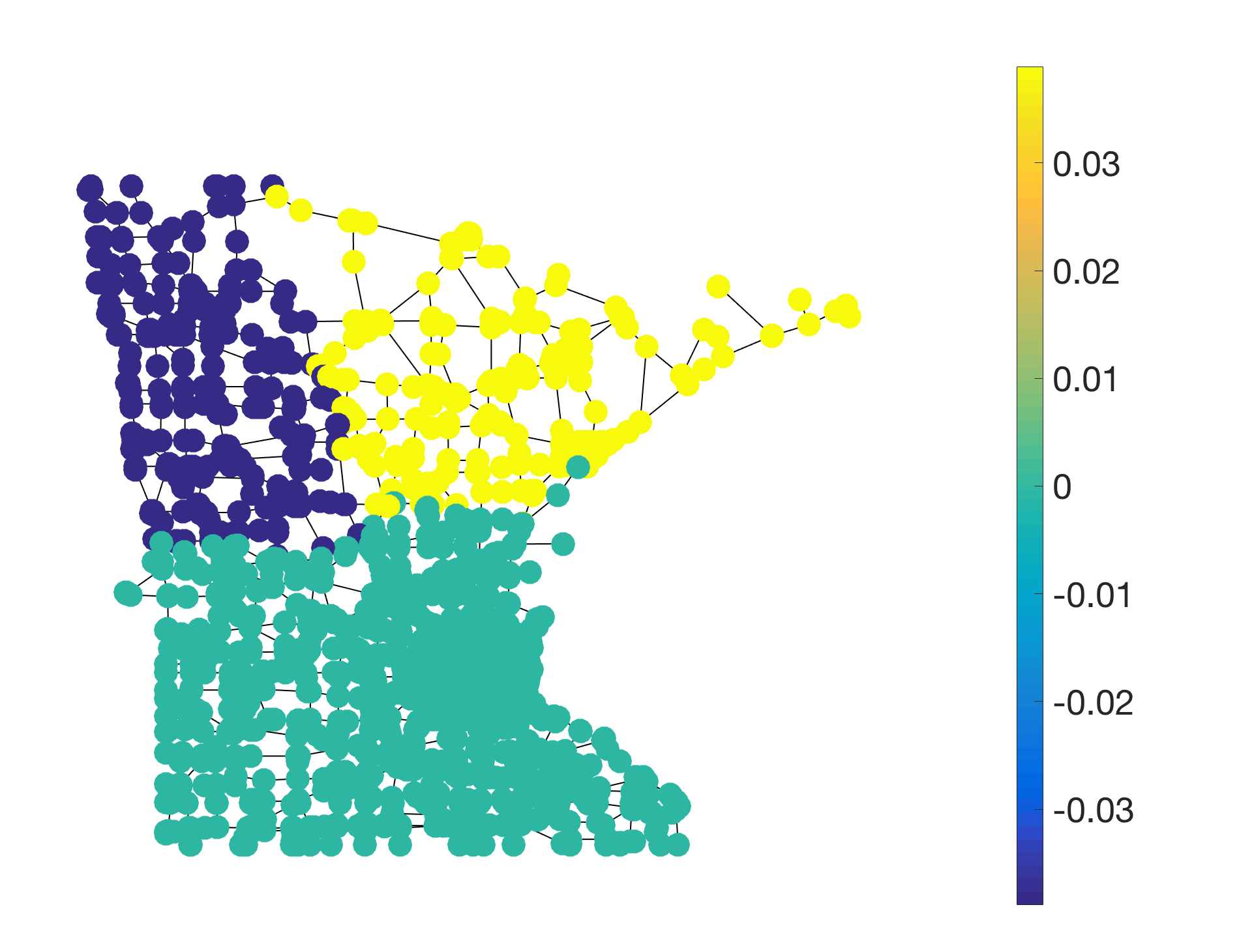} 
   &
    \includegraphics[width=0.35\columnwidth]{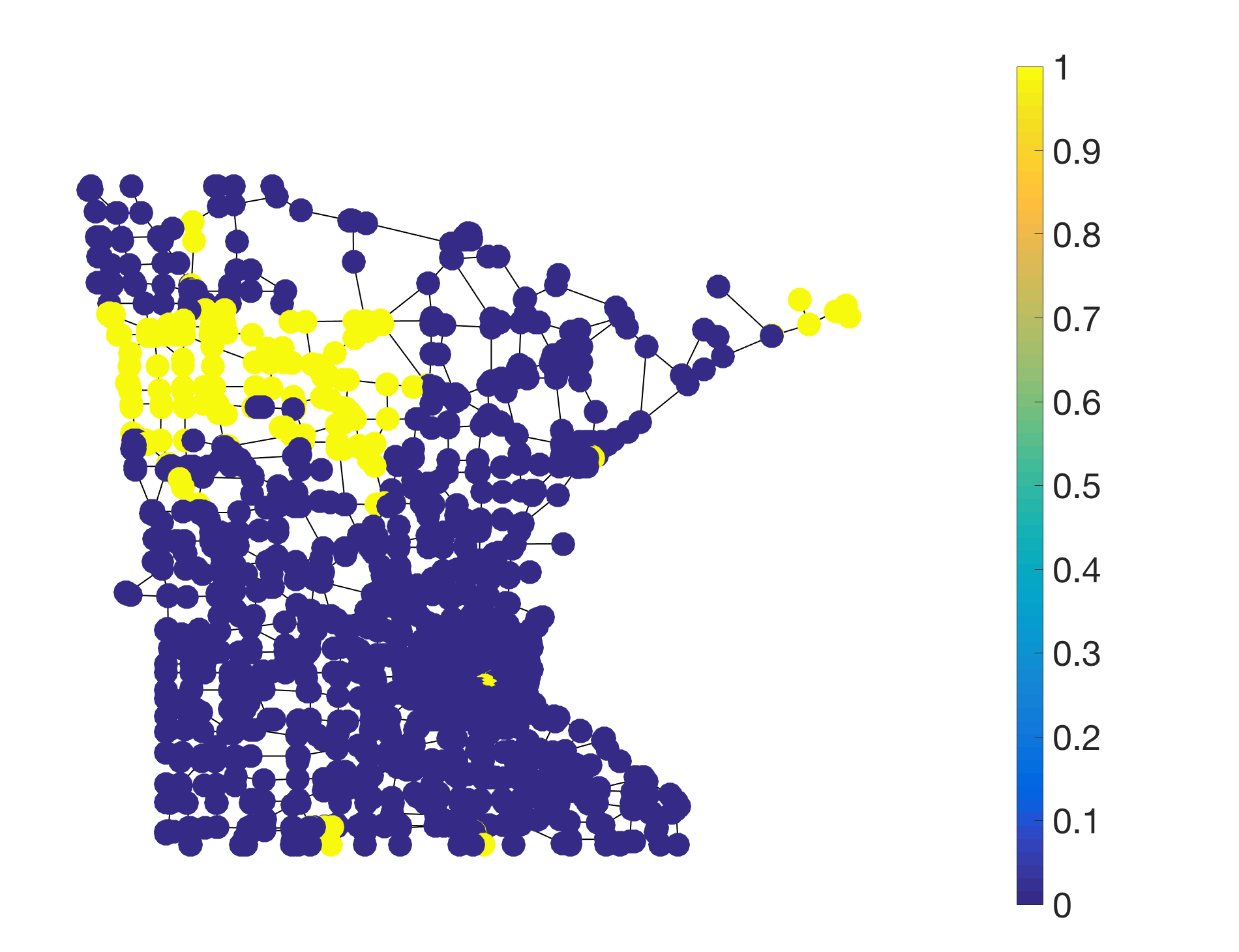} 
   &
 \includegraphics[width=0.35\columnwidth]{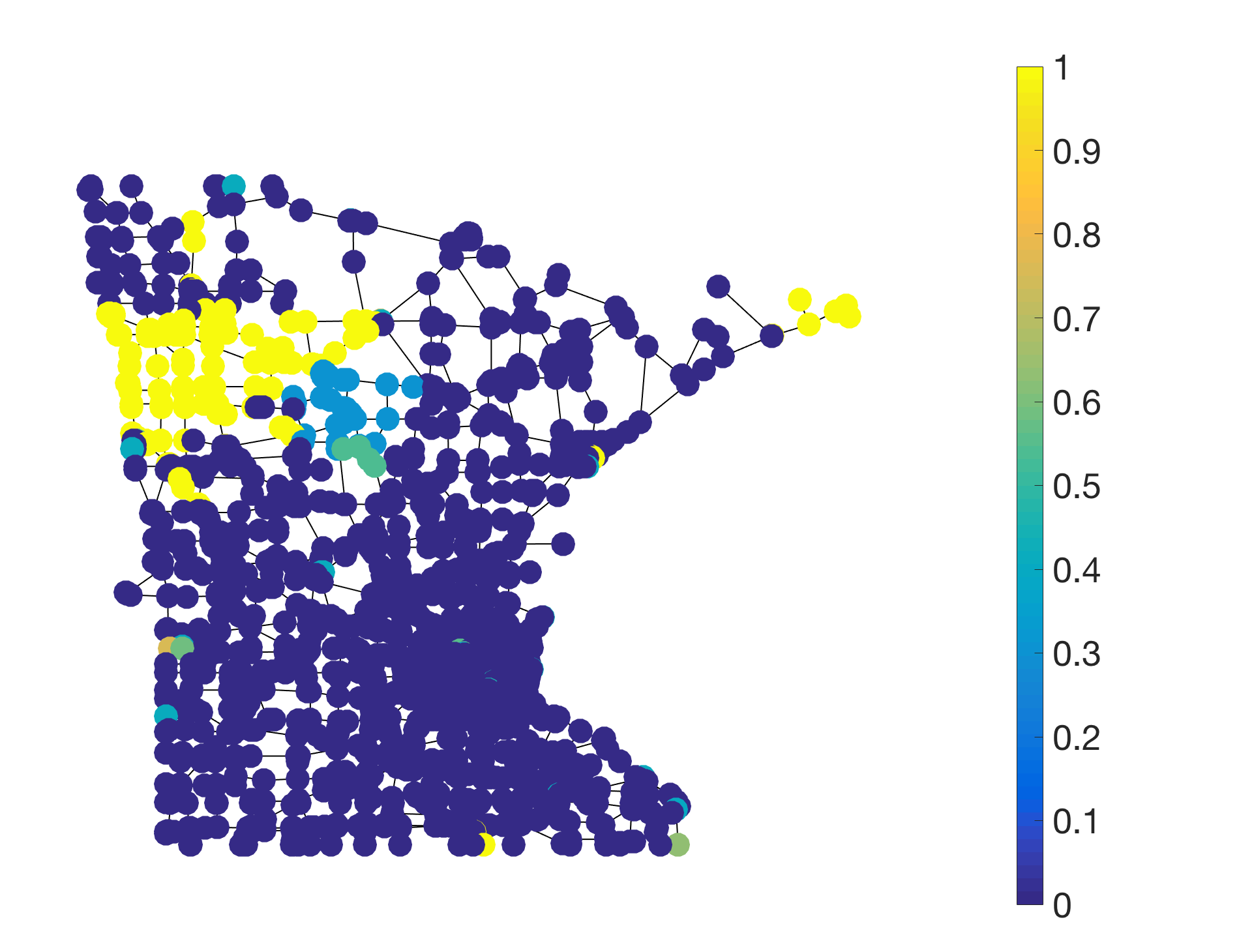} 
   \\
  {\small (a) Original attribute.}  &   {\small (b) Attribute under $H_1^N$.}   &
 {\small (d) Activated basis vector}  &    
    {\small (f) Activated region}  &   {\small (h) Activated region }   
    \\
    &  &  (wavelet) under $H_1^N$. & (LGSS) under $H_1^N$.  &  (CGSS) under $H_1^N$.
  \\
   &
     \includegraphics[width=0.35\columnwidth]{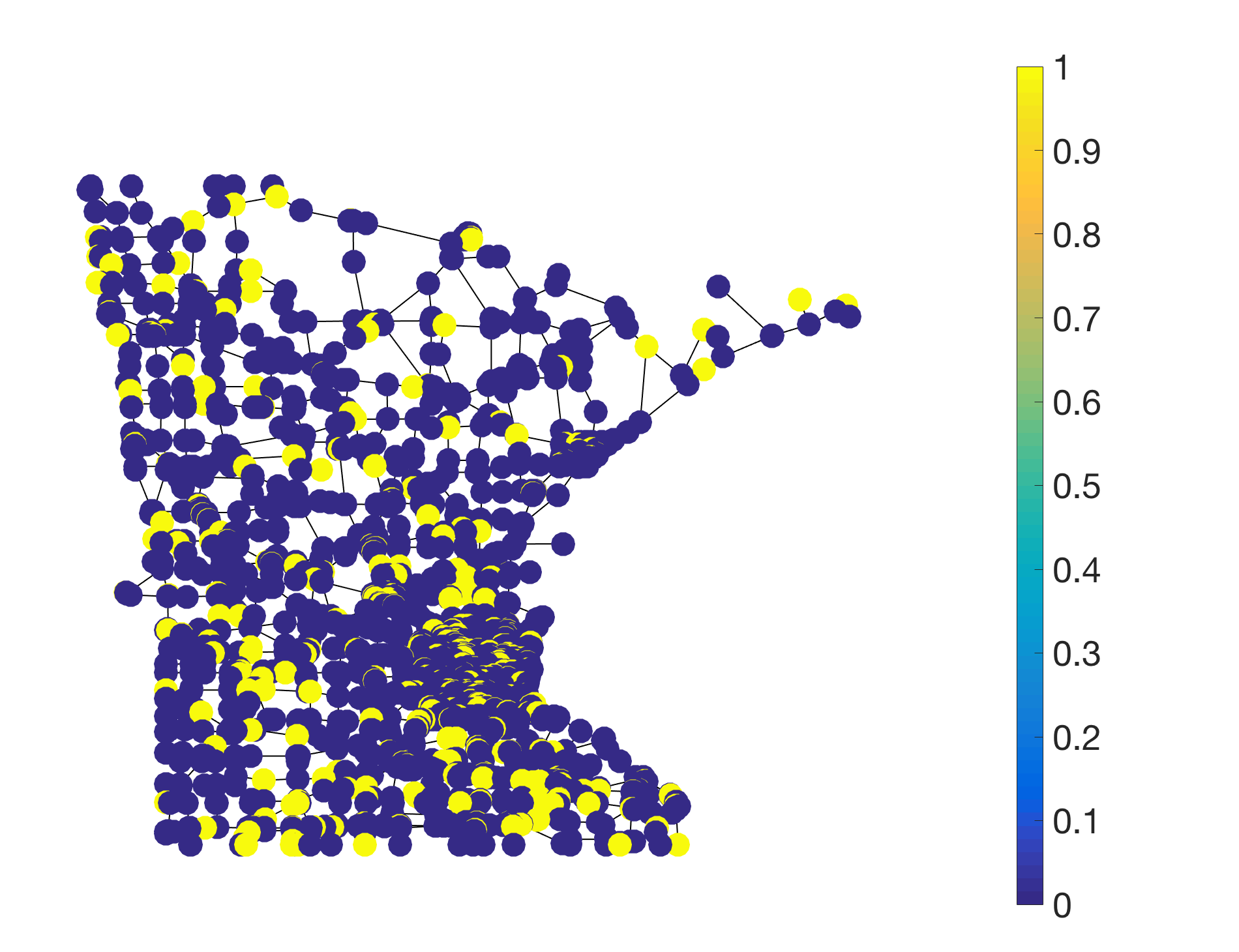} 
   &
 \includegraphics[width=0.35\columnwidth]{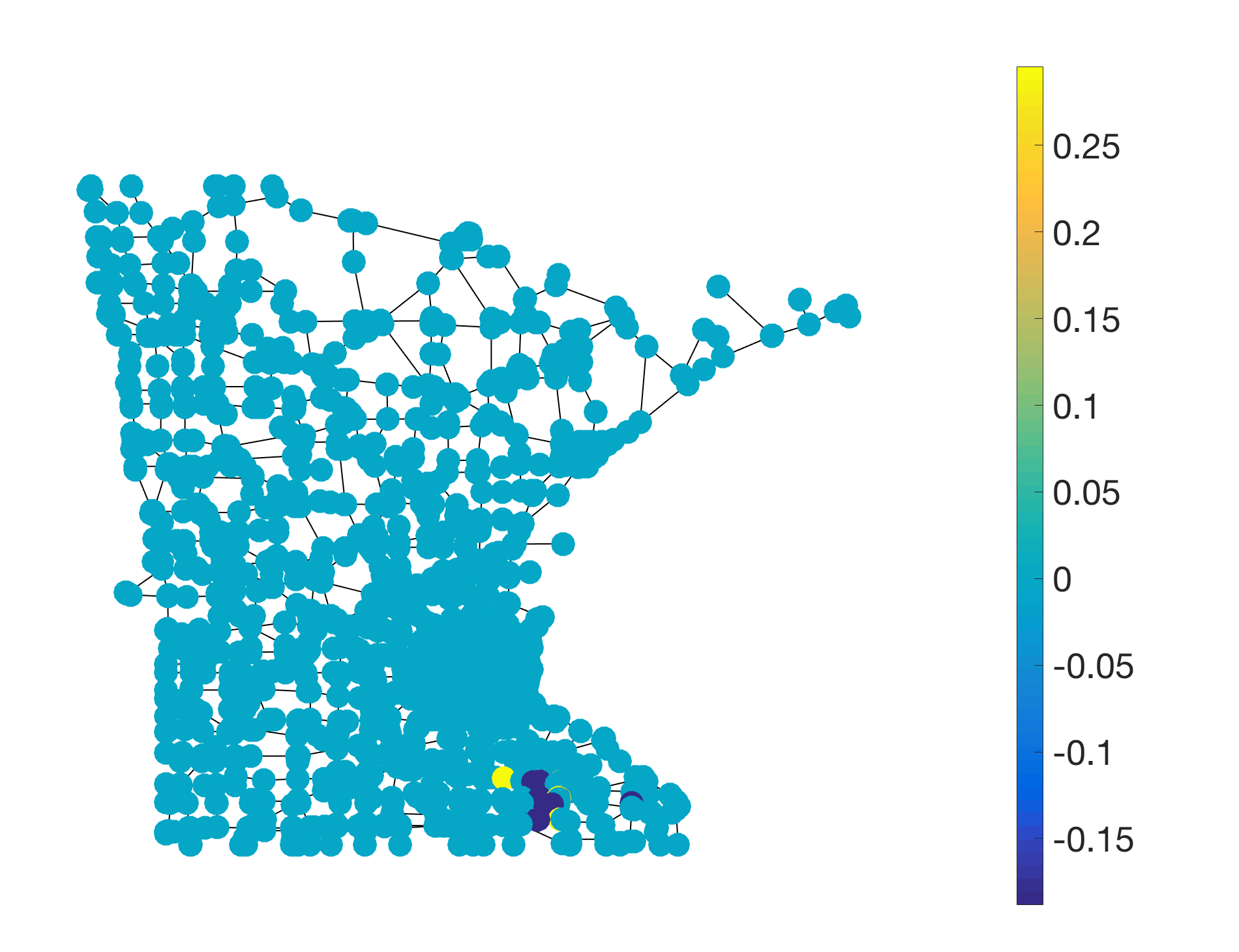} 
   &
    \includegraphics[width=0.35\columnwidth]{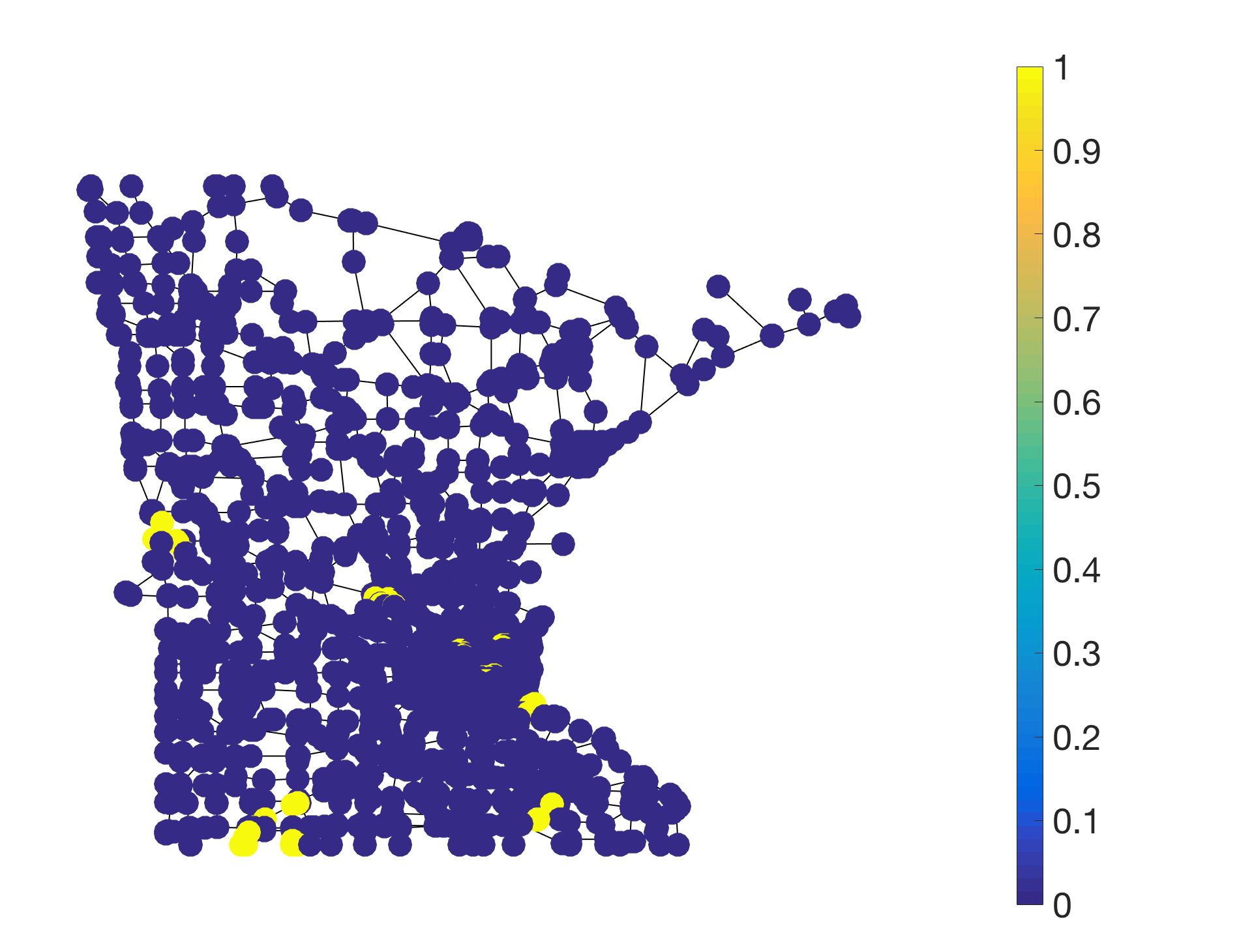} 
   &
 \includegraphics[width=0.35\columnwidth]{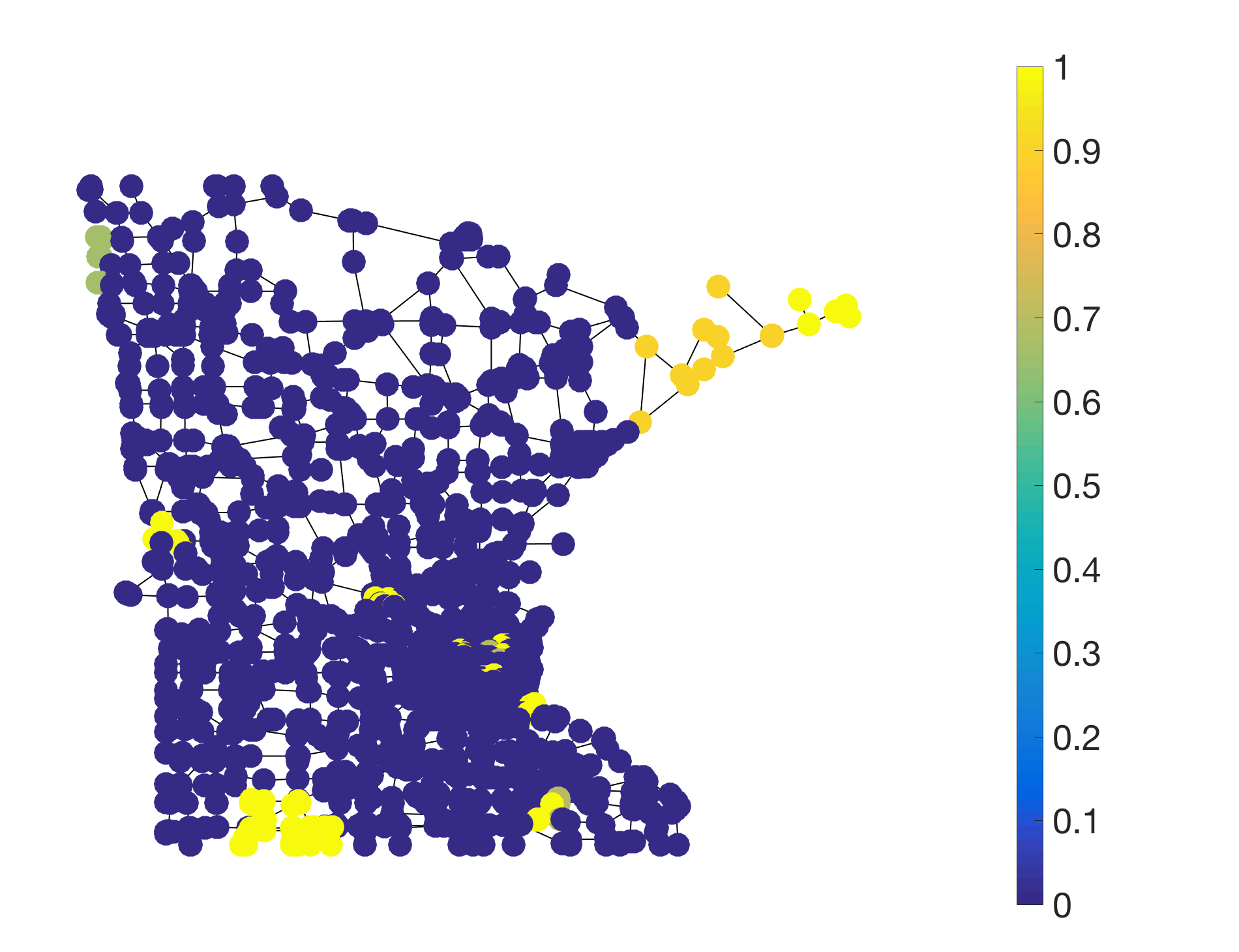} 
   \\
  &   {\small (c) Attribute under $H_0^N$.}   &
 {\small (e) Activated basis vector}  &    
    {\small (g) Activated region}  &   {\small (i) Activated region }   
    \\
    &  &  (wavelet) under $H_0^N$. & (LGSS) under $H_0^N$.  &  (CGSS) under $H_0^N$.
  \\
\end{tabular}
\end{center}
   \caption{\label{fig:example} Illustration of how the proposed statistics work. Under $H_1^N$, the graph wavelet statistic, local graph scan statistic and convex graph scan statistic denoise the given attribute and localize the true community. Under $H_0^N$,  the graph wavelet basis, local graph scan statistic and convex graph scan statistic cannot localize the true community. The denoising procedure is the key to robustness. Graph wavelet statistic extracts features from original attributes and is a discriminative approach. (d) shows a graph wavelet basis vector corresponding to the maximum absolute value of the graph wavelet coefficients. 
 Graph scan statistic recovers denoised attributes and is a generative approach. (f) and (h) show the activated regions recovered by graph scan statistics.  For CGSS, due to the convex relaxation, the recovered activated region is not binary.  A higher value of $x_i$ indicates a higher confidence that the $i$th node is activated.
 }
 \vspace{-3mm}
\end{figure*}

\subsection{Simulation Results}
We generate simulated data on the Minnesota road
graph~\cite{MinnesotaGraph} and study how the parameters, including
the activation probability inside the locaized pattern $\mu$, the noise level $\epsilon$ and the
activation size $|C|$, influence the detection performance. The
Minnesota road graph is a standard dataset including 2642 nodes and
3304 undirected edges~\cite{MinnesotaGraph}. We generate two binary
graph signals as follows: we randomly choose one node as a cluster head
and assign all other nodes that are within $k$ steps to the cluster head to an activated node set, where $k$ varies from $6$ to
$12$. Figures~\ref{fig:simulation_phase} (a) and (b) show these two
binary graph signals, where the nodes in yellow indicates the
activated nodes and the nodes in blue indicates the nonactivated
nodes. Using these two binary graph signals as templates, we then
generate two classes of random graph signals: for attributes under
$H_1^N$, each node inside the activated region is activated with
probability $\mu$ and each node outside the activated region is
activated with probability $\epsilon$. Both $\mu$ and $\epsilon$ vary
from $0.05$ to $0.95$ with interval of $0.1$. For each combination of
$\mu$ and $\epsilon$, we generate corresponding attributes under
$H_0^N$, the activation probability for each node is $(\mu
|C| + \epsilon (N- |C|))/N$. We run 100 random tests to compute the
statistics and quantify the performance by the area under the receiver
operating characteristic curve (AUC)~\cite{Bishop:06}. 


Figures~\ref{fig:simulation_phase} (c), (e) and (g) show AUCs of the graph wavelet statistic, the local graph scan statistic (LGSS) and the convex graph scan statistic (CGSS) for the small activated region, where the step $k=6$. For example, each block in Figure~\ref{fig:simulation_phase} (c) corresponds to the AUC of the graph wavelet statistic given a pair of $\mu$ and $\epsilon$. A whiter block indicates a higher AUC and a better performance. Note that when $\mu$ is smaller than $\epsilon$, we did not run the experiments and directly set the corresponding AUC to zero. We see that the graph wavelet statistic has a similar performance with the convex graph scan statistic and both outperform the local graph scan statistic. Figures~\ref{fig:simulation_phase} (d), (f) and (h) show AUCs of the graph wavelet statistic, the local graph scan statistic and the convex graph scan statistic for the large activated region, where the step $k=12$. We see that the convex graph scan statistic perform the best and the graph wavelet statistic has a slightly better performance than the local graph scan statistic.  Comparing the results from two activated regions (left column versus right column in Figure~\ref{fig:simulation_phase}),  we see that all the methods perform better when the activated region is large. For example, both Figures~\ref{fig:simulation_phase}  (c) and (d) use graph wavelet statistic. Given a fixed pair of $\mu$ and $\epsilon$, a large activated region has a larger AUC, indicating higher probability to be detection.

To have a clearer understanding of how the proposed statistics work, we set the activation probability $\mu = 0.35$ and the noise level $\epsilon = 0.15$.  Figures~\ref{fig:example} (b) and (c) show the attribute under $H_1^N$ and $H_0^N$ given the ground-truth activated region in Figure~\ref{fig:example} (a).  When we compare the attributes under $H_1^N$ and $H_0^N$, it is clear that distinguishing $H_1^N$ from $H_0^N$ is not trivial.


 Figures~\ref{fig:example} (d), (f) and (h) compare the activated regions detected by graph wavelet basis, local graph scan statistic and convex graph scan statistic under $H_1^N$. The graph wavelet basis compares the average values between the nodes in yellow and the nodes in blue. When the difference is large, the activated region is detected. Ideally, we want all the nodes in blue are activated and all the nodes in yellow are nonactivated. Considering the graph wavelet basis is designed before obtaining any data, it captures the activated region fairly well. As expected, local graph scan statistic and convex graph scan statistic perform similarly and capture the activated region well. We also show the noisy attribute and the activated regions detected by graph wavelet basis, local graph scan statistic and convex graph scan statistic under $H_0^N$ in Figures~\ref{fig:example} (c), (e), (g) and (i). The graph wavelet basis, local graph scan statistic and convex graph scan statistic cannot detect regions that are close to the true activated region from the pure noisy attribute.

\begin{table}[htbp]
  \footnotesize
  \begin{center}
    \begin{tabular}{@{}lll@{}}
      \toprule
   &  Attribute under $H_0^N$  &  Attribute under $H_1^N$   \\
   &  Figure~\ref{fig:example} (b)  &  Figure~\ref{fig:example} (c) \\
         \midrule \addlinespace[1mm]
Activated nodes &  $\phantom{+}422$ &   $\phantom{+}420$  \\
Modularity &  $\phantom{+}9.1228$  &   $\phantom{+}1.1422$  \\
Cut cost &  $\phantom{+}845$  &  $\phantom{+}887$  \\
Wavelet &  $\phantom{+}1.50$  &  $\phantom{+}2.14$  \\
LGSS &  $\phantom{+}57.58$  &  $\phantom{+}68.12$  \\
CGSS &  $\phantom{+} 42.97$  &  $\phantom{+}72.33$  \\
\bottomrule
\end{tabular} 
\caption{\label{tab:example}  Facts about the data in Figures~\ref{fig:example} (b) and (c).}
\vspace{-3mm}
\end{center}
\end{table}

Table~\ref{tab:example} shows some facts about data in Figures~\ref{fig:example} (b) and (c), including modularity, cut cost, graph wavelet statistic, graph scan statistic and convex graph scan statistic. Modularity is a popular metric to measure the strength of communities~\cite{Newman:10}. Networks with high modularity have dense connections within communities but sparse connections in different communities. Mathematically, the modularity of a binary attribute $\one_C \in \R^N$ is\footnote{We drop a constant factor $M$ here.}
\begin{equation*}
{\rm Modularity} \ = \ \sum_{i,j} \left(  \Adj_{i,j} - \frac{d_i d_j}{M} \right) (\one_{C})_i (\one_{C})_j,
\end{equation*}
where $d_i$ is the degree of the $i$th node, $M = \sum_i d_i$ is the total number of edges, and $(\one_{C})_i = 1$ when the $i$th node is activated; otherwise, $(\one_{C})_i = 0$. A large modularity means the activated nodes are strongly connected. 

Graph cuts measure the cost to separate a community from the other nodes. Mathematically, the cut cost of the binary attribute $\one_C \in \R^N$ is
$
{\rm Cut} \ = \ \TV_0 (\one_{C} ).
$
A small cut cost means the activated nodes are easily separated from the nonactivated nodes. 

We expect that under $H_1^N$, modularity is larger, indicating dense internal connections, and the cut cost is smaller, indicating few external connections. From Table~\ref{tab:example}, however, we see that attributes under $H_0^N$ and $H_1^N$ contain similar number of activated nodes, the modularity under $H_1^N$ is smaller than the modularity under $H_0^N$, indicating the activated nodes under $H_0$ have even stronger internal connections.  The cut cost under $H_0^N$ is smaller than the cut cost under $H_1^N$, indicating the activated nodes under $H_0^N$ are easier to be separated from the nonactivated nodes. It is clear that both modularity and number of cuts fail when the noise level is high. On the other hand, the graph wavelet statistic, local graph scan statistic and convex graph scan statistic under $H_1^N$ are much higher than those under $H_0^N$, indicating these three proposed statistics succeed even when the noise level is high. Graph wavelet statistic is robust because it selects a useful feature by using the graph wavelet basis. The graph scan statistic also is robust is because it localizes the true activated region first, which is equivalent to denoise the attribute based on the graph structure. Based on the denoised attribute, we compute the statistic values and the results are more robust. In other words, graph wavelet statistic extracts features from original attributes and is a discriminative approach to detect and  graph scan statistic recovers a denoised attributes and is a generative approach.

In terms of the computational complexity, for each random test, it takes around $30$ seconds to construct the graph wavelet basis, around $0.01$ seconds to calculate the graph wavelet statistic, around $5$ seconds to calculate the local graph scan statistic, around $10$ seconds to calculate the convex graph scan statistic. Overall, the proposed statistics provide efficient and effective performances. 

\begin{figure*}[p]
  \begin{center}
    \begin{tabular}{cccc}
      \includegraphics[width=0.35\columnwidth]{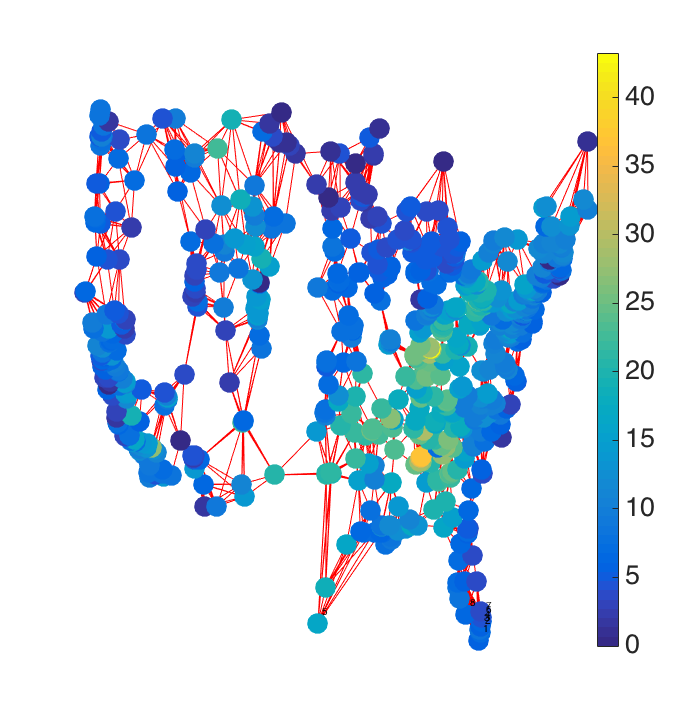} 
      &
      \includegraphics[width=0.35\columnwidth]{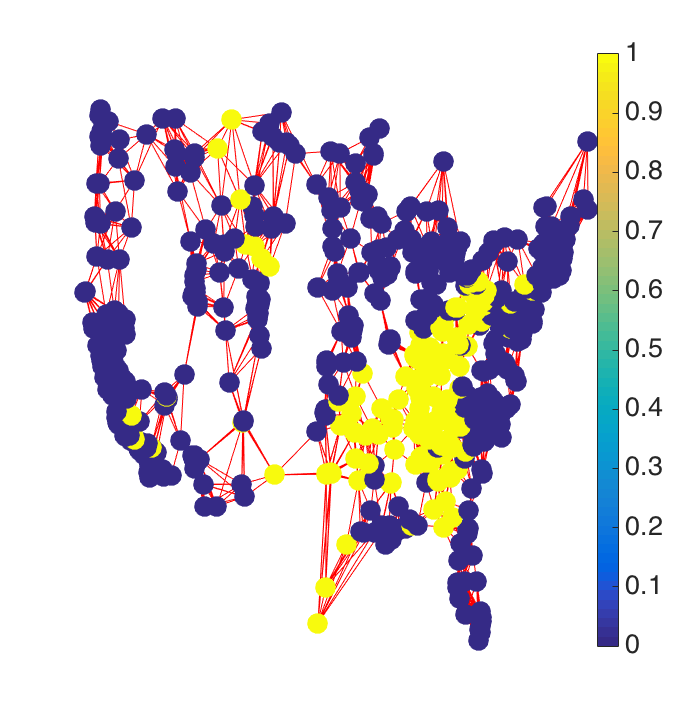} 
      &
      \includegraphics[width=0.35\columnwidth]{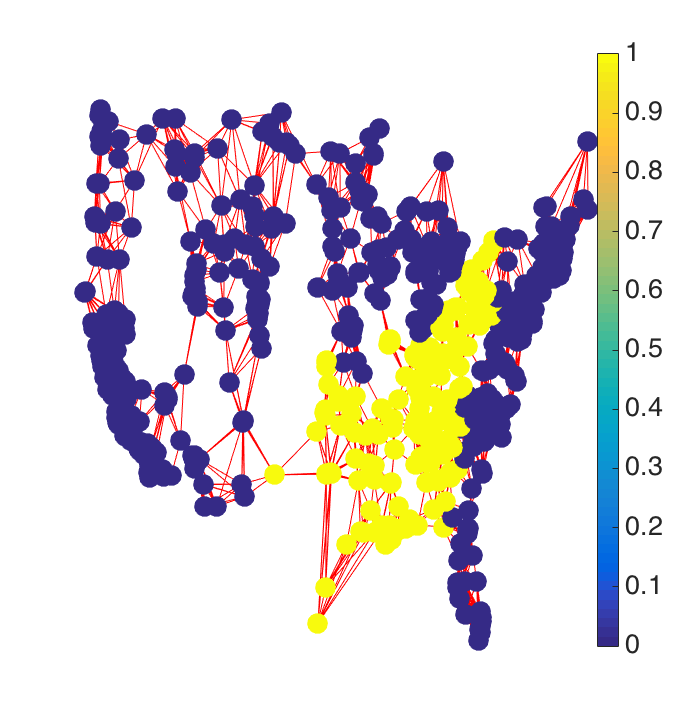} 
      &
      \includegraphics[width=0.35\columnwidth]{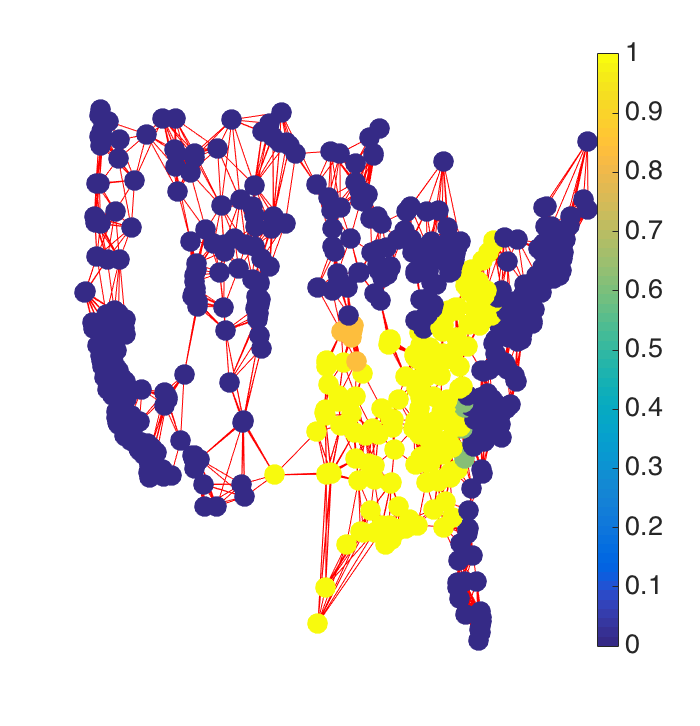} 
      \\
      {\small (a) Original data.}  &   {\small (b) Input attribute ($>15$).}   &
      {\small (c) Localized attribute detected by LGSS.}  &   {\small (d) Localized attribute detected by CGSS.}
      \\
      \\
      &
      \includegraphics[width=0.4\columnwidth]{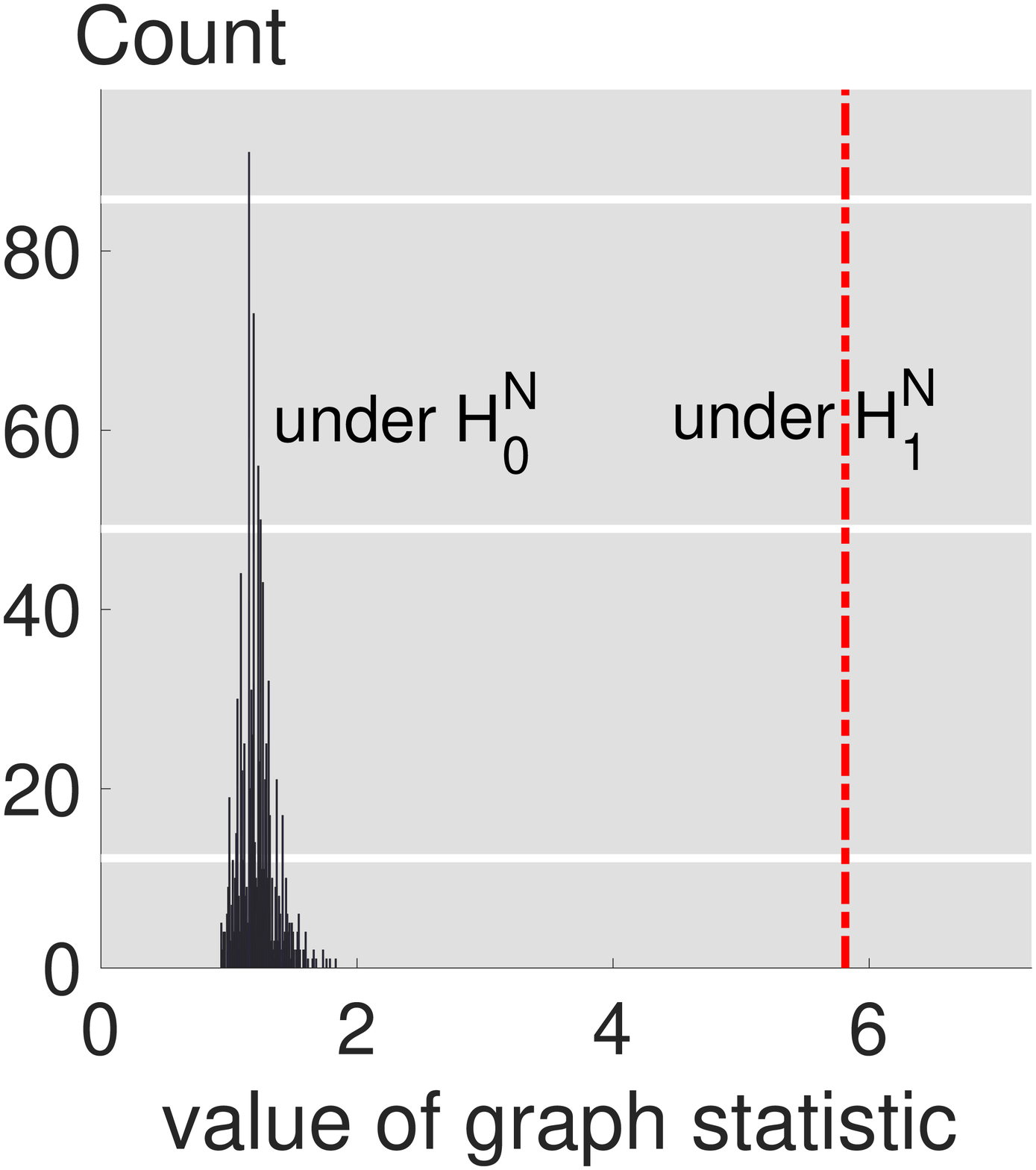}
      &
      \includegraphics[width=0.4\columnwidth]{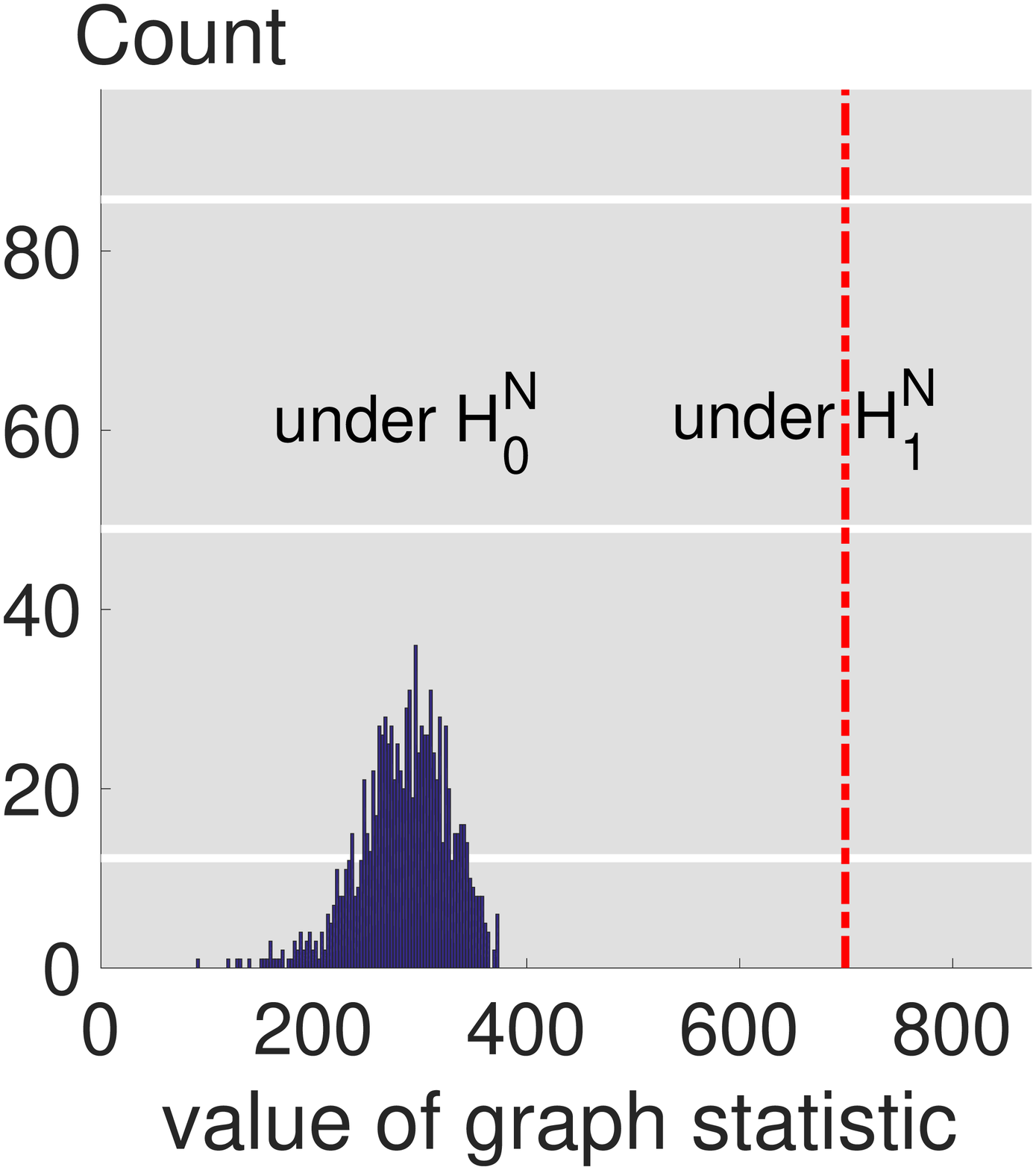}
      &
      \includegraphics[width=0.4\columnwidth]{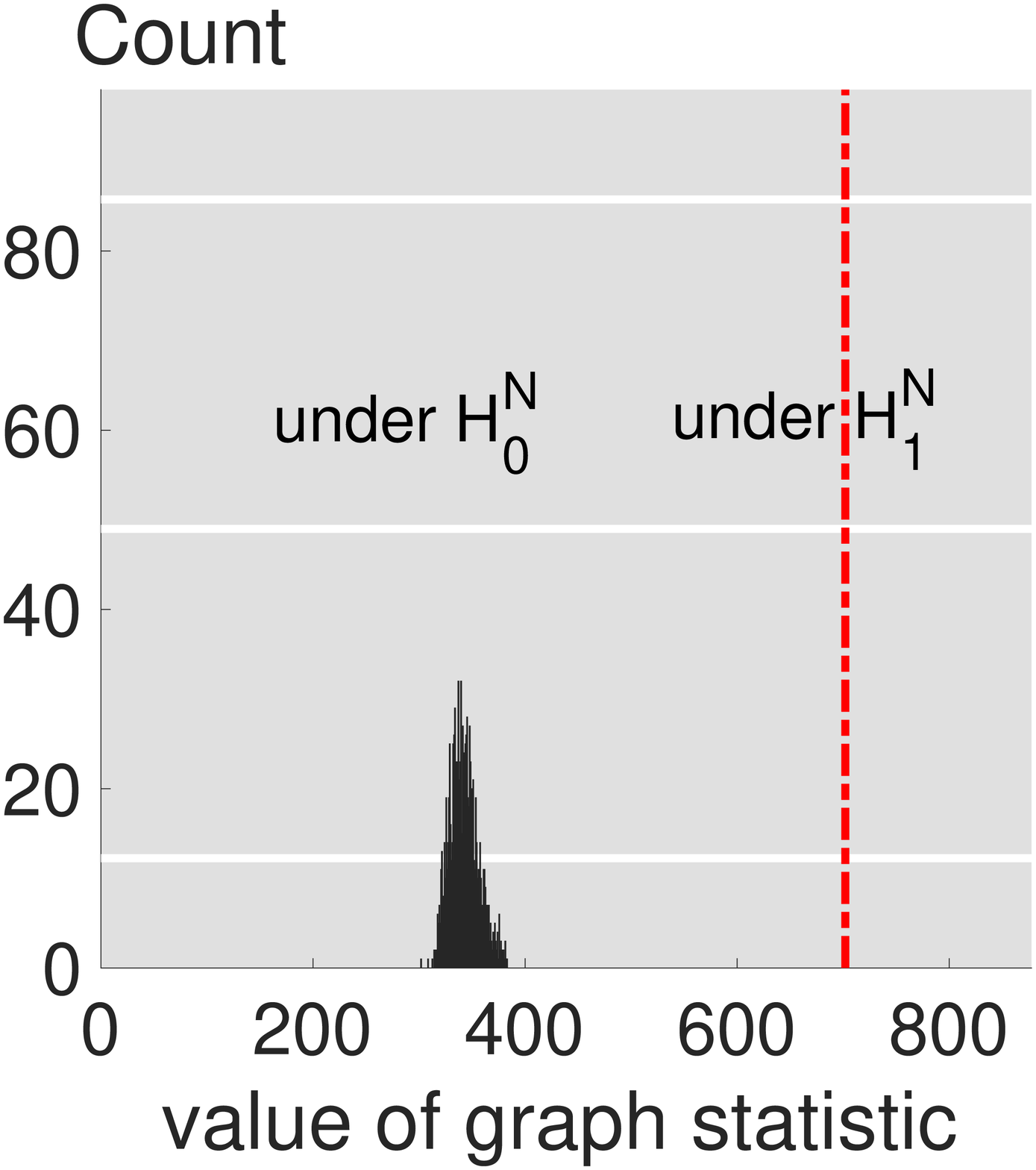}
      \\
      & {\small (e) Graph wavelet statistic.}  &  {\small (f) LGSS.}  &
      {\small (g) CGSS.}  
      \\
\end{tabular}
  \end{center}
  \caption{\label{fig:air_detection_july} Detecting the high-pollution
    region on July 1st, 2014. (a) Original data. (b) High-pollution
    cities (in yellow). (c)--(d) High-pollution regions recovered by
    the graph scan statistics. (e)--(g) Detection of high-pollution regions
    from random attributes. For each plot, the red dashed line shows
    the value of the graph statistic for the real pollution graph
    signal from  (b) and the black curves show the
    empirical histograms of the graph statistics under 1,000 random
    trials. 
  }
\end{figure*}

\begin{figure*}[p]
  \begin{center}
    \begin{tabular}{cccc}
      \includegraphics[width=0.35\columnwidth]{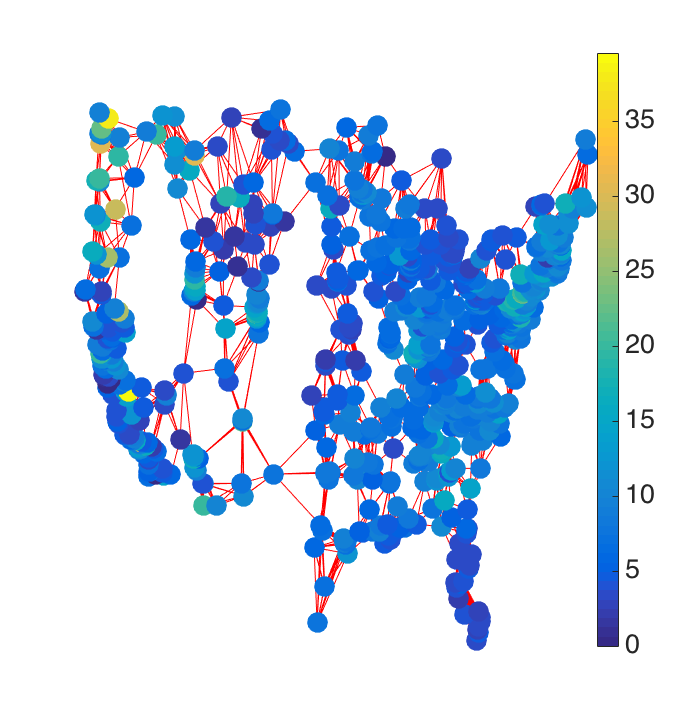} 
      &
      \includegraphics[width=0.35\columnwidth]{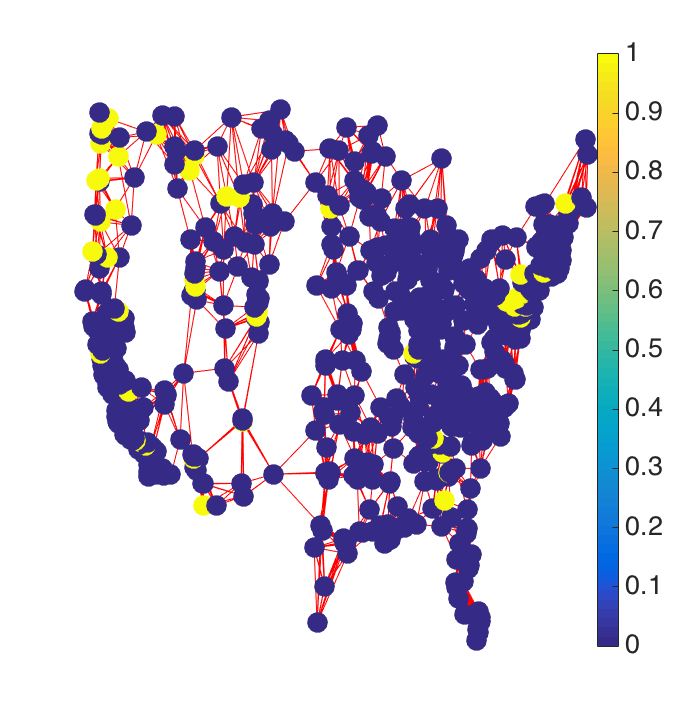} 
      &
      \includegraphics[width=0.35\columnwidth]{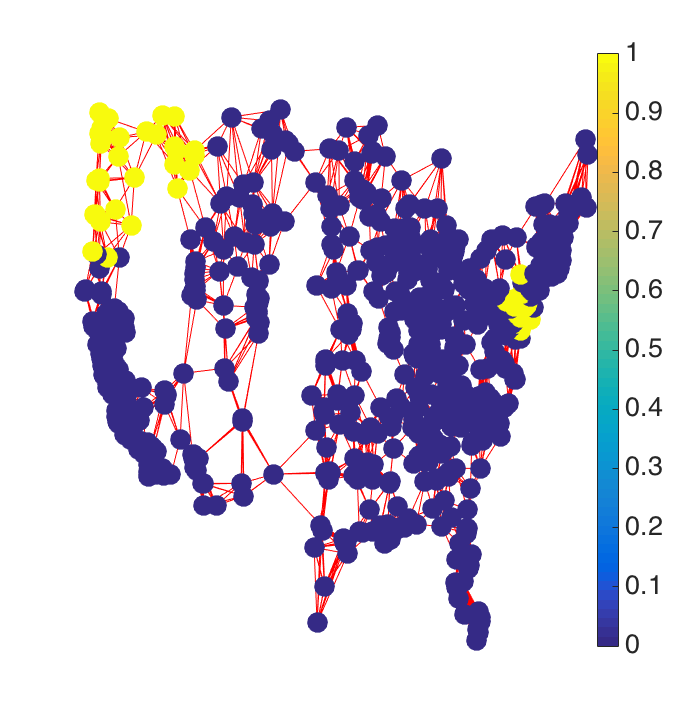} 
      &
      \includegraphics[width=0.35\columnwidth]{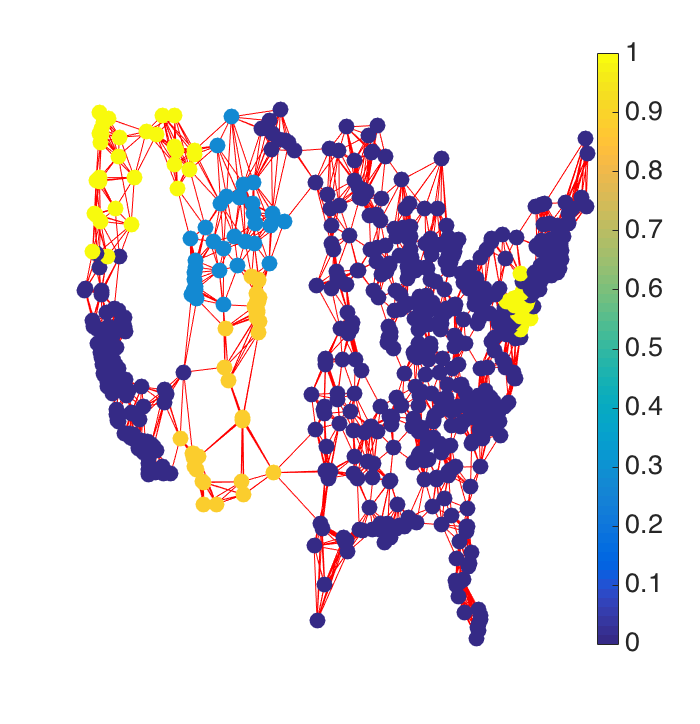} 
      \\
      {\small (a) Original data.}  &   {\small (b) Input attribute ($>15$).}   &
      {\small (c) Localized attribute detected by LGSS.}  &   {\small (d) Localized attribute detected by CGSS.}
      \\
      \\
      &
      \includegraphics[width=0.4\columnwidth]{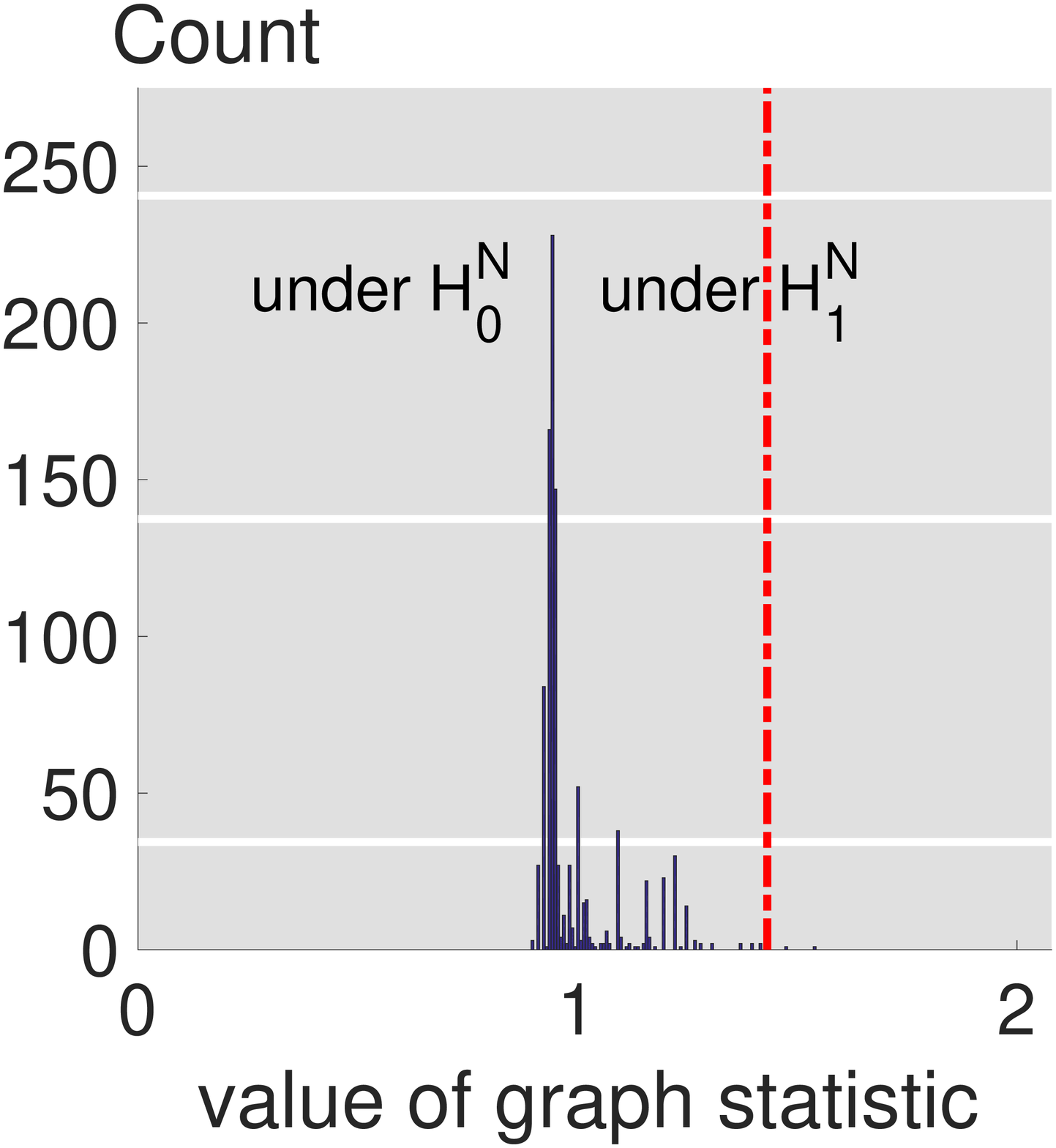}
      &
      \includegraphics[width=0.4\columnwidth]{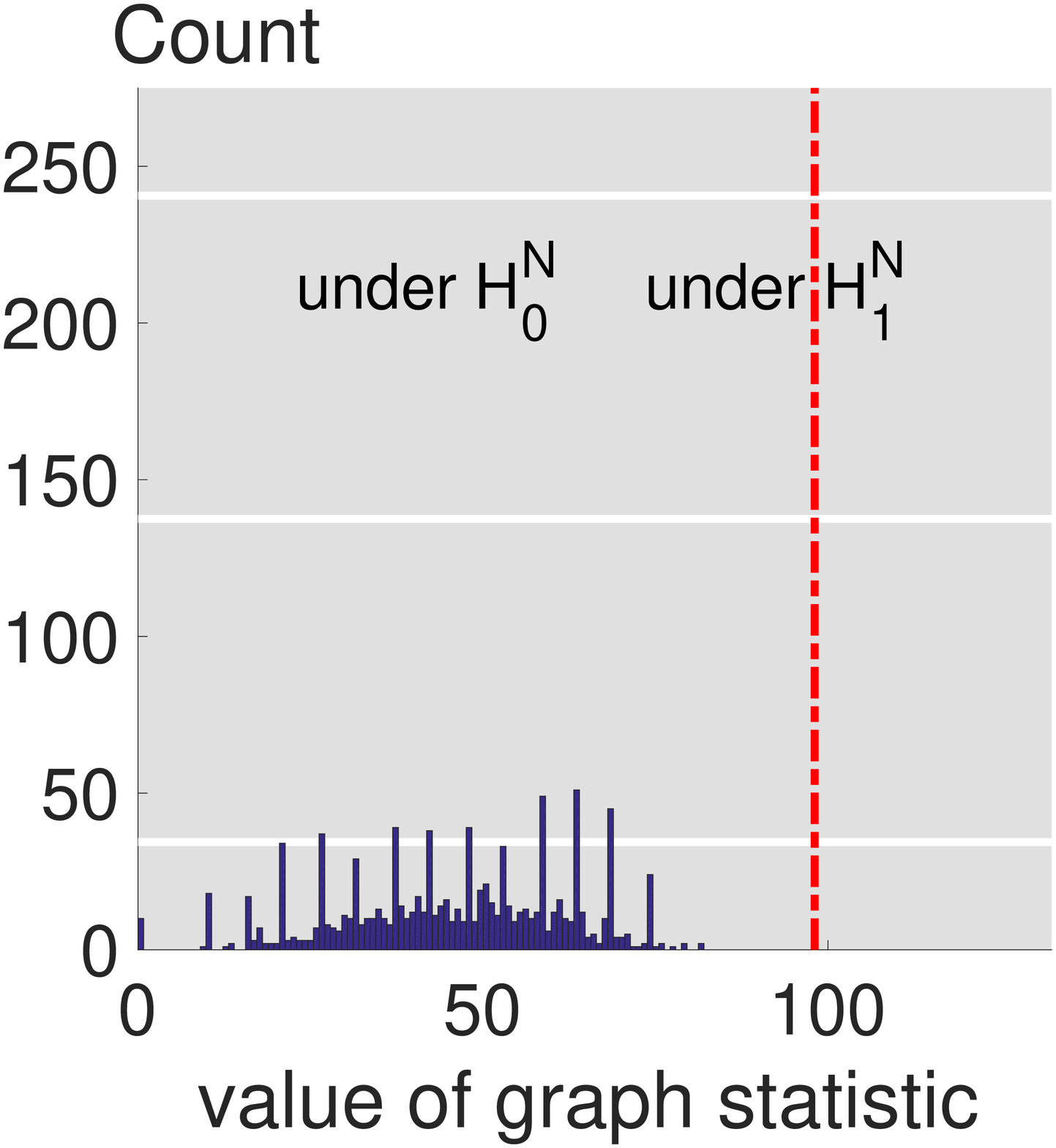}
      &
      \includegraphics[width=0.4\columnwidth]{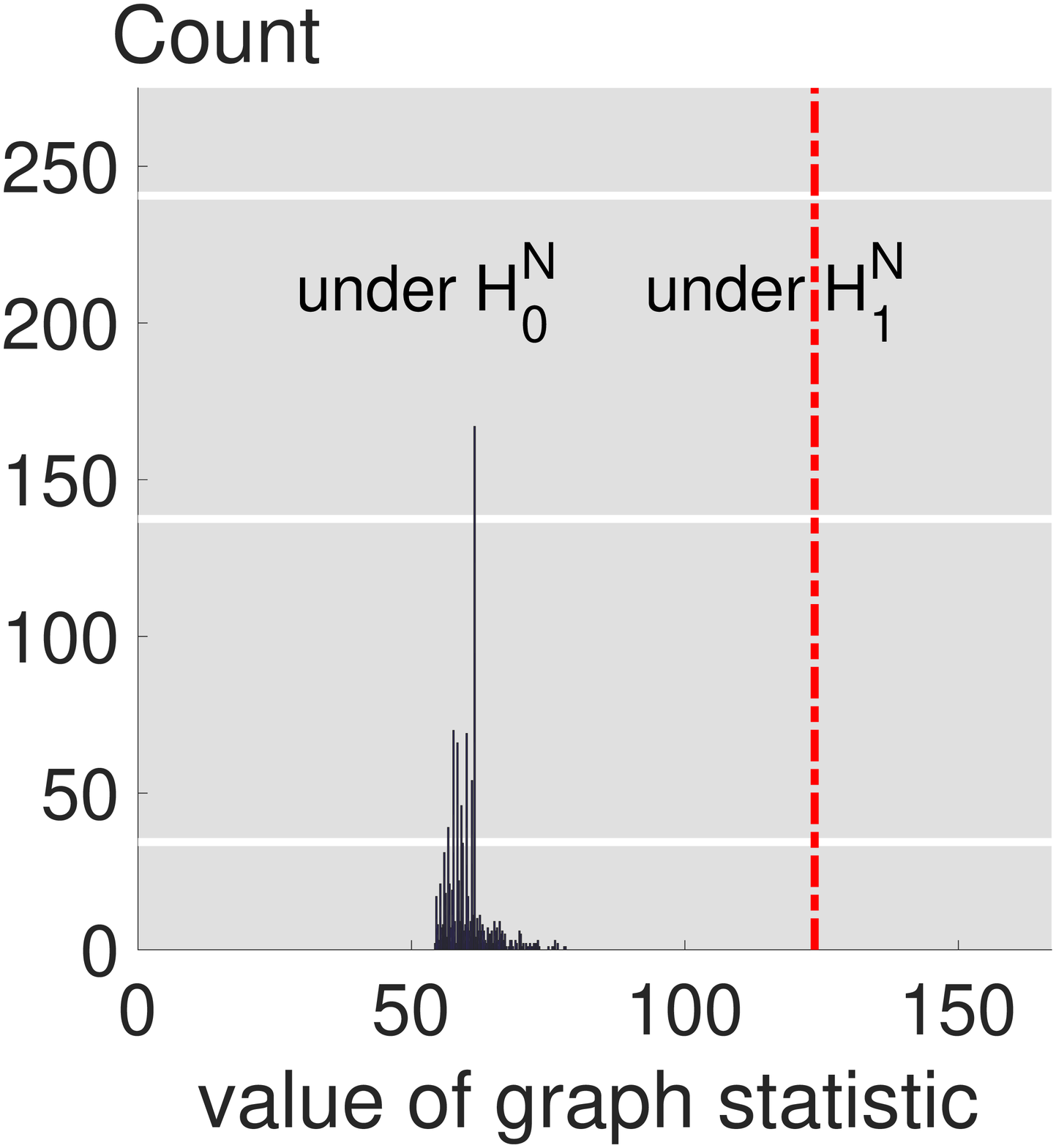}
      \\
      & {\small (e) Graph wavelet statistic.}  &  {\small (f) LGSS.}  &
      {\small (g) CGSS.} 
    \end{tabular}
  \end{center}
  \caption{\label{fig:air_detection_dec} 
Detecting the high-pollution
    region on December 1st, 2014. (a) Original data. (b) High-pollution
    cities (in yellow). (c)--(d) High-pollution regions recovered by
    the graph scan statistics. (e)--(g) Detection of high-pollution regions
    from random attributes. For each plot, the red dashed line shows
    the value of the graph statistic for the real pollution graph
    signal from  (b) and the black curves show the
    empirical histograms of the graph statistics under 1,000 random
    trials.  
  }
\end{figure*}

\subsection{High Air-Pollution Detection}
The first real-world example is on air-pollution detection; we are
interested in particle pollution as indicated by the fine particulate
matter (PM 2.5), particles that are 2.5 micrometers in diameter or
smaller and can only be seen with an electron microscope. These tiny
particles are produced from all types of combustion, including motor
vehicles, power plants, residential wood burning, forest fires,
agricultural burning, and some industrial processes. High PM 2.5
is linked to increased mortality rate for patients suffering from heart
and lung disease~\cite{ZhangC:15}. We aim to provide an efficient and
effective approach for detecting high PM 2.5 regions, which can guide
authorities in designing remedial measures.

The dataset comes from~\cite{AirCondition} and includes 756 operating
sensors that record the daily average at various locations;
Figure~\ref{fig:air_detection_july}(a) shows the PM 2.5 distribution
on July 1st, 2014, in the mainland U.S.  We construct a ten-nearest
neighbor graph with each sensor a node connecting to ten neighboring
sensors. Figure~\ref{fig:air_detection_july}(b) shows the input attribute obtained by thresholding the measurements above 15 (high-pollution cities, marked in yellow).

When high-pollution cities are far from each other, high pollution may
be caused by random events or measurement failures, which makes the
detection sensitive to noise. When high-pollution cities are clustered
together, high pollution is prevalent in the area, which makes the
detection robust. Using detection algorithms here aims to answer
whether high-pollution cities are clustered and provides a more robust
high-pollution detector. In Figure~\ref{fig:air_detection_july}(b),
high-pollution cities seem to concentrate in the mideast part of the
U.S. We now verify whether this claim is true by using the graph
wavelet statistic, graph scan statistic and convex graph scan
statistica.  To make a comparison, we also simulate 1,000 attributes
(graph signals) that have the same number of high-pollution cities,
but are scattered across the
U.S. Figures~\ref{fig:air_detection_july}(e)--(g) show the
values of the graph statistics. For each plot, the red dashed line
shows the value of the graph statistic for the real pollution
attribute as shown in Figure~\ref{fig:air_detection_july}(b) and the
black curves show the empirical histograms of the graph statistics
under $1,000$ random trials.  We see that the values of the graph
statistics of the real attributes are always much larger than those of
the scattered simulated attributes for all three statistics, which
means that it is easy to reject the null hypothesis and confirm that
the high-pollution cities in Figure~\ref{fig:air_detection_july}(b)
form a local cluster. Figures~\ref{fig:air_detection_july}(c)--(d)
show the localized attributes detected by the graph scan statistic and
the convex graph scan statistic, respectively. We see that these two
detected localized attributes are similar and confirm that the
mideast region has relatively high pollution.

We did the same experiments for the data collected on December 1st,
2014, as shown in Figure~\ref{fig:air_detection_dec}(a). The dataset
includes 837 operating sensors (the operating sensors are different
every day) and we still construct a ten-nearest neighbor graph. We see
that high-pollution cities are more scattered, though some of them
seem to cluster in the northwestern
corner. Figures~\ref{fig:air_detection_dec}(e)--(g) show that the
values of the graph statistics of the real attributes are still larger
than those of the scattered simulated attributes for all three
statistics most of the time, which confirms that high-pollution cities
cluster together. Again, the two detected localized attributes in
Figures~\ref{fig:air_detection_dec}(c)-- and (d) confirm that the
northwestern corner has relatively high pollution.

\subsection{Ranking Attributes for Community Detection}
As discussed in Section~\ref{sec:relatedwork}, while relevant node
attributes improve the accuracy of community detection and add meaning
to the detected communities, irrelevant attributes may harm the
accuracy and cause computational inefficiency. By using the proposed
statistics, we can quantify the usefulness of each attribute.  As
localized attributes tend to be related to community structure, our
methods can serve to filter out the most useful attributes for
community detection.

As a dataset, we use the IEEE Xplore database to find working
collaborators~\cite{IEEEX}. We construct three
bipartite networks: papers and journals, papers and authors and papers
and keywords (keywords are automatically assigned by IEEE). We focus
on papers in ten journals: IEEE Transactions on Magnetics, Information
Theory, Nuclear Science, Signal Processing, Electron Devices,
Communications, Applied Superconductivity, Automatic Control,
Microwave Theory and Techniques and Antennas and Propagation.

We project the bipartite network of papers and authors onto authors to
create a co-authorship network where two authors are connected when
they co-author at least four papers. We keep the largest connected
component of the co-authorship network. As explained in more detail
below, we project the network of papers and journals and the network
of papers and authors onto the authors in the largest connected
component to create the author-journal matrix, where rows denote
authors and columns denote journals.  We project the network of papers
and keywords and a network of papers and authors onto the authors in
the largest connected component to create the author-keyword matrix,
where rows denote authors and columns denote keywords.

The entire dataset includes the co-authorship network with 7,330
authors (nodes) and 108,719 co-authorships (edges), the
author-journal matrix with ten journals and the author-keyword matrix
with 3,596 keywords (attributes).
We want to detect academic communities (defined based on the journal)
based on the graph structure and attributes. Since our attributes are
keywords, we use those to improve community detection; for example, in
the signal processing community, some frequently-used keywords are
`filtering', `Fourier transform' and `wavelets'. Our ground truth
communities are the ten journals; when authors publish at least ten
papers in the same journal, we assign them to an eponymous community.

\begin{table*}[htbp]
  \footnotesize
  \begin{center}
    \begin{tabular}{@{}llll@{}}
      \toprule
 Modularity-based ranking &    Cut-based ranking  &  Graph wavelet statistic based ranking &  Local graph scan statistic based ranking \\
      \midrule \addlinespace[1mm]
Dielectrics  &   Data analysis    &  Subtraction techniques  & Maximum likelihood detection \\
Next generation networking  &  Parasitic capacitance  &  Mice   &  Upper bound \\
Undulators   &  Alpha particles  &  Parallel architectures   &  Antenna arrays \\
Pathology & Strips  & Integrated circuit interconnections   &  Biological information theory  \\
Servosystems   & Customer relationship management  &  Yttrium barium copper oxide & Interchannel interference  \\
Electronic design automation &  Biology computing  &  Dielectrics  & Microscopy  \\
Subtraction techniques   &  Uncertainty  & Distributed Bragg reflectors  & Signal analysis  \\
Optical noise  &  Piecewise linear techniques  &  Fuel economy  & Broadcasting  \\
Implants  & Forensics  & Magnetic analysis  & Array signal processing  \\
Biomedical signal processing  &   Oceans   &  Tactile sensors   &  Time-varying systems \\
\bottomrule
\end{tabular} 
\caption{\label{tab:topuniversity}  Top ten most important keywords that potentially form academic communities.}
\vspace{-2mm}
\end{center}
\end{table*}

The goal is to rank all keywords based on their contribution to
community detection, where the value of the corresponding statistic is
used to determine the rank.  We consider four ranking methods: graph
wavelet statistic based ranking, local graph scan statistic based
ranking, modularity-based ranking~\cite{Newman:10} and
cut-based ranking~\cite{Luxburg:07}.  We did not use the convex graph scan statistic due
to computational cost. For the first two ranking methods, we compute
the values of the graph statistics and rank the keywords according to
the values of the graph statistics in a descending order. This is
because a larger statistic means a larger probability that this
keyword forms a community. For the modularity-based ranking, we
compute the modularity of each keyword and rank the keywords according
to modularity in a descending order. For the cut-based ranking, we
compute the cut cost of each keyword and rank the keywords according
to the cut cost in a ascending order. Table~\ref{tab:topuniversity}
lists the top ten most important keywords that potentially form
communities provided by the above four rankings.

To quantify the real community detection power of keywords, we compare
each keyword to the ground-truth community and compute the
correspondence by using the average F1
score~\cite{YangL:13,YangML:13}. Let $C^*$ be a set of the
ground-truth communities and $\widehat{C}$ a set of the activated node
sets provided by the node attributes. Each node set $\widehat{C}_i \in
\widehat{C}$ collects the nodes that have the same attribute. The
average F1 score is
\begin{equation*}
  \frac{1}{2|C^*|} \sum_{C_i \in C^*} \text{F1} (C_i, \widehat{C}_{g(i)}) + \frac{1}{2|\widehat{C}|} \sum_{\widehat{C}_i \in \widehat{C}} \text{F1} (\widehat{C}_{g'(i)}, \widehat{C}_i),
\end{equation*}
where the best matching $g$ and $g'$ are 
\begin{equation*}
  g(i) \ = \  \arg \max_j \text{F1}(C_i, \widehat{C}_j)~\text{and}~g'(i) = \arg \max_j \text{F1}(C_j, \widehat{C}_i), 
\end{equation*}
where $\text{F1}(C_i, \widehat{C}_j)$ is the harmonic mean of
precision and recall. A large average F1 score means that the
community induced by a keyword agrees with the community induced by
journal papers. We also compute the average F1 score of each keyword
and rank the keywords according to the average F1 scores in a
descending order, which is the ground-truth ranking.  We compare the
four estimated rankings with the ground-truth ranking by using the
Spearman's rank correlation coefficient~\cite{MyersWL:10}. The
Spearman correlation coefficient is defined as the Pearson correlation
coefficient between the ranked variables,
\begin{equation*}
  \text{correlation} \ = \  1 - \frac{6 \sum |p_i - q_i|^2}{N(N^2-1)},
\end{equation*}
where $p_i-q_i$ is the difference between two rankings. The Spearman
correlation coefficients of modularity-based ranking, cut-based
ranking, graph wavelet statistic based ranking and local graph scan
statistic based ranking are shown in Figure~\ref{fig:IEEE_rank}. We
see that the graph wavelet statistic and local graph scan statistic
outperform other methods. Cut-based ranking performs poorly because it
may rank infrequent keywords higher. To account for this effect, we
also consider the average modularity---the modularity divided by the
number of activated authors and the average cut cost---the cut cost
divided by the number of activated authors. We see that the average
cut cost performs much better than the total cuts.

\begin{figure}[htb]
  \begin{center}
 \includegraphics[width=0.9\columnwidth]{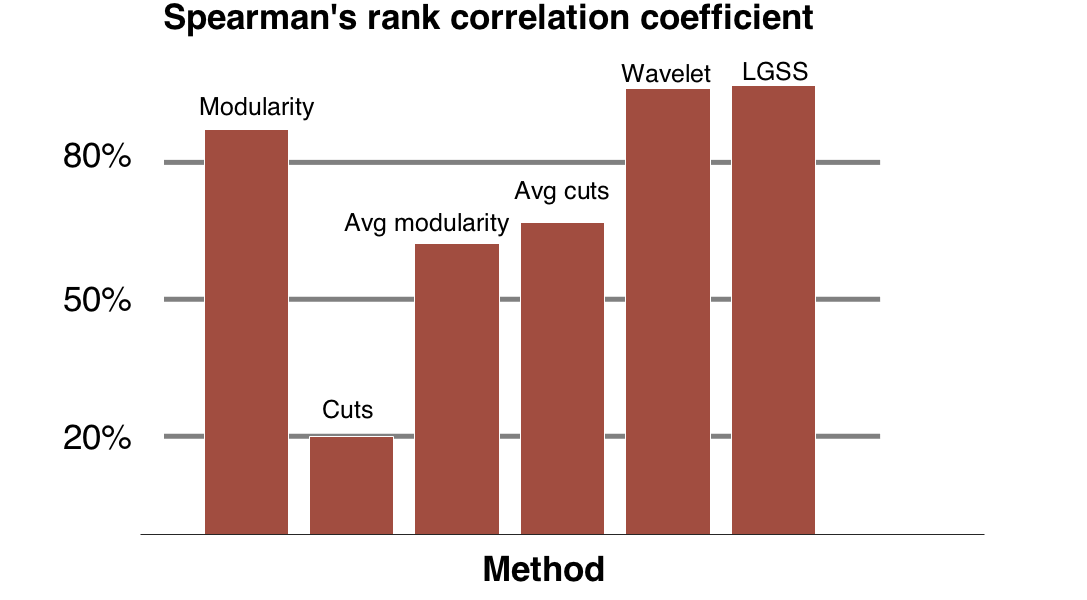} 
	\\
  \end{center}
  \caption{\label{fig:IEEE_rank} Comparison of Spearman's rank
    correlation coefficients; the higher the correlation coefficient,
    the higher correlation to the ground-truth ranking.}
\end{figure}

Figure~\ref{fig:IEEE_aF1} compares the average F1 scores as a function
of individual keywords ordered by the six ranking methods. The
$x$-axis is the ranking provided by the proposed ranking methods and
the $y$-axis is the average F1 score of the corresponding keyword. For
example, since the local graph scan statistic ranks \emph{Maximum
  likelihood detection} first, we put the corresponding average F1
score as the first element on the red curve (leftmost). We expect that
the curve goes down as the rank increases because a good ranking
method ranks the important keywords higher. We also use cluster
affiliation model for big networks (BIGCLAM, shown in black), a
large-scale overlapping community detection algorithm to provide a
baseline~\cite{YangL:13}. We see that the local graph scan statistic
is slightly better than the graph wavelet statistic and both of them
outperform the other methods, which is consistent with the results
given by the Spearman correlation coefficients in
Figure~\ref{fig:IEEE_rank}.  Average cuts rank important keywords
lower, causing the F1 score to increase as the rank decreases. The
average cuts fail because small average cuts may come from just a few
activations. For example, a keyword activating only one author has a
small cut number, although this keyword is actually
trivial. Surprisingly, using high-ranking keywords selected by LGSS
and wavelets works even better than BIGCLAM on the task of community
detection. The reason may be that in this co-authorship network, only
some keywords are informative and strongly related to the journals,
which are the ground-truth communities, while most keywords may
provide trivial or misleading information.

\begin{figure}[htb]
  \begin{center}
    \includegraphics[width=0.75\columnwidth]{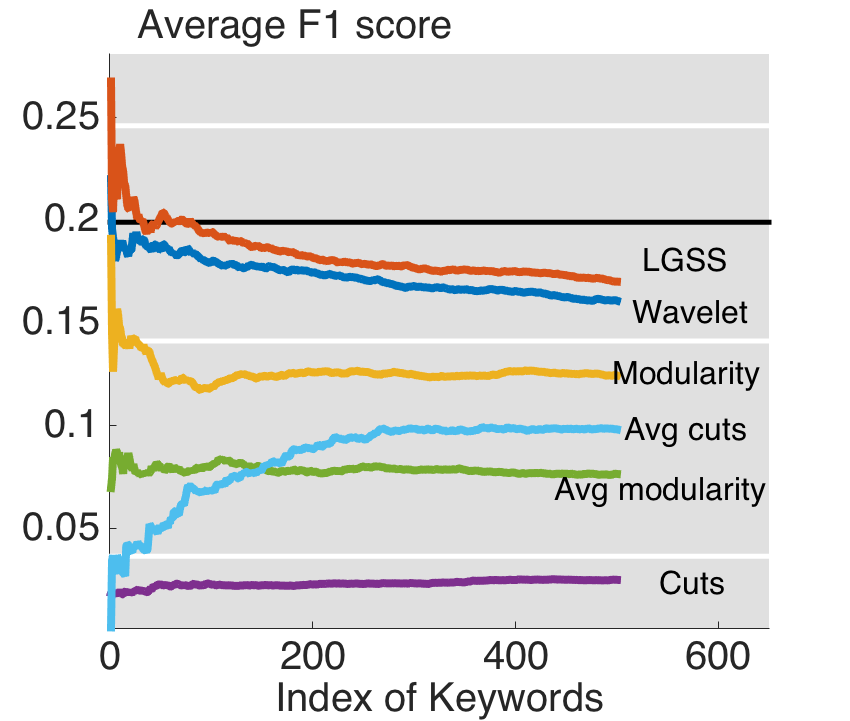} \\
  \end{center}
  \caption{\label{fig:IEEE_aF1} Comparison of the average F1 score as
    a function of the top $k$ ranked keywords; the higher the average
    F1 score, the higher the detection performance by using each
    individual keyword. The black horizontal line shows the
    performance of BIGCLAM, a large-scale overlapping community
    detection algorithm.  }
\end{figure}

\section{Conclusions}
\label{sec:conclusions}
The goal of this paper was to detect localized attributes on a graph
when observations are corrupted by noise. We formulate hypothesis
tests to decide whether the observations activate a community in a
graph corrupted by Bernoulli noise. We model our noisy attributes as 
binary graph signals: a positive signal coefficient indicates an
activated node; when the activated nodes form a cluster, we say that
the attribute contains a localized activation. We proposed two
statistics for testing: graph wavelet statistic and graph scan
statistic, both of which are shown to be efficient and statistically
effective to detect activations. The graph wavelet statistic works by
detecting the boundary of the underlying localized attribute while the
graph scan statistic works by localizing the underlying localized
attribute, which is equivalent to denoising a given
attribute. Theorems~\ref{thm:wavelet_statistic},~\ref{thm:graph_scan_statistic}
and~\ref{thm:relax_scan_statistic} show that the key to distinguishing
the activation is the activation probability difference and the size
of the activated region.  We validate the effectiveness and robustness
of the proposed methods on simulated data first; experimental results
match the theorems well. We further validate them on two real-world
applications: high air-pollution detection and attribute ranking;
experimental results show the proposed statistics are effective and
efficient.

\bibliographystyle{IEEEbib}
\bibliography{bibl_jelena}


\appendix
\section{Appendices}
\subsection{Proof of Theorem~\ref{thm:wavelet_statistic} }
\label{sec:Appendices1}
Under $H_0^N$, the observation is $\y = \text{Bernoulli}( \epsilon
\one_\V )$. Let the $i$th graph wavelet coefficient be
\begin{eqnarray*}
  \w_i^T \y  & = &  \sqrt{ \frac{|S^{(i)}_1| |S^{(i)}_2|}{|S^{(i)}_1|+|S^{(i)}_2|}} \left( \frac{ \one_{S^{(i)}_1}^T \y  }{|S^{(i)}_1|} - \frac{ \one_{S^{(i)}_2}^T \y  }{|S^{(i)}_2|} \right)
  \\
  & = &  \frac{ \sqrt{ |S^{(i)}| } } {2} \left( \frac{ \one_{S^{(i)}_1} }{|S^{(i)}_1|} - \frac{ \one_{S^{(i)}_2} }{|S^{(i)}_2|} \right)^T \y,
\end{eqnarray*}
where $S^{(i)}_1, S^{(i)}_2$ are two local child sets that form $\w_i$
and $S^{(i)}$ is the parent local set. The second equality follows
from the even partition. Since each element in $\y$ is a Bernoulli
random variable that takes value one with probability $\epsilon$, the
first term in the parentheses is the mean of $|S^{(i)}_1|$ Bernoulli
random variables with success probability of $\epsilon$. We then can
show that
\begin{eqnarray*}
  \Ep \left(  { \one_{S^{(i)}_1}^T \y }/{|S^{(i)}_1|} \right) = { |S^{(i)}_1| \epsilon }/{|S^{(i)}_1|} = \epsilon.
\end{eqnarray*}
Following from the Hoeffding inequality (Proposition~2.7~in
\cite{Massart:07}), for any $\eta$, $ \mathbb{P} \left( \left| {
      \one_{S^{(i)}_1}^T \y }/{|S^{(i)}_1|} - \epsilon \right| \geq
  \eta \right) \leq 2 e^{-2|S^{(i)}_1|\eta^2 }.  $ Thus,
$(\one_{S^{(i)}_1}^T \y / |S^{(i)}_1| - \epsilon)$ is $1/ (2
\sqrt{|S^{(i)}_1|})$-sub-Gaussian random variable. Combing two terms,
we obtain that $\w_i^T \y$ is $\sqrt{2}/2$-sub-Gaussian random
variable. Then,
\begin{displaymath}
  \Ep \left(  \left\| \W_{(-1)}^T \y \right\|_\infty  \right) = 
  \Ep \left(  \max_{2 \leq i \leq N}  \w_i^T \y \right) \leq \sqrt{ \log N}.
\end{displaymath}
When var$(Z) = \sigma^2$, the Cirelson-Ibragimov-Sudakov inequality
(Page 10 in~\cite{Massart:07}) ensures $\mathbb{P} \left(Z \geq
  \mathbb{E} Z + u \right) \ \leq \ e^{-u^2/2\sigma^2}. $ Thus, with
probability at least $1-\delta_1$, we obtain
\begin{eqnarray}
  \label{eq:wavelet_H0}
  \widehat{w} \ = \ \left\| \W_{(-1)}^T \y \right\|_\infty  & \leq &  \Ep \left(  \left\| \W_{(-1)}^T \y \right\|_\infty  \right) + \sqrt{ \log (1/\delta_1)}
  \nonumber \\
  & = & \sqrt{\log N} + \sqrt{\log(1/\delta_1)}.
\end{eqnarray}
In other words, when a threshold is $\tau = \sqrt{\log N} +
\sqrt{\log(1/\delta_1)}$, the type-1 error is upper bounded by
$1-\delta_1$. This proves Lemma~\ref{lem:wavelet_statistic}.

Under $H_1^N$, the observation is $\y = \text{Bernoulli}( \mu \one_C +
\epsilon \one_\V )$. We need to show there exists at least one wavelet
basis vector capturing $C$ well. We start with the following lemmas.
\begin{myLem} 
  \label{lem:concentration}
  \begin{displaymath}
    \left\| \W^T_{(-1)} \one_C  \right\|_\infty \  \geq \  \sqrt{ \frac{ |C| (1-|C|/N) }{  1+ \rho \log N } }.
  \end{displaymath}
\end{myLem}
\noindent
\begin{proof}
\begin{eqnarray*}
\left\| \W^T_{(-1)} \one_C  \right\|_\infty^2  & \geq &
  \frac{ \left\|  \W^T_{(-1)}  \one_C \right\|_{2}^2 }{ \left\|  \W^T_{(-1)} \one_C \right\|_{0} } 
\\
& \geq & \frac{ \left\|  \W^T \one_C \right\|_{2}^2 - \left( \frac{1}{\sqrt{N}}\one_{\V}^T \one_C \right)^2 } { 1+ \rho \log N  }
\\
& \geq & \frac{  \left\| \one_C \right\|_{2}^2 - |C|^2 /N }{  1+\rho \log N } \  \geq \  \frac{ |C| (1-|C|/N) }{  1+ \rho \log N }.
\end{eqnarray*}
The second inequality follows from Lemma~\ref{lem:sparse} (Sparsity) in the paper with $\TV_0( \one_C) \leq \rho$ and the even partition.
\end{proof}

 Let
 \begin{eqnarray*}
   i  & = & 
   \arg \max_{2 \leq j \leq N}  \w_j^T \one_C  
   \\
   & = &  \arg \max_{2 \leq j \leq N}  \frac{ \sqrt{S^{(j)}} } {2} \left( \frac{|S^{(j)}_1 \cap C|}{|S^{(j)}_1|} - \frac{|S^{(j)}_2 \cap C|}{|S^{(j)}_2|}  \right),
 \end{eqnarray*}
 where $S^{(j)}_1, S^{(j)}_2$ are the children local sets of
 $S^{(j)}$. Following from Lemma~\ref{lem:concentration},
\begin{eqnarray}
  \label{eq:wavelet_reform}
  \frac{ \sqrt{S^{(i)}} } {2} \left( \frac{|S^{(i)}_1 \cap C|}{|S^{(i)}_1|} - \frac{|S^{(i)}_2 \cap C|}{|S^{(i)}_2|}  \right) \geq   \sqrt{ \frac{ |C| (1-|C|/N)  }{  1+ \rho \log N } }.
\end{eqnarray}

 \begin{myLem} 
   \label{lem:binomial}
   When $a \sim \text{Binomial} \left( n_1, p_1 \right), b \sim
   \text{Binomial} \left( n_2, p_1 \right)$, with probability
   $(1-\delta)^2$,
  \begin{equation*}
    \left| a  + b - (n_1 p_1 + n_2 p_2) \right| \ \leq \ \sqrt{ \frac{n_1}{2} \log (\frac{2}{\delta})} + \sqrt{ \frac{n_2}{2} \log (\frac{2}{\delta})}.
 \end{equation*}
\end{myLem}
\noindent
\begin{proof}
Following from the Hoeffding inequality, with probability $1-\delta$, 
$ | a - n_1 p_1 | \ \leq \ \sqrt{ n_1\log (2/\delta)/2}$. We bound both $a$ and $b$, and rearrange the terms to obtain Lemma~\ref{lem:binomial}.
\end{proof}

 We aim to show that $
|\w_i^T \y| $ is sufficiently large. Here, the term
$\one_{S^{(i)}_1}^T \y$ counts how many ones appear inside the local
set $S^{(i)}_1$ and is a random variable under the distribution of
$\text{Binomial} \left(|S^{(i)}_1 \cap C |, \mu \right) +
\text{Binomial} \left(|S^{(i)}_1 \cap \overline{C} |, \epsilon
\right)$. Following Lemma~\ref{lem:binomial}, with probability
$(1-\delta)^2$,
  \begin{eqnarray*}
    && \left| \one_{S^{(i)}_1}^T \y -  \left( |S^{(i)}_1  \cap C |  \mu + |S^{(i)}_1  \cap \overline{C} | \epsilon | \right) \right| 
    \\
    & \leq & \sqrt{ \frac{|S^{(i)}_1  \cap C | }{2} \log (\frac{2}{\delta})} + \sqrt{ \frac{|S^{(i)}_1  \cap \overline{C} |}{2} \log (\frac{2}{\delta})},
   \end{eqnarray*}
   where the similar formulation also holds for $S^{(i)}_2$ . Without
   losing generality, we assume that $|S^{(i)}_1 \cap C | \geq
   |S^{(i)}_2 \cap C |$. With probability $(1- \delta_2)^4$,
   \begin{eqnarray*}
     \one_{S^{(i)}_1}^T \y  & \geq & |S^{(i)}_1 \cap C | \mu + |S^{(i)}_1 \cap \overline{C} | \epsilon 
     \\
     &&  - \left( \sqrt{ |S^{(i)}_1  \cap C | } +  \sqrt{ |S^{(i)}_1  \cap \overline{C} | }  \right)  \sqrt{ \frac{1}{2} \log (\frac{2}{\delta_2})},
   \end{eqnarray*}
   where the similar formulation also holds for $S^{(i)}_2$. We have
   \begin{eqnarray}
     \label{eq:wavelet_H1}
     &&  | \w_i^T \y |  \ = \   | \frac{1}{\sqrt{|S^{(i)}|} } \left( \one_{S^{(i)}_1}^T \y -  \one_{S^{(i)}_2}^T \y  \right)|
     \nonumber \\  \nonumber
     & \geq &  \frac{ |S^{(i)}_1 \cap C|  - |S^{(i)}_2 \cap C| }{ \sqrt{S^{(i)}}} \left( \mu - \epsilon \right)
     -  \sqrt{ 2 \log (\frac{2}{\delta_2})}
     \\
     & \geq &   \sqrt{ \frac{ |C| (1-|C|/N)  }{  1+ \rho \log N } }  \left( \mu - \epsilon \right) -  \sqrt{ 2 \log (\frac{2}{\delta_2})}.
   \end{eqnarray}
   The last inequality follows from~\eqref{eq:wavelet_reform}. To
   simultaneously bound both type-1 and type-2 errors, the right hand
   side of~\eqref{eq:wavelet_H1} should be larger than the right hand
   side of~\eqref{eq:wavelet_H0}. Thus, we obtain
   \begin{eqnarray*}
     \sqrt{ |C| \left(1- \frac{|C|}{N} \right) }  \left( \mu - \epsilon \right) & \geq & \sqrt{1+ \rho \log N } \Bigg( \sqrt{ \log N } +
     \\
     && \sqrt{ 2 \log(\frac{2}{\delta_1}) } + \sqrt{ 2 \log(\frac{2}{\delta_2})  }  \Bigg),
   \end{eqnarray*}
which proves Theorem~\ref{thm:wavelet_statistic}.

\subsection{Proof of Theorem~\ref{thm:graph_scan_statistic} }
\label{sec:Appendices2}
Under $H_0^N$, the observation is $\y = \text{Bernoulli}( \epsilon
\one_{\V})$. Let $\mathcal{C} = \{ C: \TV_1 (\one_C) \leq \rho
\}$. Similarly to Lemma 4, by using the Hoeffding inequality, we can
show that both $z_1(C) = \one_{C}^T (\y-\epsilon \one) / \sqrt{|C|}$
and $z_2 = \one^T (\y-\epsilon \one) / \sqrt{N}$ are sub-Gaussian
random variables with mean zero and variance $1/2$. The graph scan
statistic is then
\begin{eqnarray*}
  \widehat{g} & = &  \max_{C \in \mathcal{C}}  |C| \KL \left( \frac{ \one_C^T \y} {|C| } \| \frac{ \one^T \y} {N } \right)
  \\
  & \leq & 8 \max_{C \in \mathcal{C}}  |C|  \left( \frac{ \one_C^T (\y - \epsilon \one) } {|C| } - \frac{ \one^T (\y- \epsilon \one)} {N } \right)^2
  \\
  & = &  8 \max_{C \in \mathcal{C}}  \left( z_1(C) -  \sqrt{ \frac{|C|}{N} }  z_2  \right)^2
  \\
  & = &  8  \max \Bigg[ \Bigg(   \max_{C \in \mathcal{C}} \left(  z_1(C) -  \sqrt{ \frac{|C|}{N} }  z_2 \right) \Bigg)^2,
  \\
  && ~~~~~~~~~~
  \Bigg( \min_{C \in \mathcal{C}} \left( z_1(C) -  \sqrt{ \frac{|C|}{N} }  z_2 \right)   
  \Bigg)^2 \Bigg]
\end{eqnarray*}
The random variable $z_1(C) - \sqrt{ {|C|}/{N} } z_2 $ is the
difference between two sub-Gaussian random variables with same
mean. The last equality takes care of both maximum and minimum. Since
the distribution of $z_1(C) - \sqrt{ {|C|}/{N} } z_2 $ is symmetric
and the tail probabilities on two sides are same, we only need to
consider one side of the tail probability of $z_1(C) - \sqrt{
  {|C|}/{N} } z_2 $.

Similarly to Theorem 5 in~\cite{SharpnackKS:13a}, we can show that
with probability $1-\delta_1$,
\begin{eqnarray*}
  \max_{C \in \mathcal{C}} z_1(C)
  & \leq & \left( \sqrt{\rho} + \sqrt{\frac{1}{2} \log N} \right) \sqrt{ 2 \log (N-1)} +
  \\
  &&  \sqrt{2 \log 2} + \sqrt{ 2 \log \frac{1}{\delta_1}}.
\end{eqnarray*}
Since $z_2$ is a sub-Gaussian random variable.  $z_2 >
-\sqrt{\log(2/\delta_1)/2}$ with probability $1-\delta$. Finally, with
probability at least $(1-\delta_1)^2$, we have
\begin{eqnarray*}
  && \max_{C \in \mathcal{C}} \left(  z_1(C) -  \sqrt{ \frac{|C|}{N} }  z_2 \right) \Bigg)  \ \leq \   \sqrt{ \frac{9}{2}\log \frac{2}{\delta_1}} +
  \\
  & &  \left( \sqrt{\rho} + \sqrt{\frac{1}{2} \log N} \right) \sqrt{ 2 \log (N-1)} + \sqrt{2 \log 2}.
\end{eqnarray*}

Theorem 5~\cite{SharpnackKS:13a} is concerned with a Gaussian variable
and here we are concerned with a sub-Gaussian variable. Thus, under
the null hypothesis $H_0^N$, with probability $(1-\delta_1)^2$,
\begin{eqnarray}
  \label{eq:gss_H0}
  \widehat{g}  & \leq &  8 \Bigg( \bigg( \sqrt{\rho} + \sqrt{\frac{1}{2} \log N} \bigg) \sqrt{ 2 \log (N-1)} +
  \\ \nonumber
  &&  \sqrt{2 \log 2} + \sqrt{ \frac{9}{2} \log (2/\delta_1)} \Bigg)^2.
\end{eqnarray}
In other words, when we set the threshold $\tau$ to be the right hand
side of the above formula, the type-1 error is upper bounded by
$(1-\delta_1)^2$. This proves Lemma~\ref{lem:graph_scan_statistic}.

Under $H_1^N$, the observation is $\y = \text{Bernoulli}( \mu \one_{C}
+ \epsilon \one_{\overline{C}} )$. Based on Lemma~\ref{lem:binomial},
with probability $(1-\delta_2)^2$,
\begin{displaymath}
  | \frac{\one^T \y}{N} - \epsilon - (\mu-\epsilon)\frac{|C|}{N} |
  \leq   \sqrt{ \frac{|C|}{2|N|^2} \log (\frac{2}{\delta_2})} + \sqrt{ \frac{N - |C|}{2 N^2} \log (\frac{2}{\delta_2})}.
\end{displaymath}
Following from the Hoeffding inequality, with probability $1-\delta$,
$|\one_C^T \y/{ |C| } -\mu | \leq \sqrt{ \frac{1}{2|C|} \log
  (2/\delta_2) }. $ Then, with probability $(1-\delta_2)^3$,
 \begin{eqnarray}
   \label{eq:gss_H1}
   & \widehat{g} &  =  \max_{C \in \mathcal{C}}  |C| \KL \left( \frac{ \one_C^T \y} {|C| } \| \frac{ \one^T \y} {N} \right)
   \nonumber \\ \nonumber
   & \geq &  \max_{C \in \mathcal{C}} |C|  \left( \frac{ \one_C^T \y} {|C| } - \epsilon - (\mu-\epsilon) \frac{|C|}{N} - \sqrt{\frac{1}{2N} \log(2/\delta_2) } \right)^2
   \\ \nonumber
   & \geq &  |C|  \left( \frac{(N-|C|) \left( \mu - \epsilon \right)}{N}  - (  \sqrt{ \frac{1}{2|C|} } + \sqrt{\frac{1}{2N} } ) \sqrt{ \log(\frac{2}{\delta_2}) } \right)^2.
   \\
\end{eqnarray}
To simultaneously bound both type-1 and type-2 errors, the right hand
side of~\eqref{eq:gss_H1} should be larger than the right hand side
of~\eqref{eq:gss_H0}. Thus, we obtain
\begin{eqnarray*}
  &&  \mu - \epsilon 
  \ \geq  \
  \frac{N}{N -|C|} \Bigg(  \frac{2\sqrt{2}}{\sqrt{|C|}} \Bigg( \sqrt{2 \log 2} + \sqrt{ \frac{9}{2} \log (\frac{2}{\delta_1} ) } +
  \\
  &&   \bigg( \sqrt{\rho} + \sqrt{\frac{1}{2} \log N} \bigg) \sqrt{ 2 \log (N-1)} \Bigg) +  
  \\ 
  &&
  \left(  \sqrt{ \frac{1}{2|C|} } + \sqrt{\frac{1}{2N} } \right) \sqrt{ \log (\frac{2}{\delta_2}) }    \Bigg).
\end{eqnarray*}
We can reformulate the above formula and obtain
Theorem~\ref{thm:graph_scan_statistic}.

\subsection{Proof of Theorem~\ref{thm:relax_scan_statistic} }
\label{sec:Appendices3}
Under $H_0^N$, the observation is $\y = \text{Bernoulli}( \epsilon
\one_{\V})$. Let $\mathcal{X}_t = \{ \x: \x \in [0, 1]^N, \TV_1(\x)
\leq \rho, \one^T \x \leq t. \}$. Let $z_1(\x, t) = \x^T (\y -
\epsilon \one) / t $ and $z_2 = \one^T (\y - \epsilon \one) / N$. The
statistic is then
\begin{eqnarray*}
  \widehat{r}  
  & = & \max_{t, \x \in \mathcal{X}_t}  t \KL \left( \frac{ \x^T \y} {t} \| \frac{ \one^T \y} {N} \right) 
  \\
  & \leq & 8   \max_{ t, \x \in \mathcal{X}_t}  t \left( \frac{ \x^T (\y - \epsilon \one) } {t} -  \frac{ \one^T (\y - \epsilon \one) } {N} \right)^2
  \\
  & = &  8  \max_{ t, \x \in \mathcal{X}_t}  \left( z_1(\x, t) -  \sqrt{ \frac{t}{N} }  z_2  \right)^2
  \\
  & = &  8  \max \Bigg[ \Bigg(    \max_{ t, \x \in \mathcal{X}_t}  \left(  z_1(\x, t) -  \sqrt{ \frac{t}{N} }  z_2 \right) \Bigg)^2,
  \\
  && ~~~~~~~~~~
  \Bigg(  \min_{ t, \x \in \mathcal{X}_t}  \left( z_1( \x, t) -  \sqrt{ \frac{t}{N} }  z_2 \right)   
  \Bigg)^2 \Bigg]
\end{eqnarray*}
The last equality takes care of both maximum and minimum. Since the
distribution of $z_1(\x ,t) - \sqrt{ {t}/{N} } z_2 $ is symmetric, we
only need to consider one side of the distribution.

\begin{myLem}
  \label{lem:subgaussian}
  \begin{eqnarray*}
    \Ep_{\y}   \max_{ t, \x \in [0,1]^N, \one^T \x \leq t }  z_1(\x, t)  \ \leq \  \sqrt{ \frac{N \log 2}{2}}. 
  \end{eqnarray*}
\end{myLem}
\noindent
\begin{proof}
\begin{eqnarray*}
&& \Ep_{\y}   \max_{ t, \x \in [0,1]^N, \one^T \x \leq t }  \frac{ \x^T (\y - \epsilon \one)} { \sqrt{t} } 
\\
& = & \Ep_{\y}  \max_{t} \frac{1}{ \sqrt{t} } \max_{\x \in [0,1]^N, \one^T \x \leq t }   \x^T (\y - \epsilon \one)
\\
& = & \Ep_{\y}  \max_{t} \frac{1}{ \sqrt{t} } \max_{\x \in \{ 0,1 \}^N, \one^T \x \leq t }   \x^T (\y - \epsilon \one)
\\
& = & \Ep_{\y} \max_{\x \in \{ 0,1 \}^N }  \frac{ \x^T (\y - \epsilon \one) }{ \sqrt{ \one^T \x } }.
\end{eqnarray*}
When $\x$ is a binary variable, $\x^T (\y - \epsilon \one) / \sqrt{ \one^T \x } $ is a 1/2-sub-Gaussian random variable as shown previously. The cardinality of the set $ \{ 0,1 \}^N$ is $2^N$. Thus, 
$$
\Ep_{\y}   \max_{ t, \x \in [0,1]^N, \one^T \x \leq t }  \frac{ \x^T (\y - \epsilon \one)} { \sqrt{t} }  \leq \sqrt{\frac{1}{2} \log 2^N}.
$$
\end{proof}

Combining
Lemma~\ref{lem:subgaussian} and to Theorem 5
in~\cite{SharpnackKS:13a}, we can show that with probability
$1-\delta_1$, 
\begin{eqnarray*}
  && \max_{ t, \x \in \mathcal{X}_t}  z_1(\x, t)
  \ \leq \ \frac{ \log 2N +1 }{ \sqrt{ \left( \sqrt{\rho} + \sqrt{\frac{1}{2} \log N} \right)^2 \log N  } }
  \\
  && + 2 \sqrt{ \left( \sqrt{\rho} + \sqrt{\frac{1}{2} \log N}\right)^2 \log N } + \sqrt{2 \log 2} + \sqrt{ 2 \log \frac{1}{\delta_1}}.
\end{eqnarray*}
Since $z_2$ is a sub-Gaussian random variable.  $z_2 >
-\sqrt{\log(2/\delta_1)/2}$ with probability $1-\delta_1$. Finally,
with probability at least $(1-\delta_1)^2$, we have
\vspace{-1mm}
\begin{eqnarray*}
  && \max_{ t, \x \in \mathcal{X}_t}   \left( z_1(\x, t) -  \sqrt{ \frac{t}{N} }  z_2 \right) \Bigg) 
  \\
  &\leq & \frac{ \log 2N +1 }{ \sqrt{ \left( \sqrt{\rho} + \sqrt{\frac{1}{2} \log N} \right)^2 \log N  } } + \sqrt{2 \log 2} 
  \\
  && + 2 \sqrt{ \left( \sqrt{\rho} + \sqrt{\frac{1}{2} \log N}\right)^2 \log N } + \sqrt{ \frac{9}{2} \log \frac{2}{\delta_1}}.
\end{eqnarray*}
Thus, under the null hypothesis $H_0^N$, with probability
$(1-\delta_1)^2$,
\begin{eqnarray}
  \label{eq:rss_H0}
  && \widehat{r}  \ \leq \  8 \Bigg( \frac{ \log 2N +1 }{ \sqrt{ \left( \sqrt{\rho} + \sqrt{\frac{1}{2} \log N} \right)^2 \log N  } }  + \sqrt{2 \log 2} 
  \nonumber  \\  
  && + 2 \sqrt{ \left( \sqrt{\rho} + \sqrt{\frac{1}{2} \log N}\right)^2 \log N }+ \sqrt{ \frac{9}{2} \log (\frac{2}{\delta_1})} \Bigg)^2.
\end{eqnarray}
In other words, when we set the threshold $\tau$ to be the right hand
side of the above formula, the type-1 error is upper bounded by
$(1-\delta_1)^2$. This proves Lemma~\ref{lem:relax_scan_statistic}.

Under $H_1^N$, the observation is $\y = \text{Bernoulli}( \mu \one_{C}
+ \epsilon \one_{\overline{C}} )$. Let $t^* = |C|, \x^* = \one_C$.
Similarly to the proof in Theorem~\ref{thm:graph_scan_statistic}, with
probability at least $(1-\delta_2)^3$,

\begin{eqnarray}
  \label{eq:rss_H1}
  && \widehat{r} \ = \  t^* \KL \left( \frac{ \y^T \x^*} {t^*} \| \frac{\one^T \y}{N} \right) \ = \  |C| \KL \left( \frac{ \one_C^T \y} {|C| } \| \frac{\one^T \y}{N} \right)
  \nonumber  \\ \nonumber
  & \geq &  |C|  \left( \frac{N-|C|}{N} \left( \mu - \epsilon \right) - \left(  \sqrt{ \frac{1}{2|C|} } + \sqrt{\frac{1}{2N} } \right) \sqrt{ \log( \frac{2}{\delta_2} ) } \right)^2.
  \\
\end{eqnarray}

To simultaneously bound both type-1 and type-2 errors, the right hand
side of~\eqref{eq:rss_H1} should be larger than the right hand side
of~\eqref{eq:rss_H0}. Thus, we obtain
\begin{eqnarray*}
  &&  \mu - \epsilon 
  \ \geq  \
  \frac{N}{N -|C|} \Bigg(  \frac{2\sqrt{2}}{\sqrt{|C|}} \Bigg( \frac{ \log 2N +1 }{ \sqrt{ \left( \sqrt{\rho} + \sqrt{\frac{1}{2} \log N} \right)^2 \log N  } }  
  \\
  &&  + \sqrt{2 \log 2}  + 2 \sqrt{ \left( \sqrt{\rho} + \sqrt{\frac{1}{2} \log N}\right)^2 \log N }+ \sqrt{ \frac{9}{2} \log (\frac{2}{\delta_1})}  \Bigg)
  \\
  &&   +  
  \left(  \sqrt{ \frac{1}{2|C|} } + \sqrt{\frac{1}{2N} } \right) \sqrt{ \log (\frac{2}{\delta_2}) }  \Bigg);
\end{eqnarray*}
We can reformulate the above formula and obtain
Theorem~\ref{thm:relax_scan_statistic}.

\subsection{Naive Approach }
\label{sec:Appendices4}
As mentioned in Section~\ref{sec:gss}, a naive statistic is to compute
the average value $\widehat{n} = \one_{\V}^T \y$.
\begin{myThm}
  \label{thm:naive}
  The sufficient condition to distinguish $H_1^N$ from $H_0^N$ by the
  average statistic is
  \begin{displaymath}
    |C| (\mu-\epsilon) \ \geq \ O(\sqrt{N}).
  \end{displaymath}
\end{myThm}
\begin{proof}
  Using Hoeffding inequality, we can show that under 
  $H_0^N$, $ \mathbb{P} \left( \left| { \one^T \y }/{N} - \epsilon
    \right| \geq \eta \right) \leq 2 e^{-2 N \eta^2 }, $ for any
  $\eta$, while under $H_1^N$, $ \mathbb{P} \left( \left| { \one^T \y
      }/{N} - {\left(\mu |C| + \epsilon (N - |C|) \right)}/{N} \right|
    \geq \eta \right) \leq 2 e^{-2 N \eta^2 }, $ for any $\eta$. Thus,
  when $ |C| (\mu-\epsilon) \geq \sqrt{ (N/2) \log({2}/{\delta}
    ) } $ we have that $\mathbb{P}\{ T = 1 | H_0^N \text{ is true} \} \leq
  \delta$ and $\mathbb{P}\{ T = 0 | H_1^N \text{ is true}\} \geq
  1-\delta$.
\end{proof}
We thus see that the sufficient condition in Theorem~\ref{thm:naive}
is much looser than the sufficient conditions in
Theorems~\ref{thm:wavelet_statistic} and
~\ref{thm:graph_scan_statistic}, indicating that the proposed
approaches are provably better than this naive approach.
\end{document}